\documentclass[11pt]{article}
\usepackage{amsmath}
\usepackage{amsthm}
\usepackage{amssymb}
\usepackage{algorithm}
\usepackage{subfig}
\usepackage{color}
\usepackage[english]{babel}
\usepackage{graphicx}
\usepackage{wrapfig,epsfig}
\usepackage{psfrag}
\usepackage{epstopdf}
\usepackage{url}
\usepackage{graphicx}
\usepackage{color}
\usepackage{epstopdf}
\usepackage{algpseudocode}
\usepackage[T1]{fontenc}
\usepackage{verbatim}
\let\C\relax
\usepackage{tikz}
\usetikzlibrary{arrows}
\usepackage[margin=1in]{geometry}
\author{
   Xue Chen\thanks{Supported by NSF Grant CCF-1526952.}\\
  \texttt{xchen@cs.utexas.edu}\\
  The University of Texas at Austin
  \and
  Daniel M. Kane\\
  \texttt{dakane@cs.ucsd.edu}\\
  University of California, San Diego
  \and
  Eric Price \\
  \texttt{ecprice@cs.utexas.edu}\\
  The University of Texas at Austin
  \and
  Zhao Song \\
  \texttt{zhaos@utexas.edu}\\
  The University of Texas at Austin
}

\date{\today}
\title{Fourier-sparse interpolation without a frequency gap}
\newtheorem{theorem}{Theorem}[section]
\newtheorem{lemma}[theorem]{Lemma}
\newtheorem{definition}[theorem]{Definition}

\newtheorem{proposition}[theorem]{Proposition}
\newtheorem{corollary}[theorem]{Corollary}

\newtheorem{observation}[theorem]{Observation}
\newtheorem{fact}[theorem]{Fact}
\newtheorem{remark}[theorem]{Remark}
\newtheorem{claim}[theorem]{Claim}

\newcommand{\abs}[1]{|#1|}

\newcommand{\wh}{\widehat}
\newcommand{\wt}{\widetilde}
\newcommand{\ov}{\overline}
\newcommand{\eps}{\epsilon}
\newcommand{\N}{\mathcal{N}}
\newcommand{\R}{\mathbb{R}}
\newcommand{\volume}{\mathrm{volume}}
\newcommand{\Conpoly}{4}
\newcommand{\Conlog}{3}
\newcommand{\finalC}{5}

\newcommand{\polydelta}{\poly(k,\log(1/\delta))/T}
\newcommand{\RHS}{\mathrm{RHS}}
\newcommand{\LHS}{\mathrm{LHS}}
\renewcommand{\i}{\mathbf{i}}
\newcommand{\norm}[1]{\left\lVert#1\right\rVert}
\newenvironment{proofof}[1]{\bigskip \noindent {\it Proof of #1.}\quad }
{\qed\par\vskip 4mm\par}

\DeclareMathOperator*{\E}{\mathbb{E}}
\DeclareMathOperator*{\var}{\mathrm{Var}}
\DeclareMathOperator*{\Z}{\mathbb{Z}}
\DeclareMathOperator*{\C}{\mathbb{C}}
\DeclareMathOperator*{\median}{median}
\DeclareMathOperator*{\mean}{mean}
\DeclareMathOperator{\supp}{supp}
\DeclareMathOperator{\poly}{poly}

\DeclareMathOperator{\loc}{loc}
\DeclareMathOperator{\repeats}{repeat}
\DeclareMathOperator{\heavy}{heavy}
\DeclareMathOperator{\emp}{emp}
\DeclareMathOperator{\est}{est}

\DeclareMathOperator{\rect}{rect}
\DeclareMathOperator{\sinc}{sinc}
\DeclareMathOperator{\Gram}{Gram}
\DeclareMathOperator{\Gaussian}{Gaussian}
\DeclareMathOperator{\dis}{dis}
\DeclareMathOperator{\Sym}{Sym}
\DeclareMathOperator{\Comb}{Comb}

\makeatletter
\newcommand*{\RN}[1]{\expandafter\@slowromancap\romannumeral #1@}
\makeatother

\usepackage{lineno}

\usepackage{etoolbox}

\makeatletter

\newcommand{\define}[4][ignore]{%
  \ifstrequal{#1}{ignore}{}{
  \@namedef{thmtitle@#2}{#1}}%
  \@namedef{thm@#2}{#4}%
  \@namedef{thmtypen@#2}{lemma}%
  \newtheorem{thmtype@#2}[theorem]{#3}%
  \newtheorem*{thmtypealt@#2}{#3~\ref{#2}}%
}

\newcommand{\state}[1]{%
  \@namedef{curthm}{#1}
  \@ifundefined{thmtitle@#1}{
  \begin{thmtype@#1}
    }{
  \begin{thmtype@#1}[\@nameuse{thmtitle@#1}]
  }
    \label{#1}
    \@nameuse{thm@#1}
  \end{thmtype@#1}
  \@ifundefined{thmdone@#1}{
  \@namedef{thmdone@#1}{stated}%
  }{}
}

\newcommand{\restate}[1]{%
  \@namedef{curthm}{#1}
  \@ifundefined{thmtitle@#1}{
    \begin{thmtypealt@#1}
    }{
  \begin{thmtypealt@#1}[\@nameuse{thmtitle@#1}]
  }
    \@nameuse{thm@#1}
  \end{thmtypealt@#1}
  \@ifundefined{thmdone@#1}{
  \@namedef{thmdone@#1}{stated}%
  }{}
}

\newcommand{\thmlabel}[1]{
  \@ifundefined{thmdone@\@nameuse{curthm}}{\label{#1}
    }{\tag*{\eqref{#1}}}
}
\makeatother

\begin{document}

\begin{titlepage}
  \maketitle
  \begin{abstract}
We consider the problem of estimating a Fourier-sparse signal from
noisy samples, where the sampling is done over some interval $[0, T]$
and the frequencies can be ``off-grid''.  Previous methods for this
problem required the gap between frequencies to be above $1/T$, the
threshold required to robustly identify individual frequencies.  We
show the frequency gap is not necessary to estimate the signal as a
whole: for arbitrary $k$-Fourier-sparse signals under $\ell_2$ bounded
noise, we show how to estimate the signal with a constant factor
growth of the noise and sample complexity polynomial in $k$ and
logarithmic in the bandwidth and signal-to-noise ratio.

As a special case, we get an algorithm to interpolate degree $d$
polynomials from noisy measurements, using $O(d)$ samples and
increasing the noise by a constant factor in $\ell_2$.

  \end{abstract}
  \thispagestyle{empty}
\end{titlepage}


\tableofcontents
\newpage

\section{Introduction}

In an interpolation problem, one can observe $x(t) = x^*(t) + g(t)$,
where $x^*(t)$ is a structured signal and $g(t)$ denotes noise, at
points $t_i$ of one's choice in some interval $[0, T]$.  The goal is
to recover an estimate $\wt{x}$ of $x^*$ (or of $x$).  Because we can
sample over a particular interval, we would like our approximation to
be good on that interval, so for any function $y(t)$ we define
\[
\norm{y}^2_T = \frac{1}{T} \int_0^T \abs{y(t)}^2\mathrm{d}t.
\]
to be the $\ell_2$ error on the sample interval.  For some parameters
$C$ and $\delta$, we would then like to get
\begin{align}\label{eq:interpolate}
  \norm{\wt{x} - x^*}_T &\leq C\norm{g}_T + \delta \norm{x^*}_T
\end{align}
while minimizing the number of samples and running time.  Typically,
we would like $C$ to be $O(1)$ and to have $\delta$ be very small
(either zero, or exponentially small).  Note that, if we do not care
about changing $C$ by $O(1)$, then by the triangle inequality it
doesn't matter whether we want to estimate $x^*$ or $x$ (i.e. we could
replace the LHS of~\eqref{eq:interpolate} by $\norm{\wt{x} - x}_T$).

Of course, to solve an interpolation problem one also needs $x^*$ to
have structure.  One common form of structure is that $x^*$ have a
sparse Fourier representation.  We say that a function $x^*$ is
$k$-Fourier-sparse if it can be expressed as a sum of $k$ complex
exponentials:
\[
x^*(t) = \sum_{j=1}^k v_je^{2\pi \i f_j t}.
\]
for some $v_j \in \C$ and $f_j \in [-F, F]$, where $F$ is the
``bandlimit''.  Given $F$, $T$, and $k$, how many samples must we take
for the interpolation~\eqref{eq:interpolate}?

If we ignore sparsity and just use the bandlimit, then Nyquist
sampling and Shannon-Whittaker interpolation uses $FT + 1/\delta$
samples to achieve~\eqref{eq:interpolate}.  Alternatively, in the
absence of noise, $x^*$ can be found from $O(k)$ samples by a variety
of methods, including Prony's method from 1795 or Reed-Solomon
syndrome decoding~\cite{M69}, but these methods are not robust to
noise.

If the signal is periodic with period $T$---i.e., the frequencies are
multiples of $1/T$---then we can use sparse discrete Fourier transform
methods, which take $O(k \log^c (FT/\delta))$ time and samples
(e.g. \cite{GGIMS,HIKP,IKP}).  If the frequencies are not multiples of
$1/T$ (are ``off the grid''), then the discrete approximation is only
$k/\delta$ sparse, making the interpolation less efficient; and even
this requires that the frequencies be well separated.

A variety of algorithms have been designed to recover off-grid
frequencies directly, but they require the minimum gap among the
frequencies to be above some threshold.  With frequency gap at least
$1/T$, we can achieve a $k^c$ approximation factor using $O(FT)$
samples~\cite{M15}, and with gap above $O(\log^2 k)/T$ we can get a
constant approximation using $O(k \log^c(FT/\delta))$ samples and
time~\cite{PS15}.

Having a dependence on the frequency gap is natural.  If two
frequencies are very close together---significantly below $1/T$---then
the corresponding complex exponentials will be close on $[0, T]$, and
hard to distinguish in the presence of noise.  In fact, from a lower
bound in~\cite{M15}, below $1/T$ frequency gap one cannot recover the
frequencies in the presence of noise as small as $2^{-\Omega(k)}$.
The lower bound proceeds by constructing two signals using
significantly different frequencies that are exponentially close over
$[0, T]$.

But if two signals are so close, do we need to distinguish them?  Such
a lower bound doesn't apply to the interpolation problem, it just says
that you can't solve it by finding the frequencies.  Our question
becomes: can we benefit from Fourier sparsity in a regime where we
can't recover the individual frequencies?

We answer in the affirmative, giving an algorithm for the
interpolation using $O(\poly(k \log(FT/\delta))$ samples.  Our main
theorem is the following:

\define{thm:main}{Theorem}{
  Let $x(t) = x^*(t) + g(t)$, where $x^*$ is $k$-Fourier-sparse signal with
  frequencies in $[-F, F]$.  Given samples of $x$ over $[0, T]$ we can
  output $\wt{x}(t)$ such that with probability at least $1-2^{-\Omega(k)}$,
  \[
  \norm{\wt{x} - x^*}_T \lesssim \norm{g}_T + \delta \norm{x^*}_T.
  \]
  Our algorithm uses $\poly(k,\log (1/\delta) ) \cdot \log(FT)$
  samples and $\poly(k,\log(1/\delta)) \cdot \log^2(FT)$ time.  The output $\wt{x}$
  is $\poly(k,\log(1/\delta))$-Fourier-sparse signal.
}
\state{thm:main}
Relative to previous work, this result avoids the need for a frequency
gap, but loses a polynomial factor in the sample complexity and time.
We lose polynomial factors in a number of places; some of these are
for ease of exposition, but others are challenging to avoid.

Degree $d$ polynomials are the special case of $d$-Fourier-sparse
functions in the limit of $f_j \to 0$, by a Taylor expansion.  This is
a regime with no frequency gap, so previous sparse Fourier results
would not apply but Theorem~\ref{thm:main} shows that
$\poly(d\log(1/\delta))$ samples suffices.  In fact, in this special
case we can get a better polynomial bound:
\define{thm:faster_poly_learning}{Theorem}{%
  For any degree $d$ polynomial $P(t)$ and an arbitrary function $g(t)$, Procedure \textsc{RobustPolynomialLearning}
  in Algorithm~\ref{alg:main_1}
  takes $O(d)$ samples from $x(t)=P(t) + g(t)$
  over $[0,T]$ and reports a degree $d$ polynomial $Q(t)$ in time
  $O(d^\omega)$ such that, with probability at least $99/100$,
\begin{equation*}
\|P(t)-Q(t)\|_T^2 \lesssim \|g(t)\|_T^2.
\end{equation*}
where $\omega<2.373$ is matrix multiplication exponent \cite{S69},\cite{CW87},\cite{W12}.
}
\state{thm:faster_poly_learning}

We also show how to reduce the failure probability to an arbitrary $p
> 0$ with $O(\log(1/p))$ independent repetitions, in Theorem
\ref{thm:accurate_poly_learning}.

Although we have not seen such a result stated in the literature, our
method is quite similar to one used in~\cite{CDL13}.  Since $d$
samples are necessary to interpolate a polynomial without noise, the
result is within constant factors of optimal.

One could apply Theorem~\ref{thm:faster_poly_learning} to approximate
other functions that are well approximated by polynomials or piecewise
polynomials.  For example, a Gaussian of standard deviation at least
$\sigma$ can be approximated by a polynomial of degree
$O(\left(\frac{T}{\sigma}\right)^2 + \log(1/\delta))$; hence the same
bound applies as the sample complexity of improper interpolation of a
positive mixture of Gaussians.

\subsection{Related work}

\paragraph{Sparse discrete Fourier transforms.}
There is a large literature on sparse discrete Fourier transforms.
Results generally are divided into two categories: one category of
results that carefully choose measurements that allow for sublinear
recovery time, including~\cite{GGIMS,GMS,HIKP12,
  Iw13,HIKP,IK,IKP,K16}.  The other category of results expect
randomly chosen measurements and show that a generic recovery
algorithm such as $\ell_1$ minimization will work with high
probability; these results often focus on proving the Restricted
Isometry Property~\cite{CRT06,RV08,B14, H15}.  At the moment, the
first category of results have better theoretical sample complexity
and running time, while results in the second category have better
failure probabilities and empirical performance.  Our result falls in
the first category.  The best results here can achieve $O(k \log n)$
samples~\cite{IK}, $O(k \log^2 n)$ time~\cite{HIKP12}, or within
$\log \log n$ factors of both~\cite{IKP}.

For signals that are not periodic, the discrete Fourier transform will
not be sparse: it takes $k/\delta$ frequencies to capture a $1-\delta$
fraction of the energy.  To get a better dependence on $\delta$, one
has to consider frequencies ``off the grid'', i.e. that are not
multiples of $1/T$.

\paragraph{Off the grid.} Finding the frequencies of a signal with
sparse Fourier transform off the grid has been a question of extensive
study.  The first algorithm was by Prony in 1795, which worked in the
noiseless setting.  This was refined by classical algorithms like
MUSIC~\cite{schmidt81} and ESPRIT~\cite{roy86}, which empirically work
better with noise.  Matrix pencil~\cite{BM86} is a method for
computing the maximum likelihood signal under Gaussian noise and
evenly spaced samples.  The question remained how accurate the maximum
likelihood estimate is;~\cite{M15} showed that it has an $O(k^c)$
approximation factor if the frequency gap is at least $1/T$.

Now, the above results all use $FT$ samples, which is analogous to $n$
in the discrete setting.  This can be decreased down till $O(k)$ by
only looking at a subset of time, i.e. decreasing $T$; but doing so
increases the frequency gap needed for decent robustness results.

A variety of works have studied how to adapt sparse Fourier techniques
from the discrete setting to get sublinear sample complexity; they
all rely on the minimum separation among the frequencies to be at least
$c/T$ for $c \geq 1$.  \cite{TBSR} showed that a convex program can
recover the frequencies exactly in the noiseless setting, for $c \geq
4$.  This was improved in \cite{CF14} to $c \geq 2$ for complex
signals and $c \geq 1.87$ for real signals.  \cite{CF14} also gave a
result for $c \geq 2$ that was stable to noise, but this required the
signal frequencies to be placed on a finely spaced grid.  \cite{YX15}
gave a different convex relaxation that empirically requires smaller
$c$ in the noiseless setting.  \cite{BD13} used model-based compressed
sensing when $c = \Omega(1)$, again without theoretical noise
stability.  Note that, in the noiseless setting, exact recovery can be
achieved without any frequency separation using Prony's method or
Berlekamp-Massey syndrome decoding~\cite{M69}; the benefit of the
above results is that a convex program might be robust to noise, even
if it has not been proven to be so.

In the noisy setting, \cite{FL12} gave an extension of Orthogonal
Matching Pursuit (OMP) that can recover signals when $c = \Omega(k)$,
with an approximation factor $O(k)$, and a few other assumptions.
Similarly,~\cite{BCGLS} gave a method that required $c = \Omega(k)$
and was robust to certain kinds of noise.  \cite{HK15} got the
threshold down to $c = O(1)$, in multiple dimensions, but with
approximation factor $O(FT k^{O(1)})$.

\cite{TBR15} shows that, under Gaussian noise and with separation $c
\geq 4$, a semidefinite program can optimally estimate $x^*(t_i)$ at
evenly spaced sample points $t_i$ from observations $x^*(t_i) +
g(t_i)$.  This is somewhat analogous to our setting, the differences
being that (a) we want to estimate the signal over the entire
interval, not just the sampled points, (b) our noise $g$ is
adversarial, so we cannot hope to reduce it---if $g$ is also
$k$-Fourier-sparse, we cannot distinguish $x^*$ and $g$, and of course
(c) we want to avoid requiring frequency separation.

In \cite{PS15}, we gave the first algorithm with $O(1)$ approximation
factor, finding the frequencies when $c \gtrsim \log(1/\delta)$, and
the signal when $c \gtrsim \log(1/\delta) + \log^2 k$.

Now, all of the above results algorithms are designed to recover the
frequencies; some of the ones in the noisy setting then show that this
yields a good approximation to the overall signal (in the noiseless
setting this is trivial).  Such an approach necessitates $c \geq 1$:
\cite{M15} gave a lower bound, showing that any algorithm finding the
frequencies with approximation factor $2^{o(k)}$ must require $c \geq
1$.

Thus, in the current literature, we go from not knowing how to get any
approximation for $c < 1$, to getting a polynomial approximation at $c
= 1$ and a constant approximation at $c \gtrsim \log^2 k$.  In this work,
we show how to get a constant factor approximation to the signal
regardless of $c$.

\paragraph{Polynomial interpolation.}  Our result is a generalization
of robust polynomial interpolation, and in
Theorem~\ref{thm:faster_poly_learning} we construct an optimal method
for polynomial interpolation as a first step toward interpolating
Fourier-sparse signals.

Our result here can be seen as essentially an extension of a technique
shown in~\cite{CDL13}.  The focus of~\cite{CDL13} is on the setting
where sample points $x_i$ are chosen independently, so $\Theta(d\log
d)$ samples are necessary.  One of their examples, however, shows
essentially the same thing as our
Corollary~\ref{cor:good_approximation}.  From this, getting our
theorem is not difficult.

The recent work~\cite{GZ16} looks at robust polynomial interpolation
in a different noise model, featuring $\ell_\infty$ bounded noise with
some outliers.  In this setting they can get a stronger $\ell_\infty$
guarantee on the output than is possible in our setting.




\paragraph{Nyquist sampling.} The classical method for learning
bandlimited signals uses Nyquist sampling---i.e., samples at rate $1/F$,
for $FT$ points---and interpolates them using Shannon-Nyquist
interpolation.  This doesn't require any frequency gap, but also
doesn't benefit from sparsity like sparse Fourier transform-based
techniques.  As discussed in~\cite{PS15}, on the signal $x(t) = 1$ it
takes $FT + O(1/\delta)$ samples to get $\delta$ error on average.
Our dependence is logarithmic on both those terms.

\subsection{Our techniques}

Previous results on sparse Fourier transforms with robust recovery all
required a frequency gap.  So consider the opposite situation, where
all the frequencies converge to zero and the coefficients are adjusted
to keep the overall energy fixed.  If we take a Taylor expansion of
each complex exponential, then the signal will converge to a degree
$k$ polynomial.  So robust polynomial interpolation is a necessary
subproblem for our algorithm.

\paragraph{Polynomial interpolation.}  Let $P(x)$ be a degree $d$
polynomial, and suppose that we can query $f(x) = P(x) + g(x)$ over
the interval $[-1, 1]$, where $g$ represents adversarial noise.  We
would like to query $f$ at $O(d)$ points and output a degree $d$
polynomial $Q(x)$ such that $\norm{P - Q} \lesssim \norm{g}$, where we
define $\norm{h}^2 := \int_{-1}^1 \abs{h(x)}^2 \mathrm{d}x$.

One way to do this would be to sample points $S \subset [-1, 1]$
uniformly, then output the degree $d$ polynomial $Q$ with the smallest
empirical error
\[
\norm{P + g - Q}^2_S := \frac{1}{\abs{S}}\sum_{x \in S} \abs{(P + g - Q)(x)}^2
\]
on the observed points.  If $\norm{R}_S \approx \norm{R}$ for all
degree $d$ polynomials $R$, in particular for $P - Q$, then since
usually $\norm{g}_S \lesssim \norm{g}$ by Markov's inequality, the result
follows.

This has two problems: first, uniform sampling is poor because
polynomials like Chebyshev polynomials can have most of their energy
within $O(1/d^2)$ of the edges of the interval.  This necessitates
$\Omega(d^2)$ uniform samples before $\norm{R}_S \approx \norm{R}$
with good probability on a single polynomial.  Second, the easiest
method to extend from approximating one polynomial to approximating
all polynomials uses a union bound over a net exponential in $d$,
which would give an $O(d^3)$ bound.

To fix this, we need to bias our sampling toward the edges of the
interval and we need our sampling to not be iid.  We partition $[-1,
1]$ into $O(d)$ intervals $I_1, \dotsc, I_n$ so that the interval
containing each $x$ has width at most $O(\sqrt{1 - x^2})$, except for
the $O(1/d^2)$ size regions at the edges.  For any degree $d$
polynomial $R$ and any choice of $n$ points $x_i \in I_i$, the
appropriately weighted empirical energy is close to $\norm{R}$.  This
takes care of both issues with uniform sampling.  If the points are
chosen uniformly at random from within their intervals, then
$\norm{g}$ is probably bounded as well, and the empirically closest
degree $d$ polynomial $Q$ will satisfy our requirements.

This result is shown in Section~\ref{sec:robustpoly}.

\paragraph{Clusters.}
Many previous sparse Fourier transform algorithms start with a
one-sparse recovery algorithm, then show how to separate frequencies
to get a $k$-sparse algorithm by reducing to the one-sparse case.
Without a frequency gap, we cannot hope to reduce to the one-sparse
case; instead, we reduce to individual clusters of nearby frequencies.

Essentially the problem is that one \emph{cannot} determine all of the
high-energy frequencies of a function $x$ only by sampling it on a
bounded interval, as some of the frequencies might cancel each other
out on this interval. We also cannot afford to work merely with the
frequencies of the truncation of $x$ to the interval $[0,T]$, as the
truncation operation will spread the frequencies of $x$ over too wide
a range. To fix this problem, we must do something in between the
two. In particular, we instead study $x\cdot H$ for a judiciously
chosen function $H$. We want $H$ to approximate the
indicator function of the interval $[0,T]$ and have small
Fourier-support, $\supp(\wh{H}) \subset [-k^c/T, k^c/T]$. By using
some non-trivial lemmas about the growth rate of $x^*$, we can show
that the difference between $x\cdot H$ on $\R$ and the truncation of
$x$ to $[0,T]$ has small $L^2$ mass, so that we can use the former as
a substitute for the latter.

On the other hand, the Fourier transform of $x\cdot H$ is the convolution $\wh{x}*\wh{H}$, which has most of its mass within $\poly(k)/T$ of the frequencies of $x^*$. Although it is impossible to determine the individual frequencies of $x^*$, we can hope to identify $O(k)$ intervals each of length $\poly(k)/T$ so that all but a small fraction of the energy of $\wh{x}$ is contained within these intervals.

Note that many of these intervals will represent not individual frequencies of $x^*$, but small clusters of such frequencies. Furthermore, some frequencies of $x^*$ might not show up in these intervals either because they are too small, or because they cancel out other frequencies when convolved with $\wh{H}$.

\paragraph{One-cluster recovery.} Given our notion of clusters, we
start looking at Fourier-sparse interpolation in the special case of
\emph{one-cluster recovery}.  This is a generalization of one-sparse
recovery where we can have multiple frequencies, but they all lie in
$[f - \Delta, f + \Delta]$ for some base frequency $f$ and bandwidth
$\Delta = k^c/T$.  Because all the frequencies are close to each
other, values $x(a)$ and $x(a + \beta)$ will tend to have ratio close
to $e^{2\pi \i f \beta}$ when $\beta$ is small enough.  We find that
$\beta < \frac{1}{\Delta \sqrt{T \Delta}}$ is sufficient, which lets
us figure out a frequency $\wt{f}$ with $\abs{\wt{f} - f} \leq \Delta
\sqrt{T \Delta} = k^{O(1)}/T$.

Once we have the frequency $\wt{f}$, we can consider $x'(t) = x(t)
e^{-2\pi \i \wt{f}}$.  This signal is $k$-Fourier-sparse with
frequencies bounded by $k^{O(1)}/T$.  By taking a Taylor approximation
to each complex exponential\footnote{There is a catch here, that the
  coefficients of the exponentials are potentially unbounded, if the
  frequencies are arbitrarily close together.  We first use Gram
  determinants to show that the signal is $\delta$-close to one with
  frequency gap $\delta 2^{-k}$, and coefficients at most
  $2^k/\delta$.}, can show $x^*$ is $\delta$-close to $P(t) e^{2\pi \i
  \wt{f}}$ for a degree $d = O(k^c + k \log(1/\delta))$ polynomial $P$.
Thus we could apply our polynomial interpolation algorithm to recover
the signal.

\paragraph{$k$-cluster frequency estimation.} Reminiscent of algorithms
such as~\cite{HIKP,PS15}, we choose random variables $\sigma \approx
T/k^c$, $a \in [0, 1]$, and $b \in [0, 1/\sigma]$ and look at $v \in
\C^{k^c}$ given by
\[
v_i = (x \cdot H)(\sigma(i - a))e^{-2\pi \i \sigma b i} G(i)
\]
where $G$ is a filter function.  That is, $G$ has compact support
($\supp(G) \subset [-k^c, k^c]$), and $\wh{G}$ approximates an
interval of length $\Theta(\frac{2\pi}{k})$. In other words, $G$ is
the same as $\wh{H}$ with different parameters: an interval convolved
with itself $k^c$ times, multiplied by a sinc function.

We alias $v$ down to $O(k)$ dimensions and take the discrete Fourier
transform, getting $\wh{u}$.  It has been implicit in previous
work---and we make it explicit---that $\wh{u}_j$ is equal to
$z_{\sigma a}$ for a vector $z$ defined by
\[
\wh{z} = (\wh{x} * \wh{H}) \cdot \wh{G}^{(j)}_{\sigma, b}
\]
where $\wh{G}^{(j)}_{\sigma, b}$ is a particular permutation of
$\wh{G}$.  In particular, $\wh{G}^{(j)}_{\sigma, b}$ has period
$1/\sigma$, and approximates an interval of size $\frac{1}{\sigma B}$
within each period.

In previous work, when $\sigma$ and $b$ were chosen randomly, each
individual frequency would have a good chance of being the only
frequency preserved in $\wh{z}$, and we could apply one-sparse
recovery by choosing a variety of $a$.  Without a frequency gap we
can't quite say that: we pick $1/\sigma \gg \Delta$ so that the entire
cluster usually lands in the same bin, but then nearby clusters can
also often land in the same bin.  Fortunately, it is still usually
true that only nearby clusters will collide.  Since our $1$-cluster
algorithm works when the signal frequencies are nearby, we apply it to
find a frequency approximation within $\frac{\sqrt{T/\sigma}}{\sigma}
= k^{O(1)}/T$ of the cluster.

The above algorithm recovers each individual frequency with constant
probability.  By repeating it $O(\log k)$ times, with
high probability we find a list $L$ of $O(k)$ frequencies within
$k^{O(1)}/T$ of each significant cluster.

\paragraph{$k$-sparse recovery.}
Because different clusters aren't anywhere close to orthogonal, we
can't simply approximate each cluster separately and add them up.
Instead, given the list $L$ of candidate frequencies, we consider the
$O(kd)$-dimensional space of functions
\[
\wt{x}(t) := \sum_{\wt{f} \in L} \sum_{i=0}^d \alpha_{\wt{f},i} t^i e^{2\pi\i \wt{f} t}
\]
where $d = O(k^{O(1)} + \log(1/\delta))$.  We then take a
bunch of random samples of $x$, and choose the $\wt{x}(t)$ minimizing
the empirical error using linear regression.  This regression can be
made slightly faster using oblivious subspace embeddings
\cite{CW13}, \cite{NN13}, \cite{W14},\cite{CNW15}.

Our argument to show this works is analogous to the naive method we
considered for polynomial recovery.  Similarly to the one-cluster
setting, using Taylor approximations and Gram determinants, we can
show that this space includes a sufficiently close approximation to
$x$.  Since polynomials are the limit of sparse Fourier as frequencies
tend to zero, these functions are arbitrarily close to
$O(kd)$-Fourier-sparse functions.  Hence we know that the maximum of
$\abs{\wt{x}(t)}$ is at most a $\poly(kd)$ factor larger than its
average over $[0, T]$.  Using a net argument, this shows $\poly(kd)$
samples are sufficient to find a good approximation to the nearest
function in our space.

\paragraph{Growth rate of Fourier-sparse signals.}
We need that $\frac{1}{\sqrt{T}}\norm{x^* \cdot H}_2 \approx
\norm{x^*}_T$, where $H$ approximates the interval $1_{[0,
  T]}$. Because $H$ has support size $k^c/T$, it has a transition
region of size $T/k^{c'}$ at the edges, and it decays as
$(t/T)^{-k^{c''}}$ for $t \gg T$.  The difference between
$\frac{1}{\sqrt{T}}\norm{x^* \cdot H}_2$ and $\norm{x^*}_T$ involves
two main components: mass in the transition region that is lost, and
mass outside the sampling interval that is gained.  To show the
approximation, we need that $\abs{x^*(t)} \lesssim \wt{O}(k^2) \norm{x^*}_T$
within the interval and $\abs{x^*(t)} \lesssim (kt/T)^{O(k)}
\norm{x^*}_T$ outside.

We outline the bound of $\underset{t\in [0,T] }{\max} |x^*(t)|$ in terms of its average $\norm{x^*}_T$ to bound $\abs{x^*(t)}$ within the interval. Notice that we can assume $|x^*(0)|=\underset{t\in [0,T] }{\max} |x^*(t)|$: if $t^*= \underset{ t \in [0,T] }{ \arg\max } |x^*(t)|^2$ is not $0$ or $T$, we can rescale the two intervals $[0,t^*]$ and $[t^*,T]$ to $[0,T]$ separately. Then we show that for any $t'$, there exist $m=\wt{O}(k^2)$ and constants $C_1,\cdots,C_m$ such that $x^*(0)=\sum_{j \in [m]} C_j \cdot x^*(j \cdot t')$. Then we take the integration of $t'$ over $[0,T/m]$ to bound $|x^*(0)|^2$ by its average. For any outside $t>T$, we follow this approach to show $x^*(t)=\sum_{j \in [k]} C_j \cdot x^*(t_j)$ where $t_j \in [0,T]$ and $|C_j| \le \poly(k) \cdot (kt/T)^{O(k)}$ for each $j \in [k]$. These results are shown in Section~\ref{sec:technicalsparse}.

\subsection{Organization}
This paper is organized as follows. We provide a brief overview about signal recovery in Section~\ref{sec:proof_sketch}. We introduce some notations and tools in Section~\ref{sec:preli}. Then we show our main Theorem~\ref{thm:faster_poly_learning} about polynomial interpolation in Section~\ref{sec:robustpoly}. For signals with $k$-sparse Fourier transform, we show two bounds on their growth rate in Section~\ref{sec:technicalsparse} and describe the hash functions and filter functions in Section~\ref{sec:hashfilter}. We provide the algorithm for frequency estimation and its proof in Section~\ref{sec:freq_recover}. In Section~\ref{sec:one_cluster_recovery}, we describe the algorithm for one-cluster recovery. In Section~\ref{sub:kmountains}, we show the proof of Theorem~\ref{thm:main}. We defer several technical proofs in Appendix~\ref{sec:proofs}. Appendix~\ref{sec:known_facts} gives a summary of several well-known facts are existing in literature.  We provide the analysis of hash functions and filter functions in Appendix~\ref{sec:proof_hashfilter}.

\section{Proof Sketch}\label{sec:proof_sketch}

We first consider one-cluster recovery centered at zero, i.e., $x^*(t)=  \overset{k}{ \underset{j=1}{\sum}} v_j \cdot e^{2 \pi \i f_j t}$ where every $f_j$ is in $[ - \Delta,  \Delta]$ for some small $\Delta>0$. The road map is to replace $x^*$ by a low degree polynomial $P$ such that $\|x^*(t)-P(t)\|^2_T \lesssim \delta \|x^*\|^2_T$ then recover a polynomial $Q$ to approximate $P$ through the observation $x(t)=P(t) + g'(t)$ where $g'(t)=g(t) + \big(x^*(t)- P(t)\big)$.

A natural way to replace $x^*(t)=\overset{k}{ \underset{j=1}{\sum}} v_j e^{2 \pi \i f_j t}$ by a low degree polynomial $P(t)$ is the Taylor expansion. To bound the error after taking the low degree terms in the expansion by $\delta \|x^*\|_T$, we show the existence of $x'(t)=\overset{k}{ \underset{j=1}{\sum}} v'_j e^{2 \pi \i f'_j t}$ approximating $x^*$ on $[0,T]$ with an extra property---any coefficient $v'_j$ in $x'(t)$ has an upper bound in terms of $\|x'\|^2_T=\frac{1}{T} \int_{0}^T |x'(t)|^2 \mathrm{d} t$. We prove the existence of $x'(t)$ via two more steps, both of which rely on the estimation of some Gram matrix constituted by these $k$ signals.

The first step is to show the existence of a $k$-Fourier-sparse signal $x'(t)$ with frequency gap $\eta \ge \frac{\exp\left(-\poly(k)\right) \cdot \delta}{T}$ that is sufficiently close to $x^*(t)$.
\define{lem:existence_gap}{Lemma}{
There is a universal constant $C_1>0$ such that, for any $x^*(t) = \overset{k}{ \underset{j=1}{\sum}} v_j e^{2\pi \i f_j t}$ and any $\delta>0$ , there always exist $\eta \ge \frac{\delta}{T} \cdot k^{-C_1 k^{2}}$ and $x'(t) = \overset{k}{ \underset{j=1}{\sum} } v'_j e^{2\pi \i f'_j t}$  satisfying
\begin{equation*}
\|x'(t)-x^*(t)\|_T \leq \delta \|x^*(t)\|_T
\end{equation*}
with $ \underset{i\neq j}{\min}  | f'_{i} - f'_j| \geq \eta$ and $ \underset{j\in [k]}{ \max} \{|f'_j - f_j|\} \leq k\eta$.
}
\state{lem:existence_gap}
We outline our approach and defer the proof to Section \ref{sec:one_cluster_recovery}. We focus on the replacement of one frequency $f_k$ in $x^*=\sum_{j \in [k]} v_j e^{2\pi \i f_j t}$ by a new frequency $f_{k+1} \neq f_k$ and its error. The idea is to consider every signal $e^{2\pi \i f_j t}$ as a vector and prove that for any vector $x^*$ in the linear subspace $\mathrm{span}\{e^{2\pi \i f_j t}| j \in [k]\}$, there exists a vector in the linear subspace $\mathrm{span}\{e^{2 \pi \i f_{k+1} t}, e^{2\pi \i f_j t}| j \in [k-1]\}$ with distance at most $\exp(k^2) \cdot \left(|f_k-f_{k+1}|T\right) \cdot \|x^*\|_T$ to $x^*$.

The second step is to lower bound $\|x'\|^2_T$ by its coefficients through the frequency gap $\eta$ in $x'$.
\define{lem:relation_energy_coef}{Lemma}{%
There exists a universal constant $c>0$ such that for any  $x(t) = \overset{k}{ \underset{j=1}{ \sum } } v_j e^{2\pi \i f_j t}$ with frequency gap $\eta=\underset{i\neq j}{\min}| f_i - f_j|$,
\begin{equation*}
 \|x(t)\|^2_T \ge k^{-c k^2} \min\left((\eta T)^{2k},1\right) \sum_{j=1}^k |v_j|^2.
\end{equation*}
}
\state{lem:relation_energy_coef}
Combining Lemma \ref{lem:existence_gap} and Lemma \ref{lem:relation_energy_coef}, we bound $|v'_j|$ by $\exp(\poly(k)) \cdot \delta^{-O(k)} \cdot \|x'\|_T$ for any coefficient $v'_j$ in $x'$. Now we apply the Taylor expansion on $x'(t)$ and keep the first $d=O(\Delta T+\poly(k)+ k \log \frac{1}{\delta})$ terms of every signal $v'_j \cdot e^{2 \pi \i f'_j t}$ in the expansion to obtain a polynomial $P(t)$ of degree at most $d$. To bound the distance between $P(t)$ and $x'(t)$, we observe that the error of every point $t \in [0,T]$ is at most $(\frac{2 \pi \Delta \cdot T}{d})^d \sum_{j} |v'_j|$, which can be upper bounded by $\delta \|x'(t)\|_T$ via the above connection. We summarize all discussion above as follows.
\define{lem:low_degree_approximates_concentrated_freq}{Lemma}{
For any $\Delta>0$ and any $\delta>0$, let $x^*(t)=\sum_{j \in [k]} v_j e^{2 \pi \i f_j t}$ where $|f_j| \le \Delta$ for each $j\in [k]$. There exists a polynomial $P(t)$ of degree at most
\[ d=O(T \Delta + k^3 \log k + k \log 1/\delta) \] such that
\[ \|P(t) - x^*(t)\|^2_T \le \delta \|x^*\|^2_T.\]
}
\state{lem:low_degree_approximates_concentrated_freq}
To recover $x^*(t)$, we observe $x(t)$ as a degree $d$ polynomial $P(t)$ with noise. We use properties of the Legendre polynomials to design a method of random sampling such that we only need $O(d)$ random samples to find a polynomial $Q(t)$ approximating $P(t)$.

\restate{thm:faster_poly_learning}

We can either report the polynomial $Q(t)$ or transfer $Q(t)$ to a signal with $d$-sparse Fourier transform. We defer the technical proofs and the formal statements to Section \ref{sec:one_cluster_recovery} and discuss the recovery of $k$ clusters  from now on.


As mentioned before, we apply the filter function $(H(t),\wh{H}(f))$ on $x^*$ such that $\wh{x^* \cdot H}$ has at most $k$ clusters given $\wh{x^*}$ with $k$-sparse Fourier transform. First, we show that all frequencies in the ``heavy'' clusters of $\widehat{x^* \cdot H}$ constitute a good approximation of $x^*$ in Section \ref{sub:kmountains}.
\define{def:heavy_clusters}{Definition}{
Given $x^*(t)= \overset{k}{\underset{j=1}{\sum} } v_j e^{2 \pi \i f_j t}$, any $\N>0$, and a filter function $(H,\widehat{H})$ with bounded support in frequency domain. Let  
$L_j$ denote the interval of $ ~\supp(\wh{e^{2 \pi \i f_j t} \cdot H})$ for each $j\in [k]$.

Define an equivalence relation $\sim$ on the frequencies $f_i$ by the transitive closure of the relation $f_i\sim f_j$ if $L_i\cap L_j \neq \emptyset$. Let $S_1,\ldots,S_n$ be the equivalence classes under this relation. 

Define $C_i = \underset{f\in S_i}{ \cup } L_i $ for each $i\in [n]$. We say $C_i$ is a ``heavy'' cluster iff $\int_{C_i} |\wh{H \cdot x^*}(f)|^2 \mathrm{d} f \ge T \cdot \N^2/k$.
}
\state{def:heavy_clusters}
\define{cla:guarantee_removing_x**_x*}{Claim}{
Given $x^*(t)= \overset{k}{ \underset{j=1}{\sum} } v_j e^{2 \pi \i f_j t}$ and any $\N>0$, let $H$ be the filter function defined in  Appendix \ref{sec:properties_of_H} and $C_1,\cdots,C_l$ be the heavy clusters from Definition \ref{def:heavy_clusters}. For
\begin{equation*}
S=\left\{j \in [k]\bigg{|}f_j \in C_1 \cup \cdots C_l \right\},
\end{equation*}
 we have $x^{(S)}(t)= \underset{j\in S}{\sum} v_j e^{2 \pi \i f_j t}$ approximating $x^*$ within distance $\|x^{(S)}(t) - x^*(t) \|_T^2 \lesssim \N^2.$
}
\state{cla:guarantee_removing_x**_x*}
Hence it is enough to recover $x^{(S)}$ for the recovery of $x^{*}$. Let $\Delta_h$ denote the bandwidth of $\wh{H}$.   In Section \ref{sec:freq_recover}, we choose $\Delta>k \cdot \Delta_h$  such that for any $j \in S, \int_{f_j-\Delta}^{f_j + \Delta} |\wh{H \cdot x^*}(f)|^2 \mathrm{d} f \ge T \cdot \N^2/k$ from the fact $|C_i| \le k \cdot \Delta_h$. Then we prove Theorem \ref{thm:frequency_recovery_k_cluster} in Section \ref{sec:freq_recover}, which finds $O(k)$ frequencies to cover all heavy clusters of $\widehat{x^* \cdot H}$.
\define{thm:frequency_recovery_k_cluster}{Theorem}{
 Let $x^*(t) = \overset{k}{ \underset{j=1}{\sum} } v_j e^{2\pi\i f_j
    t}$ and $x(t)= x^*(t) +g(t)$ be our observable signal where $\|g(t) \|_T^2 \le
  c\|x^*(t)\|_T^2$ for a sufficiently small constant $c$. Then
  Procedure \textsc{FrequencyRecoveryKCluster} 
  returns a set $L$ of $O(k)$
  frequencies that covers all heavy clusters of $x^*$, which uses $\poly(k, \log(1/\delta) ) \log(FT)$ samples and $\poly(k, \log(1/\delta)) \log^2 (FT)$ time.
  In particular, for $\Delta=\polydelta$ and $\N^2 :
  = \| g(t) \|_T^2 + \delta \| x^*(t) \|_T^2$, with probability $1- 2^{-\Omega(k)}$, for any $f^*$ with
  \begin{equation}
    \int_{f^*-\Delta}^{f^*+\Delta} | \widehat{x\cdot H}(f) |^2 \mathrm{d} f \geq T\N^2/k,
  \end{equation}
there  exists an $\wt{f} \in L$ satisfying
  \begin{equation*}
  |f^*-\widetilde{f} |
  \lesssim \Delta \sqrt{\Delta T}.
  \end{equation*}
}
\state{thm:frequency_recovery_k_cluster}


Let $L=\{\widetilde{f}_1,\cdots,\widetilde{f}_l\}$ be the list of frequencies from the output of Procedure \textsc{FrequencyRecoveryKCluster} in Theorem \ref{thm:frequency_recovery_k_cluster}. The guarantee is that, for any $f_j$ in $x^{(S)}$, there exists some $p_j \in [l]$ such that $|\widetilde{f}_{p_j} - f_j| \lesssim \Delta \sqrt{\Delta T}$ for $\Delta=\polydelta$. Hence we rewrite $x^{(S)}(t)=\sum_{i \in [l]} e^{2 \pi \i \widetilde{f}_{i} t} (\sum_{j \in S: p_j=i} e^{2\pi \i (f_j - \widetilde{f}_{i})t})$. For each $i \in [l]$, we apply Lemma \ref{lem:low_degree_approximates_concentrated_freq} of one-cluster recovery on $\sum_{j \in S: p_j=i} e^{2\pi \i (f_j - \widetilde{f}_{i})t}$ to approximate it by a degree~$d$ polynomial $P_i(t)$.

Now we consider $x(t) = \sum_{i \in [l]}e^{2\pi \i \widetilde{f}_{i} t} \cdot P_i(t) + g''(t)$ where $\|g''(t)\|_T \lesssim \|g(t)\|_T + \delta \|x^*(t)\|_T$. To recover $\sum_{i \in [l]}e^{2\pi \i \widetilde{f}_{i} t} \cdot P_i(t)$, we treat it as a vector in the linear subspace \[V=\mathrm{span}\left\{e^{2\pi \i \widetilde{f}_{i} t} \cdot t^j\bigg{|} j \in \{0,\cdots,d\},i\in [l] \right\}\] with dimension at most $l(d+1)$ and find a vector in this linear subspace approximating it.

We show that for any $v \in V$, the average of $\poly(kd)$ random samples on $v$ is enough to estimate $\|v\|_T^2$. In particular, any vector in this linear subspace satisfies that the maximum of it in $[0,T]$ has an upper bound in terms of its average in $[0,T]$. Then we apply the Chernoff bound to prove that $\poly(kd)$ random samples are enough for the estimation of one vector $v \in V$.
\define{cla:max_bounded_Q}{Claim}{
For any $\vec{u} \in \mathrm{span}\left\{e^{2\pi \i \widetilde{f}_{i} t} \cdot t^j\bigg{|} j \in \{0,\cdots,d\},i\in [l] \right\}$, there exists some universal constants $C_1 \le \Conpoly$ and $C_2 \le \Conlog$ such that
\begin{equation*}
\max_{t \in [0,T]}\{|\vec{u}(t)|^2\} \lesssim (l d)^{C_1} \log^{C_2} (l d) \cdot \|\vec{u}\|^2_T
\end{equation*}
}
\state{cla:max_bounded_Q}
At last we use an $\epsilon$-net to argue that $\poly(kd)$ random samples from $[0,T]$ are enough to interpolate $x(t)$ by a vector $v\in V$. Because the dimension of this linear subspace is at most $l(d+1)=O(kd)$, there exists an $\epsilon$-net in this linear subspace for unit vectors with size at most $\exp(kd)$. Combining the Chernoff bound on all vectors in the $\epsilon$-net and Claim \ref{cla:max_bounded_Q}, we know that $\poly(kd)$ samples are sufficient to estimate $\|v\|_T^2$ for any vector $v \in V$. In Section \ref{sub:kmountains}, we show that a vector $v \in V$ minimizing the distance on $\poly(kd)$ random samples is a good approximation for $\sum_{i \in [l]}e^{2\pi \i \widetilde{f}_{i} t} \cdot P_i(t)$, which is a good approximation for $x^*(t)$ from all discussion above.
%

\restate{thm:main}







%
\section{Preliminaries}\label{sec:preli}
We first provide some notations in Section \ref{sub:notations} and basic Fourier facts in Section \ref{sub:fourier_facts}. Then we review some probability inequalities in Section \ref{sec:prob_tools}. 
At last, we introduce Legendre polynomials in Section \ref{sec:legendre} and review some basic properties of Gram matrix and its determinant in Section \ref{sec:gram_determinant}.

\subsection{Notation}\label{sub:notations}
For any function $f$, we define $\widetilde{O}(f)$ to be $f \cdot \log^{O(1)} (f)$. We use $[n]$ to denote $\{1,2,\cdots, n\}$. Let $\i $ denote $\sqrt{-1}$. For any Complex number $z = a+ \i b \in C$, where $a,b\in \mathbb{R}$. We define $\overline{z}$ to be $a- \i b$ and $|z| = \sqrt{a^2+b^2}$ such that $|z|^2 = z \overline{z}$. For any function $f(t) :\mathbb{R} \rightarrow \mathbb{C}$, we use $\supp(f)$ to denote the support of $f$. 

For convenience, we define the sinc function and the Gaussian distribution $\mathrm{Gaussian}_{\mu,\sigma}$ on $\mathbb{R}$ with expectation $\mu$ and variance $\sigma^2$ as follows:
\[\sinc(t) = \frac{\sin(\pi t)}{\pi t}, \quad \mathrm{Gaussian}_{\mu,\sigma}(t)=\frac{1}{\sigma \sqrt{2\pi}} e^{- \frac{ (t-\mu)^2}{2\sigma^2}}. \]

For a fixed $T>0$, we define the inner product of two functions $x,y:[0,T] \rightarrow \mathbb{C}$ as 
\begin{equation*}
\langle x,y \rangle_T =\frac{1}{T} \int_0^T x(t) \overline{y}(t) \mathrm{d} t.  
\end{equation*}

We define the $\|\cdot \|_T$ norm as
\begin{equation*}
\|x(t) \|_T = \sqrt{\langle x(t) ,x(t)\rangle_T} =\sqrt{\frac{1}{T} \int_0^T |x(t)|^2 \mathrm{d} t}.
\end{equation*}

\subsection{Facts about the Fourier transform}\label{sub:fourier_facts}
In this work, we always use $x(t)$ to denote a signal from $\mathbb{R} \rightarrow \mathbb{C}$. The Fourier transform $\widehat{x}(f)$ of an integrable function $x : \mathbb{R} \rightarrow \mathbb{C}$ is defined as
\begin{equation*}
\widehat{x}(f) = \int_{-\infty}^{+\infty} x(t) e^{-2\pi \i f t} \mathrm{d} t, \text{~for~any~real~number~}f.
\end{equation*}
Similarly, $x(t)$ is determined from $\widehat{x}(f)$ by the inverse transform:
\begin{equation*}
x(t) = \int_{-\infty}^{+\infty} \widehat{x}(f) e^{2\pi \i f t} \mathrm{d} f, \text{~for~any~real~number~}t.
\end{equation*}

Let $\mathrm{CFT}$ denote the continuous Fourier transform, $\mathrm{DTFT}$ denote the discrete-time Fourier transform, $\mathrm{DFT}$ denote the discrete Fourier transform, and $\mathrm{FFT}$ denote the fast Fourier transform.

\begin{figure}[t]
  \centering
    \includegraphics[width=0.8\textwidth]{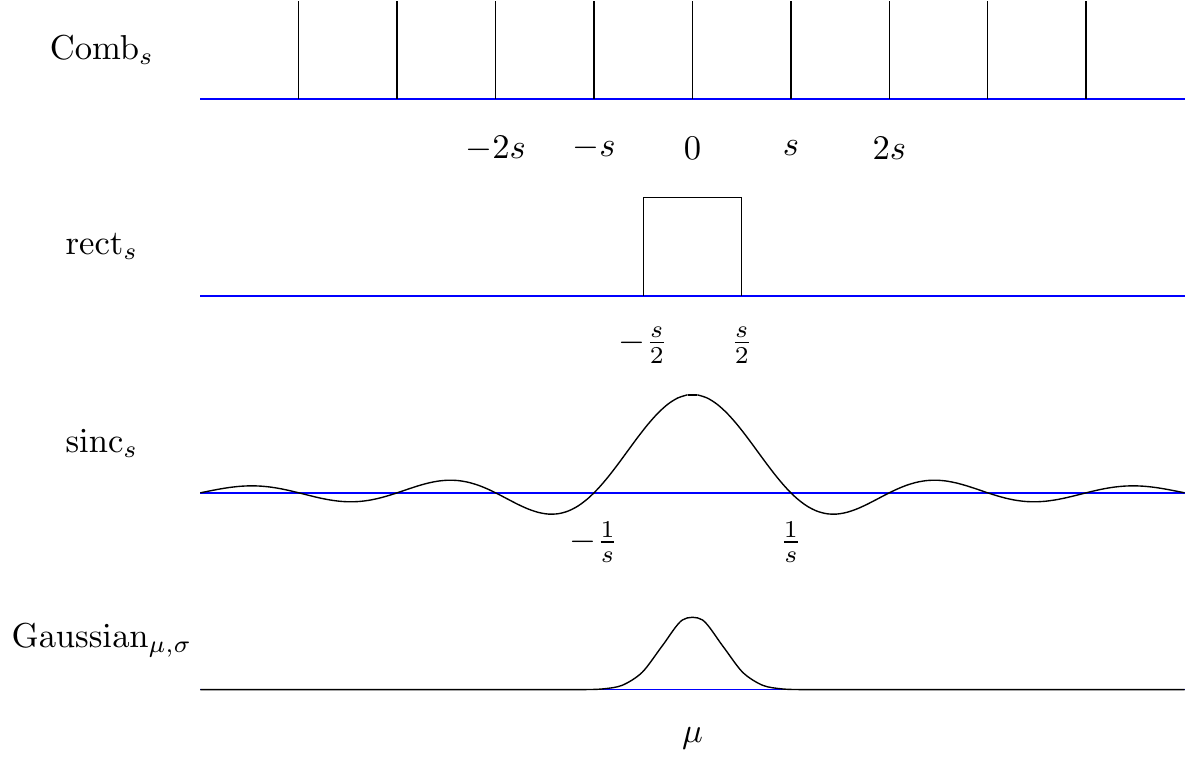}
    \caption{A picture of a $\Comb_s$, $\rect_s$, $\sinc_s$, $\Gaussian_{\mu,\sigma}$ .}\label{fig:example_comb_rect_sinc_gaussian}
\end{figure}

For any signal $x(t)$ and $n \in \mathbb{N}_+$, we define $x^{* n}(t) = \underbrace{ x(t) * \cdots * x(t) }_{n}$ and $\widehat{x}^{\cdot n}(f) = \underbrace{ \widehat{x}(f) \cdot \cdots \cdot \widehat{x}(f)  }_{n}$.
\begin{fact}
Let $\delta_{\Delta}(f)$ denote the Dirac delta at $\Delta$. Then 
\[\widehat{\delta}_{\Delta} (t) = \int_{-\infty}^{+\infty} \delta_{\Delta} (f) e^{2\pi \i f t} \mathrm{d} f = e^{2\pi\i t \Delta}.\]
\end{fact}

\begin{fact}
For any $s>0$, let $\Comb_s(t) = \underset{j\in \Z}{\sum} \delta_{js}(t) $. Then the Fourier transform of $\Comb_s(t)$  is
\begin{equation*}
\widehat{\Comb}_s(f) = \frac{1}{s} \Comb_{1/s}(f).
\end{equation*}
\end{fact}

The following fact says the the Fourier transform of a rectangle function is a sinc function.
\begin{fact}\label{fac:def_of_rect_and_sinc}
We use
\begin{equation*}
\rect_s(t) =
  \begin{cases}
    1  & \quad \rm{~if~} |t|\leq \frac{s}{2},\\
    0  & \quad \rm{~otherwise.}\\
  \end{cases}
\end{equation*}
Then the Fourier transform of $\rect_s(t)$ is $\widehat{\rect_s}(f) = \frac{\sin(\pi f s) }{\pi f s} = \sinc(fs)$.
\end{fact}

The Fourier transform of a Gaussian function is another Gaussian function.
\begin{fact}
 For $\Gaussian_{\mu,\sigma}(t) = \frac{1}{\sigma \sqrt{2\pi} }e^{-\frac{(t-\mu)^2}{2\sigma^2}}$. Then the Fourier transform is
 \begin{equation*}
\widehat{\Gaussian}_{\mu,\sigma}(f) = e^{-2\pi \i fu} \frac{1}{\sigma\sqrt{2\pi}} \Gaussian_{0,\sigma'}(f) \textit{ for } \sigma' = 1/(2\pi \sigma).
 \end{equation*}
\end{fact}

\begin{proof}
From the definition of the Fourier transform,
\begin{align*}
\widehat{\Gaussian}_{\mu,\sigma}(f) = \quad & \int_{-\infty}^{+\infty} \frac{1}{\sigma \sqrt{2\pi}} e^{-\frac{(t-\mu)^2}{2\sigma^2}} e^{-2\pi \i ft} \mathrm{d} t\\
= \quad & e^{-2\pi\i fu} \int_{-\infty}^{+\infty} \frac{1}{\sigma \sqrt{2\pi}} e^{-\frac{t^2}{2\sigma^2}} e^{-2\pi \i ft} \mathrm{d} t\\
= \quad & e^{-2\pi\i fu} \int_{-\infty}^{+\infty} \frac{1}{\sigma \sqrt{2\pi}} e^{-\frac{(t+2 \pi \i \sigma^2 f)^2}{2\sigma^2} - 2 \pi^2 f^2 \sigma^2} \mathrm{d} t\\
=\quad & e^{-2\pi \i fu} e^{-\frac{f^2}{2\sigma'^2}}
\end{align*}
where {~$\sigma'= 1/(2\sigma \pi)$}, which is $e^{-2\pi \i fu} \cdot \sigma' \sqrt{2\pi} \cdot \Gaussian_{0,\sigma'}(f)$.
\end{proof}

\subsection{Tools and inequalities}\label{sec:prob_tools}

From the Chernoff Bound (Lemma \ref{lem:chernoff_bound}), we show that if the maximum of a signal is bounded by $d$ times its energy over some fixed interval, then taking more than $d$ samples (each sample is drawn i.i.d. over that interval) suffices to approximate the energy of the signal on the interval with high probability.

\define{lem:max_is_bounded_imply_sample_is_small}{Lemma}{
Given any function $x(t) : \mathbb{R} \rightarrow \mathbb{C}$ with $\underset{t\in [0,T] }{\max} | x(t)|^2 \leq d \| x(t) \|_T^2$. Let $S$ denote a set of points from $0$ to $T$. If each point of $S$ is chosen uniformly at random from $[0,T]$, we have
\begin{eqnarray*}
\mathsf{Pr} \left[ \left| \frac{1}{|S|} \sum_{i\in S} |x(t_i)|^2- \| x(t)\|_T^2 ]  \right|  \geq \epsilon \| x(t)\|_T^2 \right] & \leq & e^{-\Omega (\epsilon^2 |S|/ d )}
\end{eqnarray*}
}
\state{lem:max_is_bounded_imply_sample_is_small}
We provide a proof in Appendix \ref{sec:proo_max_is_bounded_imply_sample_is_small}.

Because $d \cdot \frac{1}{2d} + \frac{1}{2} \cdot (1- \frac{1}{2d}) \le 1$, we have the following inequality when the maximum of $|x(t)|^2$ is at most $d$ times its average.  
\begin{lemma}\label{lem:max_is_bounded_imply_sample_is_good}
Given any function $x(t) : \mathbb{R} \rightarrow \mathbb{C}$ with $\underset{t\in [0,T] }{\max} | x(t)|^2 \leq d \| x(t) \|_T^2$. Let $S$ denote a set of points from $0$ to $T$. For any point $a$ is sampled uniformly at random from $[0,T]$, we have,
\begin{equation*}
\underset{a\sim [0,T]}{ \mathsf{Pr}} \left[ |x(a)|^2 \geq \frac{1}{2} \| x(t) \|_T^2 \right]  \geq \frac{1}{2d}.
\end{equation*}
\end{lemma}

\subsection{Legendre polynomials}\label{sec:legendre}
We provide an brief introduction to Legendre polynomials (please see \cite{Dunster} for a complete introduction). For convenience, we fix $\|f(t) \|^2_T=\frac{1}{2} \int_{-1}^1 |f(t)|^2 \mathrm{d} t$ in this section. 
\begin{definition}
Let $L_n(x)$ denote the Legendre polynomials of degree $n$, the solution to Legendre's differential equation:
\begin{equation}\label{eq:Legendre_differential}
\frac{\mathrm{d} }{ \mathrm{d} x} \left[ (1-x^2) \frac{\mathrm{d} }{ \mathrm{d} x} L_n(x) \right] + n(n+1) L_n(x) = 0
\end{equation}
\end{definition}
We will the following two facts about the Legendre polynomials in this work.
\begin{fact}
$L_n(1)=1$ for any $n \ge 0$ in the Legendre polynomials.
\end{fact}
\begin{fact}\label{fac:legendre_inner_product}
The Legendre polynomials constitute an orthogonal basis with respect to the inner product on interval $[-1,1]$:
\begin{equation*}
\int_{-1}^1 L_m (x) L_n (x) \mathrm{d} x  = \frac{2}{2n+1} \delta_{mn}
\end{equation*}
where $\delta_{mn}$ denotes the Kronecker delta, i.e., it equals to $1$ if $m=n$ and to $0$ otherwise.
\end{fact}


For any polynomial $P(x)$ of degree at most $d$ with complex coefficients, there exists a set of coefficients from the above properties such that
\begin{equation*}
P(x) = \sum_{i=0}^{d} \alpha_i \cdot L_i(x), \textit{ where }\alpha_i \in \mathbb{C}, \forall i \in \{0,1,2, \cdots, d\}.
\end{equation*}

\define{lem:max_average_square}{Lemma}{
For any polynomial $P(t)$ of degree at most $d$ from $R$ to $\mathbb{C}$, for any interval $[S,T]$, \[\max_{t \in [S,T]} |P(t)|^2 \le (d+1)^2 \cdot \frac{1}{T-S}\int_{S}^T |P(t)|^2 \mathrm{d} x.\]
}
\state{lem:max_average_square}

We provide a proof in Appendix \ref{sec:proof_max_average_square}.

\subsection{Gram matrix and its determinant}\label{sec:gram_determinant}
We provide an brief introduction to Gramian matrices (please see \cite{Hazewinkel} for a complete introduction). We use $\langle x,y\rangle$ to denote the inner product between vector $x$ and vector $y$.

Let $\vec{v}_1,\cdots,\vec{v}_n$ be $n$ vectors in an inner product space and $\mathrm{span}\{\vec{v}_1,\cdots,\vec{v}_n\}$ be the linear subspace spanned by these $n$ vectors with coefficients in $\mathbb{C}$, i.e., $\left\{ \underset{i\in [n]}{ \sum} \alpha_i \vec{v}_i| \forall i \in [n], \alpha_i \in \mathbb{C} \right\}$. The Gram matrix $\Gram_n$ of $\vec{v}_1,\cdots,\vec{v}_n$ is an $n \times n$ matrix defined as $\Gram_n(i,j)=\langle \vec{v}_i, \vec{v}_j \rangle$ for any $i \in [n]$ and $j \in [n]$.
\begin{fact}
$\det(\Gram_n)$ is the square of the volume of the parallelotope formed by $\vec{v}_1,\cdots,\vec{v}_n$.
\end{fact}
Let $\Gram_{n-1}$ be the Gram matrix of $\vec{v}_1,\cdots,\vec{v}_{n-1}$. Let $\vec{v}^{\parallel}_n$ be the projection of $v_n$ onto the linear subspace $\mathrm{span}\{\vec{v}_1,\cdots,\vec{v}_{n-1}\}$ and $\vec{v}^{\perp}_n=\vec{v}_n-\vec{v}^{\parallel}_n$. We use $\|\vec{v}\|$ to denote the length of $\vec{v}$ in the inner product space, which is $\sqrt{\langle \vec{v},\vec{v} \rangle}$.
\begin{claim}\label{cla:orthogonal_distance}
 \[\|\vec{v}^{\perp}_n\|^2=\frac{\det(\Gram_{n-1})}{\det(\Gram_n)}.\]
\end{claim}
\begin{proof}
\begin{equation*}
\det(\Gram_n)=\volume^2(\vec{v}_1,\cdots,\vec{v}_n)=\volume^2(\vec{v}_1,\cdots,\vec{v}_{n-1}) \cdot \|\vec{v}^{\perp}_n\|^2=\det(\Gram_n) \cdot \|\vec{v}^{\perp}_n\|^2.
\end{equation*}
\end{proof}

\section{Robust Polynomial Interpolation Algorithm}\label{sec:robustpoly}
In Section \ref{sec:constant_success_probability}, we show how to learn a low degree polynomial by using linear number of samples, running polynomial time, and achieving constant success probability. In Section \ref{sec:boosting_success_probability}, we show to how boost the success probability by rerunning previous algorithm several times.
\subsection{Constant success probability}\label{sec:constant_success_probability}

We show how to learn a degree-$d$ polynomial $P$ with $n=O(d)$ samples and prove Theorem \ref{thm:faster_poly_learning} in this section.  For convenience, we first fix the interval to be $[-1,1]$ and use $\|f\|_{[-1,1]}^2=\frac{1}{2} \int_{-1}^1  |f(t)|^2 \mathrm{d}t$.

\begin{lemma}\label{lmn:points_approx_polynomial}
Let $d \in \mathbb{N}$ and $\epsilon \in \mathbb{R^+}$, there exists an efficient algorithm to compute a partition of $[-1,1]$ to $n=O(d/\epsilon)$ intervals $I_1,\cdots,I_n$ such that for any degree $d$ polynomial $P(t) : \mathbb{R} \rightarrow \mathbb{C}$ and any $n$ points $x_1,\cdots,x_n$ in the intervals $I_1,\cdots,I_n$ respectively, the function $Q(t)$ defined by \begin{equation*}
Q(t) = P(x_j) \quad \text{if} \quad t \in I_j
\end{equation*}
approximates $P$ by
\begin{equation}\label{eq:points_approx_polynomial}
\|Q-P\|_{[-1,1]} \le \epsilon \|P\|_{[-1,1]}.
\end{equation}
\end{lemma}
One direct corollary from the above lemma is that observing $n=O(d/\epsilon)$ points each from $I_1,\cdots,I_n$ provides a good approximation for all degree $d$ polynomials. For any set $S=\{t_1,\cdots,t_m\}$ where each $t_i \in [-1,1]$ and a distribution with support $\{ w_1,\cdots,w_m \}$ on $S$ where $ \overset{m}{ \underset{i=1}{\sum} } w_i=1$ and $w_i \geq 0$ for each $i\in [m]$, we define $\|x \|_{S,w}= ( \sum_{i=1}^m w_i \cdot |x(t_i)|^2 )^{1/2}.$
\begin{corollary}\label{cor:good_approximation}
Let $I_1,\cdots,I_n$ be the intervals in the above lemma and $w_j=|I_j|/2$ for each $j \in [n]$. For any $x_1,\cdots,x_n$ in the intervals $I_1,\cdots,I_n$ respectively, we consider $S=\{x_1,\cdots,x_n\}$ with the distribution $w_1,\cdots,w_n$. Then for any degree $d$ polynomial $P$, we have
\[
\|P\|_{S,w} \in \left[(1-\epsilon) \|P\|_{[-1,1]}, (1+\epsilon) \|P\|_{[-1,1]}\right].\]
\end{corollary}
We first state the main technical lemma and finish the proof of the above lemma (we defer the proof of Lemma \ref{lem:derivative_poly} to Appendix~\ref{sec:proof_lem_derivative_poly}).
\define{lem:derivative_poly}{Lemma}{
For any degree $d$ polynomial $P(t) : \mathbb{R} \rightarrow \mathbb{C}$ with derivative $P'(t)$, we have,
\begin{equation}\label{ineq:norm_2_d2}
 \int_{-1}^1  (1-t^2) |P'(t)|^2 \mathrm{d}t \le 2d^2 \int_{-1}^1 |P(t)|^2 \mathrm{d}t.
 \end{equation}
}
\state{lem:derivative_poly}
\begin{proofof}{Lemma \ref{lmn:points_approx_polynomial}}
We set $m=10d/\epsilon$ and show a partition of $[-1,1]$ into $n \le 20 m$ intervals. We define $g(t)=\frac{\sqrt{1-t^2}}{m}$ and $y_0=0$. Then we choose $y_i=y_{i-1} + g(y_{i-1})$ for $i \in \mathbb{N^+}$. Let $l$ be the first index of $y$ such that $y_l \ge 1 - \frac{9}{m^2}$. We show $l \lesssim m$. 

Let $j_k$ be the first index in the sequence such that $y_{j_k} \ge 1-2^{-k}$. Notice that \[j_2 \le \frac{3/4}{\frac{\sqrt{1-(3/4)^2}}{m}} \le 1.5 m\] and \[y_{i}-y_{i-1}=g(y_{i-1})=\frac{\sqrt{1-y_{i-1}^2}}{m} \ge \frac{\sqrt{1-y_{i-1}}}{m}.\] Then for all $k>2$, we have \[j_{k}-j_{k-1} \le \frac{2^{-k}}{\frac{\sqrt{1-y_{(j_k-1)}}}{m}} \le 2^{-k/2} m.\] Therefore $j_k \le \left(1.5 + (2^{-3/2}+ \cdots 2^{-k/2})\right)m $ and $l \le 10 m$.

Because $y_{l-1} \le 1 - \frac{9}{m^2}$, for any $j \in [l]$ and any $x \in [y_{i-1},y_i]$, we have the following property:
\begin{equation}\label{eq:property_partition}
\frac{1-x^2}{m^2} \ge \frac{1}{2} \cdot \frac{(1-y_{i-1}^2)}{m^2} = (y_i - y_{i-1})^2/2.
\end{equation}

Now we set $n$ and partition $[-1,1]$ into $I_1,\cdots,I_n$ as follows:
\begin{enumerate}
\item $n=2(l+1)$.
\item For $j\in [l]$, $I_{2j-1}=[y_{j-1},y_j]$ and $I_{2j}=[-y_j,-y_{j-1}]$.
\item $I_{2l+1}=[y_l,1]$ and $I_{2l+2}=[-1,-y_l]$.
\end{enumerate}
For any $x_1,\cdots,x_n$ where $x_j \in I_j$ for each $j \in [n]$, we rewrite the $\mathrm{LHS}$ of \eqref{eq:points_approx_polynomial} as follows:
\begin{align}\label{eq:has_A_and_B}
\underbrace{\sum_{j=1}^{n-2} \int_{I_j} \left|P(x_j)-P(t)\right|^2 \mathrm{d} t}_A + \underbrace{\int_{I_{n-1}} \left|P(x_{n-1})-P(t)\right|^2 \mathrm{d} t + \int_{I_{n}} \left|P(x_{n})-P(t)\right|^2 \mathrm{d} t}_B.
\end{align}
For A in Equation (\ref{eq:has_A_and_B}), from the Cauchy-Schwarz inequality, we have
\begin{align*}
\sum_{j=1}^{n-2} \int_{I_j} \left|P(x_j)-P(t)\right|^2 \mathrm{d} t =\sum_{j=1}^{n-2} \int_{I_j} \left|\int_{x_j}^t P'(y) \mathrm{d} y\right|^2 \mathrm{d} t \le \sum_{j=1}^{n-2} \int_{I_j} |t-x_j| \int_{x_j}^t |P'(y)|^2 \mathrm{d} y\mathrm{d} t.\end{align*}
Then we swap $\mathrm{d}t$ with $\mathrm{d}y$ and use Equation ~\eqref{eq:property_partition}:
\[ \sum_{j=1}^{n-2} \int_{I_j} |P'(y)|^2 \int_{t \notin (x_j,y)} |t-x_j| \mathrm{d} t \mathrm{d} y
\le \sum_{j=1}^{n-2} \int_{I_j} |P'(t)|^2 \cdot |I_j|^2 \mathrm{d} t
\le \sum_{j=1}^{n-2} \int_{I_j} |P'(t)|^2 \frac{2(1-t^2)}{m^2} \mathrm{d} t.\]
We use Lemma \ref{lem:derivative_poly} to simplify it by
\[
\sum_{j=1}^{n-2} \int_{I_j} \left|P(x_j)-P(t)\right|^2 \mathrm{d} t \le \int_{-1}^1 |P'(t)|^2 \frac{2(1-t^2)}{m^2} \mathrm{d} t \le \frac{2d^2}{m^2}  \int_{-1}^1 |P(t)|^2 \mathrm{d} t.\]

For B in Equation (\ref{eq:has_A_and_B}), notice that $|I_{n-1}|=|I_n|=1-y_{l} \le 9 m^{-2}$ and for $j\in \{n-1,n\}$
\[|P(t)-P(x_j)|^2 \le 4 \max_{t \in [-1,1]} |P(t)|^2 \le 4 (d+1)^2 \|P\|^2_{[-1,1]}\]
from the properties of degree-$d$ polynomials, i.e., Lemma \ref{lem:max_average_square}. Therefore B in Equation (\ref{eq:has_A_and_B}) is upper bounded by $2 \cdot 4 (d+1)^2 (9 m^{-2}) \|P(t)\|^2_{[-1,1]}$.

From all discussion above, $\|Q(t)-P(t)\|_{[-1,1]}^2 \le \frac{99 d^2}{m^2} \le \epsilon^2$.
\end{proofof}

Now we use the above lemma to provide a faster learning algorithm for polynomials on interval $[-1,1]$ with noise  instead of using the $\epsilon$-nets argument. Algorithm \textsc{RobustPolynomialLearningFixedInterval} works as follows:
\begin{enumerate}
\item Let $\epsilon=1/20$ and $I_1,\cdots,I_n$ be the intervals for $d$ and $\epsilon$ in Lemma \ref{lmn:points_approx_polynomial}.
\item Random choose $x_j \in I_j$ for every $j \in [n]$ and define $S=\{x_1,\cdots,x_n\}$ with weight $w_1=\frac{|I_1|}{2}, \cdots, w_n= \frac{|I_n|}{2}$.
\item Find the degree $d$ polynomial $Q(t)$ that minimizes $\|P(t)-Q(t)\|_{S,w}$ using Fact~\ref{fac:basic_l2_regression}.
\end{enumerate}

\begin{lemma}\label{lem:fixed_interval_learning}
  For any degree $d$ polynomial $P(t)$ and an arbitrary function $g(t)$,
  Algorithm \textsc{RobustPolynomialLearningFixedInterval} takes $O(d)$ samples from $x(t)=P(t) + g(t)$
  over $[-1,1]$ and reports a degree $d$ polynomial $Q(t)$ in time
  $O(d^\omega)$ such that, with probability at least $99/100$,
\begin{equation*}
\|P(t)-Q(t)\|_{[-1,1]}^2 \lesssim \|g(t)\|_{[-1,1]}^2.
\end{equation*}
\end{lemma}
\begin{proof}
Notice that $n=O(d/\epsilon)=O(d)$ and the running time depends on solving a linear regression problem( Fact~\ref{fac:basic_l2_regression} ), which takes $O(d^\omega)$ time. It is enough to bound the distance between $P$ and $Q$:
\begin{align*}
&\| P - Q \|_{[-1,1]} \\
\le \quad & 1.09 \|P - Q\|_{S,w} & \text{by~Corollary \ref{cor:good_approximation}}\\
= \quad & 1.09\| x - g - Q\|_{S,w} &\text{~by~$x=P+g$} \\
\le \quad & 1.09\|g\|_{S,w} + 1.09\| x - Q \|_{S,w} & \text{~by~triangle~inequality}\\
\le \quad & 1.09 \| g \|_{S,w} + 1.09 \|x - P\|_{S,w} & Q = \underset{ \text{degree-}d ~R }{\arg\min} \|R-x\|_{S,w}\\ 
\le \quad & 2.2\|g\|_{S,w}
\end{align*}
Because $\underset{S}{\E} [\|g\|^2_{S,w}]=\|g\|^2_{[-1,1]}$, we know that $\| P - Q \|_{[-1,1]} \le 2200 \|g\|_{[-1,1]}$ with probability $\ge .999$ by using Markov's inequality.
\end{proof}

For any function $f:[0,T] \rightarrow \mathbb{C}$, let $\widetilde{f}(t)=f(\frac{2t-T}{T})$. Then $\|\widetilde{f}\|_{[-1,1]}=\|f\|_T$ from the definition. Hence we can switch any interval $[0,T]$ to $[-1,1]$ and use Lemma \ref{lem:fixed_interval_learning}.
\restate{thm:faster_poly_learning}

\subsection{Boosting success probability}\label{sec:boosting_success_probability}
Notice that the success probability of Theorem \ref{thm:faster_poly_learning} is only constant, and the proof technique of obtaining that result cannot be modified to $1-1/\poly(d)$ or $1-2^{-\Omega(d)}$ success probability due to using Markov's inequality. However, we can use that algorithm as a black box, and rerun it $O(\log(1/p))$ (for any $p>0$) times on fresh samples. Using the careful median analysis from \cite{MP14} gives

\define{thm:accurate_poly_learning}{Theorem}{%
  For any degree $d$ polynomial $P(t)$, an arbitrary function $g(t)$, and any $p>0$, Procedure  \textsc{RobustPolynomialLearning}$^+$ in Algorithm~\ref{alg:main_1}
  takes $O(d \log(1/p))$ samples from $x(t)=P(t) + g(t)$
  over $[0,T]$ and reports a degree $d$ polynomial $Q(t)$ in time
  $O( d^{\omega} \log(1/p))$ such that, with probability at least $1-p$,
\begin{equation*}
\|P(t)-Q(t)\|_T^2 \lesssim \|g(t)\|_T^2.
\end{equation*}
where $\omega< 2.373$ is matrix multiplication exponent.
}

\state{thm:accurate_poly_learning}
\begin{proof}
We run algorithm \textsc{RobustPolynomialLearning} $R$ rounds with $O(d)$ independent and fresh samples per round. We will obtain $R$ degree-$d$ polynomials $Q_1(t), Q_2(t), \cdots, Q_R(t)$.  We say a polynomial $Q_i(t)$ is good if $\| Q_i(t) - P(t) \|_T^2 \lesssim \| g(t) \|_T^2$. Using the Chernoff bound, with probability at least $1-2^{-\Omega(R)}$, at least a $3/4$ fraction of the polynomials are ``good''. We output polynomial $Q(t) = Q_{j^*}(t)$ such that
\begin{equation}\label{eq:median_of_R_polynomials}
j^* = \underset{j \in [R] }{\arg\min}(\median \{ \| Q_j(t) -Q_1(t) \|_T^2,  \| Q_j(t) -Q_2(t) \|_T^2, \cdots, \| Q_j(t) -Q_R(t) \|_T^2 \})
\end{equation}
The Equation (\ref{eq:median_of_R_polynomials}) can be solved in following straightforward way. For $i\neq j$, it takes $O(d)$ time to compute $\| Q_j(t) - Q_i(t) \|_T^2$. Because of the number of pairs is $O(R^2)$, thus it takes $O(R^2 d)$ time write down a $R\times R$ matrix. For each column, we run linear time $1$-median algorithm. This step takes $O(R^2)$ time. At the end, $j^*$ is index of the column that has the smallest median value. Thus, polynomial $Q(t) = Q_{j^*}(t)$ 1 the 0 with probability at least $1-p$ by choosing $R = O(\log(1/p))$. The running time is not optimized yet.

To improve the dependence on $R$ for running time, we replace the step of solving Equation (\ref{eq:median_of_R_polynomials}) by an approach that is similar to \cite{MP14}. We choose a new set of samples $S$, say $S= \{t_1, t_2,\cdots, t_n\}$ and $n=O(d)$. Using Fact \ref{fac:multipoint_evaluation_of_polynomial}, we can compute $Q_{i}(t_j)$ for all $i,j\in [R] \times [n]$ in $O(R d \poly(\log(d)) )$ time. Define
\begin{equation}\label{eq:def_of_wt_Q_j}
\wt{Q}_j = \underset{i\in [R]}{\median}~ Q_i(t_j), \forall j \in [n] .
\end{equation}
Our algorithm will output a degree-$d$ polynomial $Q$ which is the optimal solution of this problem, $\underset{ \text{degree-}d~Q' }{\min} \| Q' - \wt{Q} \|_{S,w}$.\footnote{Outputting $Q = \underset{ \text{degree-}d~Q' }{\arg \min} \| Q' - x \|_{S,w}$ is not good enough, because it only gives constant success probability.} In the rest of the proof, we will show that $\| Q - P \|_T \lesssim \| g\|_T$ with probability at least $1-2^{-\Omega(R)}$.

Notice that Equation (\ref{eq:def_of_wt_Q_j}) implies that $\wt{Q}_j - P(t_j) = \underset{ i\in [R] }{\median} (Q_i(t_j) - P(t_j) )$. Fix a coordinate $j$ and applying the proof argument of Lemma 6.1 in \cite{MP14}, we have
\begin{equation*}
(\wt{Q}_j - P(t_j) )^2 \lesssim \underset{\text{good} ~i}{\mean} ( Q_i(t_j) - P(t_j) )^2
\end{equation*}
Taking the weighted summation over all the coordinates $j$, we have
\begin{equation*}
\| \wt{Q} - P \|_{S,w}^2 \lesssim \underset{\text{good} ~i}{\mean} \| Q_i - P \|_{S,w}^2
\end{equation*}

Using Corollary \ref{cor:good_approximation}, for each good $i$,
\begin{equation*}
\| Q_i - P \|_{S,w}^2 \lesssim \| Q_i - P \|_T^2
\end{equation*}
Combining the above two inequalities gives
\begin{equation}\label{eq:wt_Q_minus_P_bound_by_g}
\| \wt{Q} - P \|_{S,w}^2 \lesssim  \underset{\text{good} ~i}{\mean}  \| Q_i - P \|_T^2 \lesssim \| g\|_T^2
\end{equation}
Because $Q$ is the optimal solution for $\wt{Q}$, then
\begin{equation}\label{eq:wt_Q_minus_Q_bound_by_g}
\| \wt{Q} - Q\|_{S,w}^2 \leq \| \wt{Q} - P\|_{S,w}^2 \lesssim \| g\|_T^2
\end{equation}
Using Corollary  \ref{cor:good_approximation} and for any good $i,i'$, $\| Q_i - Q_{i'}\|_T \lesssim \| g\|_T$, we can replace $P$ by $Q_{i'}$ in the Equation~(\ref{eq:wt_Q_minus_P_bound_by_g}). Thus, for any $Q_{i'}$ where $i'$ is good,
\begin{equation}\label{eq:wt_Q_minus_Q_i_bound_by_g}
\| \wt{Q} - Q_{i'} \|_{S,w}^2 \lesssim \| g\|_T^2
\end{equation}

For any good $i'$,
\begin{align*}
 ~&\| Q_{i'} - Q \|_T  \\
\lesssim ~ &\| Q_{i'} - Q \|_{S,w} &\text{~by~Corollary~\ref{cor:good_approximation} } \\
\leq ~ &\| Q_{i'} - \wt{Q} \|_{S,w} + \| \wt{Q} - Q \|_{S,w} & \text{~by~triangle~inequality} \\
\lesssim ~ & \| g \|_T & \text{~by~Equation~(\ref{eq:wt_Q_minus_Q_bound_by_g}) and (\ref{eq:wt_Q_minus_Q_i_bound_by_g})}
\end{align*}
Thus, our algorithm takes $O(dR)$ samples from $x(t) = P(t) +g(t)$ over $[0,T]$ and reports a polynomial $Q(t)$
in time $O(Rd^\omega)$ such that, with probability at least $1-2^{-\Omega(R)}$, $\| P(t) - Q(t) \|_T^2 \lesssim \| g(t) \|_T^2$. Choosing $R=O(\log(1/p))$ completes the proof.
\end{proof}

\section{Bounding the Magnitude of a Fourier-sparse Signal in Terms of Its Average Norm}\label{sec:technicalsparse}
The main results in this section are two upper bounds, Lemma \ref{lem:max_is_at_most_polyk_times_l2} on $\underset{t\in [0,T] }{\max} |x(t)|^2$ and Lemma \ref{lem:x_dot_H_is_small_outside_T} on $|x(t)|^2$ for $t>T$, in terms of the typical signal value $\|x\|^2_T=\frac{1}{T}\int_0^T |x(t)|^2 \mathrm{d} t$. We prove Lemma \ref{lem:max_is_at_most_polyk_times_l2} in Section \ref{sec:inside_fixed_region} and Lemma \ref{lem:x_dot_H_is_small_outside_T} in Section \ref{sec:outside_fixed_region}

\subsection{Bounding the maximum inside the interval}\label{sec:inside_fixed_region}
The goal of this section is to prove Lemma \ref{lem:max_is_at_most_polyk_times_l2}.
\define{lem:max_is_at_most_polyk_times_l2}{Lemma}{
For any $k$-Fourier-sparse signal $x(t) : \mathbb{R} \rightarrow \mathbb{C}$ and any duration $T$, we have
\begin{equation*}
\underset{t\in [0,T] }{\max} |x(t)|^2 \lesssim k^4 \log^3 k \cdot \|x\|_T^2
\end{equation*}}
\state{lem:max_is_at_most_polyk_times_l2}
\begin{proof}
Without loss of generality, we fix $T=1$. Then $\|x\|_T^2=\int_0^1 |x(t)|^2 \mathrm{d} t$. Because $\|x\|^2_T$ is the average over the interval $[0,T]$, if $t^*= \underset{ t \in [0,T] }{ \arg\max } |x(t)|^2$ is not $0$ or $T=1$, we can rescale the two intervals $[0,t^*]$ and $[t^*,T]$ to $[0,1]$ and prove the desired property separately. Hence we assume $|x(0)|^2 = \underset{t\in [0,T] }{\max} |x(t)|^2$ in this proof.
\begin{claim}\label{cla:linear_relationships}
For any $k$, there exists $m = O(k^2 \log k)$ such that for any $k$-Fourier-sparse signal $x(t)$, any $t_0 \ge 0$ and $\tau>0$, there always exist $C_1,\cdots,C_m \in \mathbb{C}$ such that the following properties hold,
\begin{eqnarray*}
\mathrm{Property~\RN{1}} &&|C_j| \le 11 \text{ ~for ~all~} j \in [m], \\
\mathrm{Property~\RN{2}} &&x(t_0)=\sum_{j \in [m]} C_j \cdot x(t_0 + j \cdot \tau).
\end{eqnarray*}
\end{claim}
We first use this claim to finish the proof of Lemma \ref{lem:max_is_at_most_polyk_times_l2}. We choose $t_0=0$ such that $\forall \tau>0$, there always exist $C_1,\cdots,C_m \in \mathbb{C},$ and
\[x(0)=\sum_{j \in [m]} C_j \cdot x(j \cdot \tau).\]
By the Cauchy-Schwarz inequality, it implies that for any $\tau$,
\begin{eqnarray}\label{eq:x_0_is_bounded_by_sum_of_m_terms}
|x(0)|^2 &\leq& m \sum_{j \in [m]} |C_j|^2 |x(j \cdot \tau)|^2 \notag \\
&\lesssim& m \sum_{j \in [m]} |x(j \cdot \tau)|^2.
\end{eqnarray}
At last, we obtain
\begin{align*}
 \quad |x(0)|^2 =\quad & m \int_{0}^{1/m} |x(0)|^2 \mathrm{d} \tau \\
\lesssim \quad & m \cdot \int_{0}^{1/m} ( m \sum_{j=1}^m |x(j \cdot \tau)|^2 )\mathrm{d} \tau\\ 
 =\quad & m^2 \cdot \sum_{j=1}^m \int_{0}^{1/m}  |x(j \cdot \tau)|^2 \mathrm{d} \tau  \\
 = \quad& m^2 \cdot \sum_{j=1}^m \frac{1}{j} \int_{0}^{j/m}  |x(\tau)|^2 \mathrm{d} \tau  \\
 \leq \quad & m^2 \cdot \sum_{j=1}^m \frac{1}{j} \cdot \int_{0}^{1}  |x(\tau)|^2 \mathrm{d} \tau \\ 
\lesssim \quad & m^2  \log m \cdot \|x\|_T^2 
\end{align*}
where the first inequality follows by Equation~(\ref{eq:x_0_is_bounded_by_sum_of_m_terms}), the second inequality follows by $j/m\leq 1$ and the last step follows by $\sum_{i=1}^m \frac{1}{i} = O(\log m)$.
From $m= O(k^2 \log k)$, we obtain $|x(0)|^2 = O(k^4 \log^3 k \|x\|_T^2) $.
\end{proof}

To prove Claim \ref{cla:linear_relationships}, we use the following lemmas about polynomials. We defer their proofs  to Appendix~\ref{sec:supplement_proof_poly}.
\define{lem:bound_residual_polynomial_coefficients}{Lemma}{
Let $Q(z)$ be a degree $k$ polynomial, all of whose roots are complex numbers with absolute value $1$. For any integer $n$, let $r_{n,k}(z)=\sum_{l=0}^{k-1} r_{n,k}^{(l)} \cdot z^l$ denote the residual polynomial of
\begin{equation*}
r_{n,k}(z) \equiv z^n  \pmod {Q(z)}.
\end{equation*}
Then, each coefficient of $r_{n,k}$ is bounded: $|r_{n,k}^{(l)}| \le 2^k n^{k-1}$ for any $l$.}
\state{lem:bound_residual_polynomial_coefficients}
\define{lem:existence_poly_k_roots}{Lemma}{
For any $k \in \mathbb{Z}$ and any $z_1,\cdots,z_k$ on the unit circle of $\mathbb{C}$, there always exists a degree $m = O(k^2 \log k)$ polynomial $P(z)= \overset{m}{\underset{j=0}{\sum}} c_j z^j$ with the following properties:
\begin{eqnarray*}
\mathrm{Property~\RN{1}} && P(z_i) =0, \forall i \in \{1, \cdots, k\},\\
\mathrm{Property~\RN{2}} &&  c_0=1,\\
\mathrm{Property~\RN{3}} &&  |c_j| \le 11, \forall j \in \{1, \cdots, m \}.
\end{eqnarray*}
}
\state{lem:existence_poly_k_roots}
\begin{proofof}{Claim \ref{cla:linear_relationships}}
For $x(t)= \overset{k}{\underset{i=1}{\sum}} v_i e^{2\pi \i f_i t}$, we fix $t_0$ and $\tau$ then rewrite $x(t_0+j \cdot \tau)$ as a polynomial of $b_i=v_i \cdot e^{2\pi\i f_i t_0}$ and $z_i = e^{2\pi \i f_i \tau}$ for each $i \in [k]$.
\begin{eqnarray*}
x(t_0 + j \cdot \tau) &=& \sum_{i=1}^k v_i e^{2\pi \i f_i \cdot (t_0+j \cdot \tau)} \\
&=& \sum_{i=1}^k v_i e^{2\pi \i f_i t_0} \cdot e^{2\pi\i f_i \cdot j \tau} \\
&=&\sum_{i=1}^k b_i \cdot z_i^j.
\end{eqnarray*}
Given $k$ and $z_1,\cdots,z_k$, let $P(z) = \sum_{j=0}^m c_j z^j$ be the degree $m$ polynomial in Lemma \ref{lem:existence_poly_k_roots}.
\begin{eqnarray}\label{eq:sum_of_f_equal_sum_of_P}
\sum_{j=0}^m c_j x(t_0 + j \tau)  &=&  \sum_{j=0}^m c_j \sum_{i=1}^k b_i \cdot z_i^j  \notag \\
&=&  \sum_{i=1}^k b_i \sum_{j=0}^m c_j \cdot z_i^j \notag \\
&=& \sum_{i=1}^k b_i P(z_i)\notag \\
&=& 0,
\end{eqnarray}
where the last step follows by Property~\RN{1}~of~$P(z)$ in Lemma \ref{lem:existence_poly_k_roots}.
From the Property \RN{2} and \RN{3} of $P(z)$, we obtain $x(t_0) = - \sum_{j=1}^m c_j x(t_0 + j\tau).$
\end{proofof}

\subsection{Bounding growth outside the interval}\label{sec:outside_fixed_region}
Here we show signals with sparse Fourier transform cannot grow too
quickly outside the interval.
\begin{lemma}\label{lem:x_dot_H_is_small_outside_T}
Let $x(t)$ be a $k$-Fourier-sparse signal. For any $T>0$ and any $t>T$,
\begin{equation*}
| x(t) |^2 \leq k^7 \cdot (2kt/T)^{ 2.5k } \cdot \|x\|^2_T.
\end{equation*}
\end{lemma}
\begin{proof}
For any $t>T$, let $t=t_0+n \cdot \tau$ such that $t_0 \in [0,T/k], \tau \in [0,T/k]$ and $n \le \frac{2kt}{T}$. We define $b_i=v_i e^{2 \pi \i f_i t_0}$, and $z_i= e^{2 \pi \i f_i \tau}$ such that $x(t_0+ n \cdot \tau)=\sum_{j \in [k]}b_j z_j^n$.

By Lemma \ref{lem:bound_residual_polynomial_coefficients}, we have for any $z_1, z_2, \cdots, z_k$ and any $n$,
\begin{equation*}
z^n \equiv \sum_{i=0}^{k-1} a_i z^i \pmod {\prod_{i=1}^{k} (z-z_i) },
\end{equation*}
where $|a_i| \leq 2^k \cdot n^{k}, \forall i \in \{0,1,\cdots, k-1\}$. Thus, we obtain
\begin{eqnarray*}
x(t_0 + n \tau) = \sum_{j=1}^k b_j z_j^n =\sum_{j=1}^k b_j ( \sum_{i=0}^{k-1} a_i z_j^i).
\end{eqnarray*}
From the fact that $x(t_0 + i \cdot \tau)=\sum_{j \in [k]} b_j z_j^i$, we simplify it to be
\[
x(t_0 + n \tau)=\sum_{i=0}^{k-1} a_i \sum_{j=1}^k b_j z_j^i =\sum_{i=0}^{k-1} a_i x(t_0 + i \cdot \tau).
\]
Because $(t_0 + i \cdot \tau) \in [0,T]$ for any $i=0,\cdots,k-1$, we have $|x(t_0 + i \tau)|^2 \le \underset{t\in[0,T] }{\max} |x(t)|^2 \lesssim k^4 \log^3 k \|x\|^2_T$ from Lemma \ref{lem:max_is_at_most_polyk_times_l2}. Hence
\begin{align*}
|x(t_0+n \cdot \tau)|^2 \leq \quad & k \sum_{i=0}^{k-1} |a_i|^2 \cdot |x(t_0 + i \cdot \tau)|^2 \\
\leq \quad & k \sum_{i=0}^{k-1} n^{2.2 k} \cdot \underset{t\in[0,T] }{\max} |x(t)|^2\\
\leq \quad & k^7 \cdot (2kt/T)^{2.2 k} \|x\|^2_T.
\end{align*}
Thus, we complete the proof.
\end{proof}




\section{Hash Functions and Filter Functions}\label{sec:hashfilter}

\subsection{Permutation function and hash function}
We first review the permutation function $P_{\sigma,a,b}$ and the hash function $h_{\sigma,b}$ in \cite{PS15}, which translates discrete settings to the continuous setting.
\begin{definition}\label{def:permutation}
For any signal $x(t) : \mathbb{R} \rightarrow \mathbb{C}$ and $a,b,\sigma \in \mathbb{R}$, let $(P_{\sigma,a,b} x)(t) = x \big( \sigma(t-a) \big) e^{-2\pi \i \sigma b t}$.
\end{definition}
\define{lem:permutation}{Lemma}{
$\widehat{P_{\sigma,a,b} x} (\sigma(f-b)) =\frac{1}{\sigma} e^{-2\pi \i \sigma a f } \widehat{x}(f)$ and $\widehat{P_{\sigma,a,b} x} ( f ) = \frac{1}{\sigma }e^{-2\pi \i \sigma a (f/\sigma+b) } \widehat{x}(f/\sigma+b)$
}
\state{lem:permutation}
For completeness, we provide a proof of Lemma \ref{lem:permutation} in Appendix \ref{sec:proof_permutation}.
\begin{definition} \cite{PS15}
Let $\pi_{\sigma,b}(f) = 2\pi \sigma(f-b) \pmod {2\pi}$ and $h_{\sigma,b}(f) = \mathrm{round} (\pi_{\sigma,b}(f) \cdot \frac{B}{2\pi})$ be the hash function that maps frequency $f \in [-F,F]$ into bins $\{0,\cdots,B-1\}$.
\end{definition}
\begin{claim}\label{cla:PS15_hash_claims} \cite{PS15}
For any $\Delta>0$, let $\sigma$ be a sample uniformly at random from $[\frac{1}{B\Delta}, \frac{2}{B\Delta}]$.

(\RN{1}) If $\Delta \leq|f^+ - f^-| < \frac{(B-1)\Delta}{2} $, then $\mathsf{Pr}[h_{\sigma,b}(f^+) = h_{\sigma,b}(f^-)]=0$

(\RN{2}) If $\frac{(B-1)\Delta}{2} \leq |f^+ - f^-|$, then
$\mathsf{Pr}[h_{\sigma,b}(f^+) = h_{\sigma,b}(f^-)] \lesssim \frac{1}{B}$
\end{claim}

From previous work \cite{HIKP12, HIKP,PS15}, uniformly sampling from $[A,2A]$ for some large $A \ge \wt{T}$ provides an almost uniform sample on $[0,\wt{T}]$ when taken modulo over $\wt{T}$.
\begin{lemma}\label{lem:wrapping}
For any $\widetilde{T}$, and $ 0 \leq \widetilde{\epsilon}, \widetilde{\delta} \leq \widetilde{T}  $, if we sample $\widetilde{\sigma}$ uniformly at random from $[A,2A]$, then
\begin{equation}
\frac{2\widetilde{\epsilon} }{ \widetilde{T} } - \frac{2\widetilde{\epsilon} }{ A } \leq \mathsf{Pr} \left[ \widetilde{\sigma}  {\pmod {\widetilde{T}} } \in [ \widetilde{\delta} - \widetilde{\epsilon}, \widetilde{\delta} + \widetilde{\epsilon} ~] \right] \leq \frac{2\widetilde{\epsilon} }{ \widetilde{T} } + \frac{4\widetilde{\epsilon} }{ A }.
\end{equation} 
\end{lemma}

\subsection{Filter function}
We state the properties of filter function $(H(t), \wh{H}(f))$ and $(G(t),\wh{G}(f))$, the details of proofs are presented in Appendix \ref{sec:properties_of_H} and \ref{sec:properties_of_G}.

\define{lem:property_of_filter_H}{Lemma}{
Given $s_0, s_1, 0 < s_3 < 1,\ell > 1, 0< \delta <1$, where $\ell = \Theta( k\log( k/\delta )) $.The filter function $(H(t),\widehat{H}(f))$ has the following properties,
\begin{eqnarray*}
&\mathrm{Property~\RN{1}} : &H(t) \in [ 1 - \delta, 1], \text{~when~} |t| \leq  ( \frac{1}{2} - \frac{2}{s_1} ) s_3.\\
&\mathrm{Property~\RN{2}} : &H(t) \in [0,1], \text{~when~}  (\frac{1}{2} - \frac{2}{s_1}) s_3 \leq |t| \leq \frac{1}{2} s_3.\\
&\mathrm{Property~\RN{3}} : &H(t) \leq s_0 \cdot (s_1( \frac{|t|}{s_3}-\frac{1}{2})+2)^{-\ell},\forall |t| > \frac{1}{2} s_3.\\
&\mathrm{Property~\RN{4}} : &\supp(\widehat{H}(f) ) \subseteq [-\frac{s_1  \ell}{2 s_3}, \frac{s_1 \ell }{2 s_3}].
\end{eqnarray*}

For any exact $k$-Fourier-sparse signal $x^*(t)$, we shift the interval from $[0,T]$ to $[-1/2,1/2]$ and consider $x^{*}(t)$ for $t \in [-1/2,1/2]$ to be our observation, which is also $x^*(t) \cdot \rect_1(t)$.
\begin{eqnarray*}
&\mathrm{Property~\RN{5}} : &\int_{-\infty}^{+\infty} \bigl|x^*(t) \cdot H(t) \cdot (1- \rect_1(t) ) \bigr|^2 \mathrm{d} t < \delta \int_{-\infty}^{+\infty} | x^*(t) \cdot \rect_1(t) |^2 \mathrm{d} t.\\
&\mathrm{Property~\RN{6}} : &\int_{-\infty}^{+\infty} |x^*(t) \cdot H(t) \cdot \rect_1(t) |^2 \mathrm{d} t \in  [1-\epsilon, 1]\cdot \int_{-\infty}^{+\infty} |x^*(t) \cdot \rect_1(t) |^2 \mathrm{d} t.
\end{eqnarray*}
for arbitrarily small constant $\epsilon$. 
}
\state{lem:property_of_filter_H}

\begin{figure}[t]
  \centering
    \includegraphics[width=1.0\textwidth]{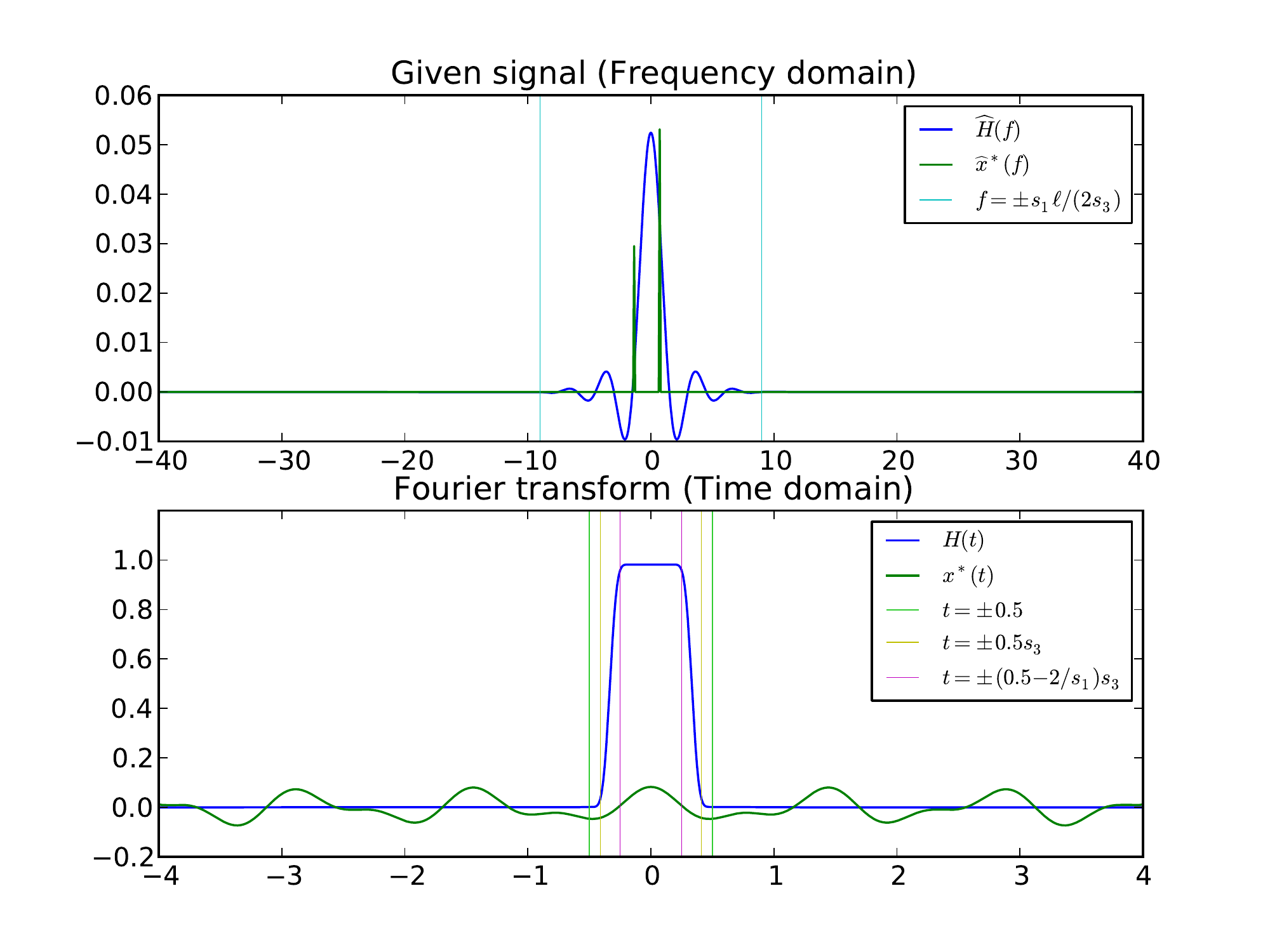}
    \caption{The filter function $(H(t), \widehat{H}(f))$ with a $k$-Fourier-sparse signal. The property \RN{1}, \RN{2} and \RN{3} are presented in the bottom one, the property \RN{4} is presented in the top one.}\label{fig:filter_H_with_sparse_signal}
\end{figure}

\define{lem:property_of_filter_G}{Lemma}{
Given $B >1$, $\delta >0$, $\alpha>0$, we set $l=\Omega(\log(\delta/k))$. The filter function $(G(t), \widehat{G}(f) )[B,\delta,\alpha,l]$ satisfies the following properties,
\begin{eqnarray*}
&\mathrm{Property~\RN{1}} : &\widehat{G}(f) \in [1 - \delta/k, 1] , \text{~if~} |f|\leq (1-\alpha)\frac{2\pi}{2B}.\\
&\mathrm{Property~\RN{2}}: &\widehat{G}(f) \in [0,1], \text{~if~}  (1-\alpha)\frac{2\pi}{2B} \leq |f| \leq \frac{2\pi}{2B}.\\
&\mathrm{Property~\RN{3}} : &\widehat{G}(f) \in [-\delta /k, \delta/k], \text{~if~}  |f| >  \frac{2\pi}{2B}
.\\
&\mathrm{Property~\RN{4}} : &\supp(G(t) ) \subset [\frac{l}{2} \cdot \frac{-B}{\pi\alpha}, \frac{l}{2} \cdot \frac{B}{\pi\alpha}].\\
&\mathrm{Property~\RN{5}} : & \underset{t}{\max} |G(t)| \lesssim  \poly(B,l).
\end{eqnarray*}
} 

\state{lem:property_of_filter_G}

\begin{figure}[t]
  \centering
    \includegraphics[width=0.6\textwidth]{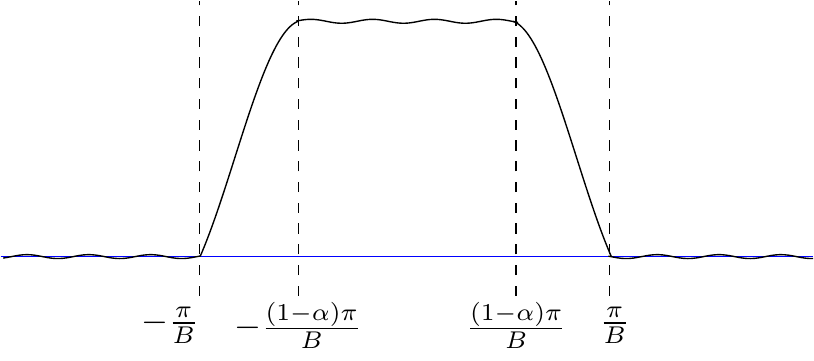}
    \caption{ $G$ and $\widehat{G}$. \cite{PS15} }\label{fig:filter_G_and_hat_G}
\end{figure}




\subsection{{\normalfont\texorpdfstring{\textsc{HashToBins}} \quad} }

We first define two functions $G^{(j)}_{\sigma,b}(t)$ and $\wh{G}^{(j)}_{\sigma,b}(f)$, then show the result returned by Procedure \textsc{HashToBins} in Algorithm \ref{alg:locateksignal_locatekinner_hashtobins} satisfying some nice properties.  The details of proofs are presented in Appendix \ref{sec:proof_hashtobins}.

\define{def:G_j_sigma_b}{Definition}{
$\forall \sigma>0, b$ and $j\in [B]$. Define,
\begin{eqnarray*}
G^{(j)}_{\sigma,b}(t) & = & \frac{1}{\sigma} G(t/\sigma) e^{2\pi\i t(j/B-\sigma b)/\sigma}\\
\widehat{G}^{(j)}_{\sigma,b}(f) &=& \widehat{G}^{\dis}(\frac{j}{B}-\sigma f -\sigma b) = \sum_{i\in \Z}\widehat{G}(i+ \frac{j}{B} -\sigma f -\sigma b) 
\end{eqnarray*}
}

\state{def:G_j_sigma_b}

\define{lem:hashtobins}{Lemma}{
  Let $u \in \C^B$ be the result of \textsc{HashToBins} under
  permutation $P_{\sigma, a, b}$, and let $j \in [B]$.  Define
  \[
  \wh{z} = \wh{x\cdot H} \cdot \wh{G}^{ (j)}_{\sigma,b},
  \]
  so
  \[
  z = (x\cdot H) * G^{ (j)}_{\sigma,b}.
  \]
  Let vector $\wh{u}\in \C^B$ denote the $B$-dimensional DFT of $u$, then $\forall j\in [B]$,
  \[
  \wh{u}[j] = z_{\sigma a}.
  \]
}
\state{lem:hashtobins}

\section{Frequency Recovery}\label{sec:freq_recover}
The goal of this section is to prove Theorem \ref{thm:frequency_recovery_k_cluster}, which is able to recover the frequencies of a signal $x^*$ has $k$-sparse Fourier transform under noise.
\restate{thm:frequency_recovery_k_cluster}
\subsection{Overview}
We give an overview of proving Theorem \ref{thm:frequency_recovery_k_cluster}.
Instead of starting with $k$-cluster recovery, we first show how to achieve one-cluster recovery.
\paragraph{One-cluster recovery.} we start with $x^*(t)=\sum_{j=1}^k v_j e^{2\pi\i f_j t}$ where there exists $f_0$ and $\Delta$ such that $f_j$ is in $[f_0-\Delta,f_0+\Delta]$ for each $j \in [k]$ and consider its properties for frequency recovery.
\begin{definition}[$(\epsilon,\Delta)$-one-cluster signal]\label{def:one_cluster}
We say that a signal $z(t)$ is an $(\epsilon,\Delta)$-one-cluster signal around $f_0$ iff $z(t)$ and $\wh{z}(f)$ satisfy the following two properties:
\begin{eqnarray*}
\mathrm{Property ~\RN{1}} &:& \int_{f_0-\Delta}^{f_0+\Delta} | \widehat{z}(f) |^2 \mathrm{d} f \geq (1-\epsilon) \int_{-\infty}^{+\infty} | \widehat{z}(f) |^2 \mathrm{d} f \\
\mathrm{Property ~\RN{2}} &:& \int_0^T | z(t) |^2 \mathrm{d} t \geq (1-\epsilon) \int_{-\infty}^{+\infty} |z(t) |^2 \mathrm{d} t.
\end{eqnarray*}
\end{definition}
The main result of one-cluster recovery is to prove that the two properties in Definition \ref{def:one_cluster} with a sufficiently small constant $\epsilon$ are sufficient to return $\wt{f_0}$ close to $f_0$ with high probability, which provides a black-box for $k$-cluster recovery algorithm.

We first prove that the pair of conditions, Property \RN{1} and Property \RN{2} in Definition \ref{def:one_cluster}, are sufficient to obtain an estimation of $e^{2 \pi \i f_0}$ in Section \ref{sec:proof_good_samples}. We also provide the proof of the correctness of Procedures \textsc{GetLegal1Sample} and \textsc{GetEmpirical1Engergy} in Section \ref{sec:proof_good_samples}.


\begin{lemma}\label{lem:get_legal_1_sample}
For a sufficiently small constant $\epsilon>0$, any $f_0 \in [-F,F]$, and $\Delta>0$, given  $\wh{\beta} \eqsim \frac{1}{\Delta \sqrt{\Delta T}}$  and an $(\epsilon,\Delta)$-one-cluster signal $z(t)$ around $f_0$, Procedure \textsc{GetLegal1Sample} in Algorithm \ref{alg:getempirical1enery_getlegal1sample} with any $\beta \le 2\wh{\beta}$ takes $O((T \Delta)^3)$ samples to output $\alpha \in \mathbb{R}$ satisfying 
\begin{equation*} 
|z(\alpha + \beta) - z(\alpha)e^{2 \pi \i f_0 \beta}| \le 0.08 (|z(\alpha)|+|z(\alpha+\beta)|),
\end{equation*}
 with probability at least 0.6.
\end{lemma}
The following lemma shows that for any $(\epsilon,\Delta)$-one-cluster signal $z(t)$ around $f_0$, we could use the above procedure to find a frequency $\wt{f_0}$ approximating $f_0$ with high probability.
\begin{lemma}\label{lem:findfrequency}
For a sufficiently small constant $\epsilon>0$, any $f_0 \in [-F,F]$, and $\Delta>0$, given an $(\epsilon,\Delta)$-one-cluster signal $z(t)$ around $f_0$ , Procedure \textsc{FrequencyRecovery1Cluster} in Algorithm \ref{alg:locate1signal_locate1inner_frequencyrecovery1cluster} returns $\wt{f_0}$ with $|\widetilde{f}_0 -f_0| \lesssim \Delta \cdot \sqrt{\Delta T} $ with probability at least $1-2^{-\Omega(k)}$.
\end{lemma}

We provide a proof of Lemma \ref{lem:findfrequency} in Section \ref{sec:proof_one_frequency_rec}. We show $z(t)=(x^*(t)+g(t)) \cdot H(t)$ satisfy Properties \RN{1} and \RN{2} (Definition \ref{def:one_cluster}) when all frequencies in $\wh{x}^*$ are in a small range in Section \ref{sec:proof_properties}.
\begin{lemma}\label{lem:three_properties}
For any $f_0 \in [-F,F]$, $\Delta'>0$, and $x^*(t) = \sum_{j=1}^k v_j e^{2\pi\i f_t}$ with $|f_j -f_0|\leq \Delta'$ for all $j \in [k]$, let $x(t)=x^{*}(t)+g(t)$ be our observable signal whose noise $\|g\|_T^2 \le c \|x^*\|^2_T$ for a sufficiently small constant $c$ and $H(t)$ be the filter function defined in Section \ref{sec:hashfilter} with $|\supp(\wh{H})|=\Delta_h$. Then $z=H \cdot x$ is an $(O(\sqrt{c}),\Delta_h+\Delta')$-one-cluster signal around $f_0$.
\end{lemma}
From all discussion above, we summarize the result of frequency recovery when $\wh{x}^*$ is in one cluster.
\define{thm:frequency_recovery_1_cluster}{Theorem}{
For any $f_0 \in [-F,F]$, $\Delta'>0$, and $x^*(t) = \sum_{j=1}^k v_j e^{2\pi\i f_j t}$ with $|f_j -f_0|\leq \Delta'$ for all $j \in [k]$, let $x(t)=x^{*}(t)+g(t)$ be our observable signal whose noise $\|g\|_T^2 \le c \|x^*\|^2_T$ for a sufficiently small constant $c$ and $H(t)$ be the filter function defined in Section \ref{sec:hashfilter} with $|\supp(\wh{H})|=\Delta_h$. 
Then Procedure $\textsc{FrequencyRecovery1Cluster}$ in Algorithm \ref{alg:locate1signal_locate1inner_frequencyrecovery1cluster} with $\Delta=\Delta'+\Delta_h$ takes $\poly(k,\log(1/\delta)) \cdot \log(FT)$ samples, runs in $\poly(k,\log(1/\delta)) \cdot \log^2(FT)$ time, returns a frequency $\widetilde{f_0}$ satisfying $| \widetilde{f}_0 -f_0 |\lesssim \Delta \sqrt{\Delta T}$ with probability at least $1-2^{-\Omega(k)}$. 
}
\state{thm:frequency_recovery_1_cluster}

\paragraph{$k$-cluster recovery.}
Given any $x^*(t)=\sum_{j=1}^k v_j e^{2\pi\i f_j t}$, we plan to convolve the filter function $G(t)$ on $x(t)\cdot H(t)$ and use Lemma \ref{lem:findfrequency} as a black box to find a list of frequencies that covers $\{f_1,\cdots,f_k\}$. 

We fix $\Delta=\polydelta$, $B = \Theta(k)$ and sample $\sigma$ uniformly at random from $[\frac{1}{B \Delta},\frac{2}{B \Delta}]$ for $k$-cluster recovery. We will cover all $f^* \in [-F,F]$ with the following property :
  \begin{equation}\label{eq:heavyfrequency}
    \int_{f^*-\Delta}^{f^*+\Delta} | \widehat{x\cdot H}(f) |^2 \mathrm{d} f \geq  T\N^2/k,
  \end{equation}
We consider one frequency $f^* \in [-F,F]$ satisfying \eqref{eq:heavyfrequency} and use $j=h_{\sigma,b}(f^*)$ to denote its index in $[B]$ after hashing $(\sigma,b)$. Recall that for $j \in [B]$, any $\sigma>0$ and any $b$, 
\[G^{(j)}_{\sigma,b}(t) = \frac{1}{\sigma} G(t/\sigma) e^{2\pi\i t(j/B-\sigma b)/\sigma} \text{ such that }
\widehat{G}^{(j)}_{\sigma,b}(f) = \sum_{i\in \Z}\widehat{G}(i+ \frac{j}{B} -\sigma f -\sigma b).\]
We set $\wh{z}=\wh{x \cdot H} \cdot \wh{G}^{(j)}_{\sigma,b}$ and $z=(x \cdot H)*G^{(j)}_{\sigma,b}$ for $f^*$ and $j=h_{\sigma,b}(f^*)$. In Section \ref{sec:proof_z_satisfies_two_properties}, we show that with high probability over the hashing $(\sigma,b)$, $(z,\wh{z})$ satisfies Property \RN{1} with $[f^*-\Delta,f^*+\Delta]$ and Property \RN{2} in Definition \ref{def:one_cluster} such that we could use Lemma~\ref{lem:findfrequency} on $z$ to recover $f^*$. 
\begin{lemma}\label{lem:z_satisfies_two_properties}
Let $f^*\in [-F,F]$ satisfy~\eqref{eq:heavyfrequency}. For a random hashing $(\sigma,b)$, let $j=h_{\sigma,b}(f^*)$ be the bucket that $f^*$ maps to under the hash such that $z=(x \cdot H)*G^{(j)}_{\sigma,b}$ and $\wh{z}=\wh{x \cdot H} \cdot \wh{G}^{(j)}_{\sigma,b}$. With probability at least $0.9$, $z(t)$ is an $(\epsilon,\Delta)$-one-cluster signal around $f^*$ .
\end{lemma}
Combining Lemma~\ref{lem:z_satisfies_two_properties} and Lemma~\ref{lem:findfrequency}, we could recover any heavy frequency $f^*$ satisfying \eqref{eq:heavyfrequency} with probability at least $0.8$. Then we repeat this procedure to guarantee that we cover all heavy frequencies and finish the proof of the main frequency recovery Theorem \ref{thm:frequency_recovery_k_cluster} in Section \ref{sec:k_mountain_frequency_recovery}.


\subsection{Analysis of {\normalfont \texorpdfstring{\textsc{GetLegal1Sample}}~ } and {\normalfont \texorpdfstring{\textsc{GetEmpirical1Energy}} ~ } }\label{sec:proof_good_samples}
Let $I = [f_0-\Delta, f_0+\Delta]$ and $\overline{I} = (-\infty, +\infty) \setminus I$ in this proof. We define $\big(z^I(t),\wh{z}^I(f)\big)$ and $\big(z^{\overline{I}}(t),\wh{z}^{\overline{I}}(f)\big)$ as follows:
\begin{equation*}
\wh{z}^I(f) =
\begin{cases}
    \wh{z}(f)    & \quad \text{if } f\in I\\
    0  & \quad \text{if } f \in \overline{I}\\
\end{cases}, \quad \wh{z}^{\overline{I}}(f) =
\begin{cases}
    0    & \quad \text{if } f\in I\\
    \wh{z}(f)  & \quad \text{if } f \in \overline{I}\\
\end{cases}
\end{equation*} 

We consider $z^I(t)$ as the ``signal'' to recover $f_0$ and treat $z^{\ov{I}}(t)$ as the ``noise''. We first show some basic properties of $z^I(t)$.

\begin{claim}\label{cla:energy_concentration}
For $z^{\ov{I}}(t)$, we have $\int_{0}^{T} |z^{\overline{I}}(t) |^2 \mathrm{d} t \le \epsilon \int_{-\infty}^{+\infty} |z(t) |^2 \mathrm{d} t$.
For $z^I(t)$, we have
\[\int_{0}^{T} |z^I(t)|^2 \mathrm{d} t \ge (1 - 5 \sqrt{\epsilon}) \int_{-\infty}^{+\infty} |z(t) |^2 \mathrm{d} t \textit{ and }\int_{0}^{T} |z^I(t)|^2 \mathrm{d} t \ge (1-6 \sqrt{\epsilon}) \int_{-\infty}^{+\infty} |z^I(t) |^2 \mathrm{d} t.\]
\end{claim}

\begin{proof}
From the definition and Property \RN{1} in Definition \ref{def:one_cluster}, we know
\begin{equation*}
z(t)=z^I(t) + z^{\ov{I}}(t) \quad \text{ and } \quad \int_{-\infty}^{+\infty} |\wh{z}^{\overline{I}}(f) |^2 \mathrm{d} f \le \epsilon \int_{-\infty}^{+\infty} |\wh{z}(f) |^2 \mathrm{d} f.
\end{equation*}
Notice that Property \RN{1}(in Definition \ref{def:one_cluster}) indicates that  \[\int_{0}^{T} |z^{\overline{I}}(t) |^2 \mathrm{d} t \le \int_{-\infty}^{+\infty} |z^{\overline{I}}(t) |^2 \mathrm{d} t=\int_{-\infty}^{+\infty} |\wh{z}^{\overline{I}}(f) |^2 \mathrm{d} f \le \epsilon \int_{-\infty}^{+\infty} |\wh{z}(f) |^2 \mathrm{d} f.\] 
On the other hand, from Property \RN{2}(in Definition \ref{def:one_cluster}), we know
\[(1-\epsilon) \int_{-\infty}^{+\infty} |z(t) |^2 \mathrm{d} t \le \int_{0}^{T} |z^I(t) + z^{\overline{I}}(t) |^2 \mathrm{d} t \le \int_{0}^{T} |z^I(t)|^2 \mathrm{d} t + 2 \int_{0}^{T}|z^I(t)| \cdot |z^{\overline{I}}(t) | \mathrm{d} t +\int_{0}^{T}|z^{\overline{I}}(t) |^2 \mathrm{d} t.\]
We have $\int_{0}^{T} |z^I(t)|^2 \mathrm{d} t \le 2 \int_{-\infty}^{+\infty} |z(t) |^2 \mathrm{d} t$ from the above inequality. From $\int_{0}^{T} |z^{\overline{I}}(t) |^2 \mathrm{d} t \le \epsilon \int_{-\infty}^{+\infty} |\wh{z}(f) |^2 \mathrm{d} f$, we bound 
\[\int_{0}^{T}|z^I(t)| \cdot |z^{\overline{I}}(t) | \mathrm{d} t \le \sqrt{2 \epsilon} \int_{-\infty}^{+\infty} |z(t) |^2 \mathrm{d} t\] by the Cauchy-Schwartz inequality and have 
\begin{equation}\label{eq:key_I_ovI}
 \int_{0}^{T} |z^I(t)|^2 \mathrm{d} t \ge (1-5 \sqrt{\epsilon}) \int_{-\infty}^{+\infty} |z(t) |^2 \mathrm{d} t.
 \end{equation}
Because $\int_{-\infty}^{+\infty} |z^{\overline{I}}(t) |^2 \mathrm{d} t \le \epsilon \int_{-\infty}^{+\infty} |z(t) |^2 \mathrm{d} t$, inequality \eqref{eq:key_I_ovI} also indicates that \[ \int_{0}^{T} |z^I(t)|^2 \mathrm{d} t \ge (1-6 \sqrt{\epsilon}) \int_{-\infty}^{+\infty} |z^I(t) |^2 \mathrm{d} t.\]
\end{proof}

One useful property of $z^I(t)$ is that its maximum can be bounded by its average on $[0,T]$.

\begin{claim}\label{cla:z_max}
$\forall t \in [0,T], |z^I(t)| \le 2 \sqrt{\Delta T} \cdot \|z^I\|_T$.
\end{claim}

\begin{proof}
From the definition $|z^I(t)|$, it is upper bounded by $\int_{f_0 - \Delta}^{f_0 + \Delta} |\wh{z^I}(f)| \mathrm{d} f$ for any $t \in [0,T]$. On the other hand,
\begin{align*}
\int_{f_0 - \Delta}^{f_0 + \Delta} |\wh{z}^I(f)| \mathrm{d} f & \le \sqrt{2 \Delta} (\int_{f_0 - \Delta}^{f_0 + \Delta} |\wh{z}^I(f)|^2 \mathrm{d} f)^{1/2}\\
& = \sqrt{2 \Delta} (\int_{ - \infty}^{+\infty} |z^I(t)|^2 \mathrm{d} t)^{1/2}\\
&\le 2 \sqrt{\Delta} (\int_{0}^{T} |z^I(t)|^2 \mathrm{d} t)^{1/2}\\
& = 2 \sqrt{\Delta T} \|z^I\|_T.
\end{align*}
\end{proof}

\begin{claim}\label{cla:close_samples}
Given $\wh{\beta}= \frac{C_{\beta}}{\Delta \cdot \sqrt{\Delta T}}$ with a sufficiently small constant $C_{\beta}$, for any two $\wh{\beta}$-close samples in $z^I(t)$, we have that
\[ \forall \alpha \in [0,T], ~\forall \beta \in [\wh{\beta},2\wh{\beta}], \quad |z^I(\alpha) e^{2\pi \i f_0 \beta}-z^{I}(\alpha+\beta)| \le 0.01 \cdot \|z^I\|_T.\]
\end{claim}

\begin{proof}
From the definition of the Fourier transform, we have
\begin{align*}
  |z^I(a+\beta) - z^I(a) e^{2\pi \i f_0 \beta} | = &~ \left|  \int_{f_0-\Delta}^{f_0+\Delta} \widehat{z^I}(f)e^{2\pi \i (f a + f_0 \beta ) } (e^{2\pi \i (f-f_0) \beta } -1) \mathrm{d} f \right| \\
 \leq &~ 2 \cdot (2 \pi \Delta \beta) \cdot \int_{f_0-\Delta}^{f_0+\Delta} | \widehat{z^I}(f) |  \mathrm{d} f & \text{~by~Taylor~expansion} \\
 \leq  &~ 4 \pi \beta \Delta \cdot \sqrt{2\Delta} \left( \int_{f_0-\Delta}^{f_0+\Delta} | \widehat{z^I}(f) |^2  \mathrm{d} f \right)^{\frac{1}{2}} & \text{~by~H\"older inequality} \\
 \leq &~ 10 \pi \wh{\beta} \Delta \cdot \sqrt{2\Delta} \left( \int_{0}^{T} | {z^I}(t) |^2  \mathrm{d} t \right)^{\frac{1}{2}} & \text{~by~inequality~\eqref{eq:key_I_ovI}} \\
 \leq &~ 10^{-2} \|z^I\|_T.
\end{align*}
\end{proof}

We consider how to output an $\alpha$ such that $e^{2 \pi \i f_0 \beta} \approx z(\alpha+\beta)/z(\alpha)$ with high probability in the rest of this section.

If we can sample from $z^I(t)$, we already know $|z^I(\alpha) e^{2\pi \i f_0 \beta}-z^{I}(\alpha+\beta)| \le 0.01 \|z^I\|_T$ from Claim~\ref{cla:close_samples}. Then it is enough to find any $\alpha$ such that $|z^I(\alpha)| \ge 0.5 \|z^I\|_T$. From Claim~\ref{cla:z_max}, we can take $O(\sqrt{\Delta T})$ samples $\big(z^I(\alpha),z^I(\alpha+\beta)\big)$ where each $\alpha$ is uniformly sampled from $[0,T]$ such that with high probability, the sample $z^I(\alpha)$ with the largest norm $|z^I(\alpha)|$ satisfies $|z^I(\alpha)| \ge 0.5 \|z^I\|_T$. Then we have $e^{2 \pi \i f_0 \beta} \approx z^I(\alpha+\beta)/z^I(\alpha)$.

Next, we move to $z(t)=z^I(t)+z^{\ov{I}}(t)$ and plan to output $\alpha \in [0,T]$ with probability at least $0.5$ such that $|z^{\ov{I}}(\alpha)| \le 0.1 |z^I(\alpha)|$ and $|z^{\ov{I}}(\alpha+\beta)| \le 0.1 |z^I(\alpha+\beta)|$. Because the ``noise'' $z^{\ov{I}}(t)$ has $\|z^{\ov{I}}(t)\|^2_T \ge \epsilon \|z^I(t)\|^2_T$ for a constant $\epsilon$ and the bound $\sqrt{\Delta T}$ in Claim~\ref{cla:z_max} is a polynomial in $k$, the approach for $z^I(t)$ cannot guarantee that $z(\alpha+\beta)/z(\alpha) \approx e^{2 \pi \i f_0 \beta}$ with probability more than $1/2$. 

The key observation is as follows:
\begin{observation}
For a sufficiently small $\epsilon$ and $\|z^{\ov{I}}\|^2_T \le \epsilon \|z\|^2_T$, let $D_T$ be the weighted distribution on $[0,T]$ according to $|z(t)|^2$, i.e., $D_T(t)=\frac{|z(t)|^2}{T \|z\|^2_T}$. If we sample $\alpha \in [0,T]$ from the distribution $D_T$ instead of the uniform distribution on $[0,T]$, $|z^{\ov{I}}(\alpha)| \le 0.01 |z^I(\alpha)|$ with probability $0.9$.
\end{observation}
It follows from the fact that 
\begin{align*}
\E_{\alpha \sim D_T} \frac{|z^{\ov{I}}(\alpha)|^2}{|z(\alpha)|^2}  =\int_{0}^T \frac{|z^{\ov{I}}(\alpha)|^2}{|z(\alpha)|^2} \cdot \frac{|z(\alpha)|^2}{T \|z\|^2_T} \mathrm{d} \alpha= \frac{\int_{0}^T |z^{\ov{I}}(\alpha)|^2 \mathrm{d} \alpha}{T \|z\|^2_T}   \le \epsilon.
\end{align*}

In Procedure \textsc{GetLegal1Sample}, we collect $(\Delta T)^2$ samples (in expectation) $\big(z(\alpha),z(\alpha+\beta)\big)$ in $S_{\heavy}$ with $|z(\alpha)| \ge 0.49 \|z\|_T$ and resample one $\alpha$ from these samples according to their norm $|z(\alpha)|^2 + |z(\alpha+\beta)|^2$. We show its correctness as follows.

Because we do not know $0.5 \|z\|_T$, we use $z_{\emp}$ to approximate it.
\begin{claim}\label{cla:get_empirical_1_energy}
Procedure \textsc{GetEmpirical1Energy} in Algorithm \ref{alg:getempirical1enery_getlegal1sample} takes $O( (T \Delta)^2)$ samples to output $z_{\emp}$ such that $z_{\emp} \in [0.8\|z\|_T, 1.2\|z\|_T]$ with prob. 0.9.
\end{claim} 
\begin{proof}
We know $z_{\emp}^2= \E_{i \in [R_{\est}]}  [|z(\alpha_i)|^2]=  \E_{i \in [R_{\est}]}  [|z^I(\alpha_i)+ z^{\ov{I}}(\alpha_i)|^2 ]$.

Notice that $\E_{i \in [R_{\est}]}  [|z^I(\alpha_i)|^2]$ is in $[0.99 \|z^I\|_T, 1.01 \|z^I\|_T]$ with prob. $0.99$ from the Chernoff bound and Claim \ref{cla:z_max}.

At the same time, $\E_{\alpha_i}[|z^{\ov{I}}(\alpha_i)|^2]=\|z^{\ov{I}}\|^2_T$. With prob. 0.92, $\E_{i \in [R_{\est}]}[|z^{\ov{I}}(\alpha_i)|^2] \le 13 \|z^{\ov{I}}\|^2_T$. For a sufficiently small $\epsilon$ and $\|z^{\ov{I}}\|^2_T \le \epsilon \|z^I\|^2_T$, $\E_{i \in [R_{\est}]}[|z^{\ov{I}}(\alpha_i)|^2] \le 13\epsilon \|z^{I}\|^2_T$.

At last, we bound the cross terms of $|z^I(\alpha_i)+ z^{\ov{I}}(\alpha_i)|^2$ by the Cauchy-Schwartz inequality,
\begin{align*}
~ &\E_{i \in R_{\est}} [ |\overline{z}^I(\alpha_i) z^{\overline{I}}(\alpha_i)| + |z^I(\alpha_i) \overline{z}^{\overline{I}}(\alpha_i)|  ] \\
\le ~ &  2 \E_{i \in R_{\est}} [ |z^I(\alpha_i)| \cdot |z^{\overline{I}}(\alpha_i)|] \\
\le ~ &  2 \left(\E_{i \in [R_{\est}]}  [|z^I(\alpha_i)|^2] \cdot \E_{i \in [R_{\est}]}[|z^{\ov{I}}(\alpha_i)|^2]\right)^{1/2} \\
 \le ~ & 10 \sqrt{\epsilon}  \|z^I\|^2_T.
\end{align*}
 For a sufficiently small $\epsilon$, we have $\E_{i \in [R_{\est}]}[|z(\alpha_i)|^2]^{1/2}$ is in $[0.9 \|z^I\|_T, 1.1 \|z^I\|_T]$, which is also in $[0.8\|z\|_T, 1.2\|z\|_T]$ because of Property \RN{2}.
\end{proof}
We assume $z_{\emp} \in [0.8\|z\|_T, 1.2\|z\|_T]$ and focus on $U=\{t \in [0,T]\big{|}|z(t)| \ge 0.5 z_{\emp}\}$. Notice that 
\[\int_U |z(t)|^2 \mathrm{d} t= \int_{0}^T |z(t)|^2 \mathrm{d} t - \int_{[0,T]\setminus U} |z(t)|^2 \mathrm{d} t \ge (1-0.6^2)\int_{0}^T |z(t)|^2 \mathrm{d} t.\] 

Let $R_{\heavy}=|S_{\heavy}|$. From Claim \ref{cla:z_max} and $\epsilon$, $\E[R_{\heavy}] \ge R_{\repeats}/(T\Delta)$. So we assume $R_{\heavy} \ge 0.01 R_{\repeats}/(T\Delta) = 0.01 (T \Delta)^2$ in the rest of this section and think each $\alpha_i \in S_{\heavy}$ is a uniform sample from $U$ over the randomness on $S_{\heavy}$. 

\begin{claim}
With probability 0.95, $\sum_{i \in S_{\heavy}}(|z^{\ov{I}}(\alpha_i)|^2+|z^{\ov{I}}(\alpha_i+\beta)|^2) \le 10^{-4} \sum_{i \in S_{\heavy}}(|z(\alpha_i)|^2+|z(\alpha_i+\beta)|^2)$ for a sufficiently small $\epsilon$ and $\|z^{\ov{I}}\|^2_T \le \epsilon \|z\|^2_T$. 
\end{claim}
\begin{proof}
At first, 
\begin{equation*}
\E_{S_{\heavy}} \left[ \sum_{i \in S_{\heavy}}(|z(\alpha_i)|^2+|z(\alpha_i+\beta)|^2) \right] \ge R_{\heavy} \cdot \E_{t \sim U}[|z(t)|^2] = R_{\heavy} \cdot \frac{\int_U |z(t)|^2 \mathrm{d} t}{|U|}.
\end{equation*} 

At the same time, 
\begin{equation*}
\E_{S_{\heavy}}\left[\sum_{i \in S_{\heavy}}[|z^{\ov{I}}(\alpha_i)|^2 + |z^{\ov{I}}(\alpha_i+\beta)|^2] \right]= R_{\heavy} \cdot \E_{t \sim U}[|z^{\ov{I}}(t)|^2 + |z^{\ov{I}}(t+\beta)|^2]  \le \frac{2 \int^0_T |z^{\ov{I}}(t)|^2 \mathrm{d} t}{|U|}.
\end{equation*}

From $\int_U |z(t)|^2 \mathrm{d} t \ge 0.64 \int_{0}^T |z(t)|^2 \mathrm{d} t$ and $\int^0_T |z^{\ov{I}}(t)|^2 \mathrm{d} t \le \epsilon \int_{0}^T |z(t)|^2 \mathrm{d} t$, we get the conclusion. 
\end{proof}
We assume all results in the above claims hold and prove that the sample from $S_{\heavy}$ is a good sample such that $z^{\ov{I}}(\alpha)$ is small.
\begin{claim}
If we sample $i \in S_{\heavy}$ according to the weight $|z(\alpha_i)|^2+|z(\alpha_i+\beta)|^2$, with prob. at least $0.9$, $|z^{\ov{I}}(\alpha_i)|+|z^{\ov{I}}(\alpha_i+\beta)| \le 0.05 (|z(\alpha_i)|+|z(\alpha_i + \beta)|)$.
\end{claim}
\begin{proof}
Similar to the proof of the key observation, we compute the expectation of $\frac{|z^{\ov{I}}(\alpha_i)|^2+|z^{\ov{I}}(\alpha_i+\beta)|^2}{|z(\alpha_i)|^2+|z(\alpha_i+\beta)|^2}$ over the sampling in $S_{\heavy}$: 
\begin{eqnarray*}
&&\sum_{i \in S_{\heavy}}\frac{|z(\alpha_i)|^2+|z(\alpha_i+\beta)|^2}{\underset{j \in S_{\heavy}}{\sum}|z(\alpha_j)|^2+|z(\alpha_j+\beta)|^2} \cdot \frac{|z^{\ov{I}}(\alpha_i)|^2+|z^{\ov{I}}(\alpha_i+\beta)|^2}{|z(\alpha_i)|^2+|z(\alpha_i+\beta)|^2} \\
&=& \frac{ \underset{i \in S_{\heavy}}{\sum} |z^{\ov{I}}(\alpha_i)|^2+|z^{\ov{I}}(\alpha_i+\beta)|^2}{ \underset{i \in S_{\heavy}}{\sum}|z(\alpha_i)|^2+|z(\alpha_i+\beta)|^2} \\
&\le& 10^{-4}.
\end{eqnarray*}
By Markov's inequality, when we sample $i \in S_{\heavy}$ according to the weight $|z(\alpha_i)|^2+|z(\alpha_i+\beta)|^2$, $\frac{|z^{\ov{I}}(\alpha_i)|^2+|z^{\ov{I}}(\alpha_i+\beta)|^2}{|z(\alpha_i)|^2+|z(\alpha_i+\beta)|^2} \le 10^{-3}$ with probability $0.9$. We have that with prob. at least $0.9$, $|z^{\ov{I}}(\alpha_i)|+|z^{\ov{I}}(\alpha_i+\beta)| \le 0.05 (|z(\alpha_i)|+|z(\alpha_i + \beta)|)$.
\end{proof}

We assume all above claims hold and finish the proof by setting $\alpha=\alpha_i$. From Claim \ref{cla:close_samples}, we know that 
\[|z^I(\alpha) e^{2\pi \i f_0 \beta}-z^{I}(\alpha+\beta)| \le 0.01 \cdot \E_{t\in [0,T]}[|z^I(\alpha)|^2]^{1/2} \le 0.03 |z^I(\alpha)|.\]

Now we add back the noise $z^{\ov{I}}(\alpha)$ and $z^{\ov{I}}(\alpha+\beta)$ to get \[|z(\alpha)e^{2\pi \i f_0 \beta}-z(\alpha+\beta)| \le |z^I(\alpha)e^{2\pi \i f_0 \beta}-z^I(\alpha+\beta)| + |z^{\ov{I}}(\alpha)| + |z^{\ov{I}}(\alpha+\beta)| \le 0.08 (|z(\alpha)| + |z(\alpha+\beta)|).\]

\subsection{A cluster of frequencies, times $H$, is a one-cluster signal per Definition~\ref{def:one_cluster}}\label{sec:proof_properties}
The goal of this section is to prove Lemma \ref{lem:three_properties}.
Without loss of generality, we assume $g(t)=0$ for any $t \notin [0,T]$ and notice that $\supp(\wh{H} * \wh{x}^*) \subseteq f_0 + [-\Delta, \Delta]$ for $\Delta=\Delta'+\Delta_h$ from the definition of $\wh{H}$. From the Property \RN{6} (presented in Lemma \ref{lem:property_of_filter_H}) of $(H,\wh{H})$, 
\[\int_0^T | x^*(t) |^2 \mathrm{d} t = (1\pm c) \int_{-\infty}^{+\infty} |H(t) \cdot x^*(t) |^2 \mathrm{d} t.\]

From the first two properties of $(H,\wh{H})$, we bound the energy of $g \cdot H$: 
\[ \int_{-\infty}^{+\infty} | H(t) \cdot g(t) |^2 \mathrm{d} t \le (1+c) \int_{0}^{T} |g(t) |^2 \mathrm{d} t.\]
Let $z(t)=(x^*(t) + g(t))H(t)$. We use the triangle inequality on the above two inequalities: 
\begin{align*}
 & \int_{0}^{T} |z(t) |^2 \mathrm{d} t \\
  \ge &~ \int_{0}^{T} | H(t) \cdot x^*(t) |^2 \mathrm{d} t - \int_{0}^{T} | H(t) \cdot g(t) |^2 \mathrm{d} t - 2 \int_{0}^{T} | H(t) \cdot x^*(t) | \cdot | H(t) \cdot g(t) | \mathrm{d} t\\
  \ge &~ (1-c) \int_0^T | x^*(t) |^2 \mathrm{d} t - (1+c) \int_{0}^{T} |g(t) |^2 \mathrm{d} t - 2 \sqrt{(1+c)^2 \int_{0}^{T} |g(t) |^2 \mathrm{d} t \int_{0}^{T} |x^*(t) |^2 \mathrm{d} t\cdot }\\
  \ge &~ \left(1 - 5 \sqrt{c} \right) \int_0^T | x^*(t) |^2 \mathrm{d} t,
\end{align*}
where we use the Cauchy-Schwarz inequality and $\int_{0}^{T} |g(t) |^2 \mathrm{d} t  \le c \int_0^T | x^*(t) |^2 \mathrm{d} t$ in the last step. Similarly, 
\begin{align*}
 &~ \int_{-\infty}^{+\infty} |z(t) |^2 \mathrm{d} t \\
\le  &~ (1+c) \int_0^T | x^*(t) |^2 \mathrm{d} t + (1+c) \int_0^T | g(t) |^2 \mathrm{d} t + 2 \sqrt{(1+c)^2 \int_0^T | x^*(t) |^2 \mathrm{d} t \int_0^T | g(t) |^2 \mathrm{d} t}\\
 \le  &~ (1 + 5\sqrt{c}) \int_0^T | x^*(t) |^2 \mathrm{d} t.
\end{align*}
Hence we obtain Property \RN{2}(in Definition \ref{def:one_cluster}) when $c$ is sufficiently small.

Then we observe that 
\begin{align*}
 & ~\int_{f_0-\Delta_h}^{f_0+\Delta_h} | \widehat{z}(f) |^2 \mathrm{d} f \\
 \ge &~ \int_{f_0-\Delta_h}^{f_0+\Delta_h} | \wh{H\cdot(x^*+g)} |^2 \mathrm{d} f \\
 \ge &~ \int_{f_0-\Delta_h}^{f_0+\Delta_h} |\wh{H\cdot x^*}|^2- | \wh{H \cdot g} |^2 - 2|\wh{H\cdot x^*}|\cdot | \wh{H \cdot g} |\mathrm{d} f \\ 
 \ge &~ \int_{f_0-\Delta_h}^{f_0+\Delta_h} |\wh{H\cdot x^*}|^2 \mathrm{d} f - \int_{-\infty}^{+\infty} | \wh{H \cdot g} |^2 \mathrm{d} f - 2 \sqrt{\int_{f_0-\Delta_h}^{f_0+\Delta_h} |\wh{H\cdot x^*}|^2 \mathrm{d} f \int_{-\infty}^{+\infty} | \wh{H \cdot g} |^2 \mathrm{d} f}\\
 = &~ \int_{-\infty}^{+\infty} |H\cdot x^*|^2 \mathrm{d} t - \int_{-\infty}^{+\infty} | H \cdot g |^2 \mathrm{d} t - 2 \sqrt{\int_{f_0-\Delta_h}^{f_0+\Delta_h} |H\cdot x^*|^2 \mathrm{d} t \int_{-\infty}^{+\infty} | H \cdot g |^2 \mathrm{d} t}\\
 \ge &~ \frac{(1-c) - c(1+c) - 3 \sqrt{c}}{1+5\sqrt{c}} \int_{-\infty}^{+\infty} |z(t) |^2 \mathrm{d} t.
 \end{align*}
Thus we have Property \RN{1}(in Definition \ref{def:one_cluster}) for $z$.



\subsection{Frequency recovery of one-cluster signals}\label{sec:proof_one_frequency_rec}
The goal of this section is prove Theorem \ref{thm:frequency_recovery_1_cluster}. We first show the correctness of Procedure \textsc{Locate1Inner}. Second, we analyze the Procedure \textsc{Locate1Signal}. At end, we rerun Procedure \textsc{Locate1Signal} and use median analysis to boost the constant success probability.\footnote{The proofs in this section are identical to \cite{HIKP12} and \cite{PS15}.}

\begin{lemma}\label{lem:angle_between_xgamma_and_xgammabeta_is_betaf}
  Let $f_0 \in \mathrm{region}(q')$. 
  Let $\beta$ is sampled from $[\frac{st}{4\Delta}, \frac{st}{2\Delta l}]$ and let $\gamma$ denote the output of Procedure \textsc{GetLegal1Sample} in Algorithm \ref{alg:locate1signal_locate1inner_frequencyrecovery1cluster}. Then using the pair of samples $z(\gamma +\beta)$ and $z(\gamma)$, we have
 
 \RN{1}. for the $q'$ with probability at least $1-s$, $v_{q'}$ will increase by one. 
 
 \RN{2}. for any $q$ such that $|q-q'|> 3$, with probability at least $1-15s$, $v_q$ will not increase.

\end{lemma}
\begin{proof}
We replace $f_0$ by $\theta$ in the rest of the proof.
By Lemma \ref{lem:get_legal_1_sample}, we have that for any $\wh{\beta} \leq \beta \leq 2\wh{\beta}$,  Procedure \textsc{GetLegal1Sample} outputs a $\gamma \in [0,T]$ satisfying 
\begin{equation*}
|z(\gamma + \beta) - z(\gamma)e^{2 \pi \i f_0 \beta}| \le 0.1 (|z(\gamma)|+|z(\gamma+\beta)|)
\end{equation*}
with probability at least 0.6. 

 Furthermore, there exists such some constant $g \in (0,1)$ such that with probability $1-g$,
\begin{equation*}
\| \phi( z(\gamma+\beta) ) - ( \phi(z(\gamma)) - 2\pi\beta \theta ) \|_{\bigcirc} \lesssim \sin^{-1} ( \frac{1}{ g} ),
\end{equation*}
where $\| x-y\|_{\bigcirc} = \underset{z\in \mathbb{Z} }{\min}| x-y + 2\pi z|$ denote the ``circular distance'' between $x$ and $y$. 
We can set $s = \Theta(g^{-1})$. 
There exists some constant $p=\Theta(s)$, with probability at least
$1 - p$,
\begin{equation*}
\| o - 2\pi \beta \theta\|_{\bigcirc} < s\pi/2
\end{equation*}
where $ o : =\phi(z(\gamma+\beta)/z(\gamma))$. The above equation shows that $o$ is a good estimate for $2\pi \beta \theta$ with good probability. We will now show that this means the true region $Q_{q'}$ gets a vote with large probability.

For each $q'$ with $\theta \in [l-\frac{\Delta l}{2} + \frac{q'-1}{t} \Delta l, l - \frac{\Delta l}{2} + \frac{q'}{t}\Delta l] \subset [-F,F]$, we have that $\theta_{q'} = l- \frac{\Delta l}{2} + \frac{q'-0.5}{t} \Delta l$ satisfies that
\begin{equation*}
\theta - \theta_{q'} \leq \frac{\Delta l}{2t}.
\end{equation*}
Note that we sample $\beta$ uniformly at random from $[\widehat{\beta}, 2\widehat{\beta}]$, then $ 2\widehat{\beta} = \frac{st}{2\Delta l} \leq \frac{cT}{10 A^\frac{3}{2}}$ (Note that $A$ is some constant $ > 1$), which implies that $2\pi \beta \frac{\Delta l}{2t} \leq \frac{s\pi}{2}$. Thus, we can show the observation $o$ is close to the true region in the following sense,
\begin{eqnarray*}
&&\| o - 2\pi \beta \theta_{q'} \|_{\bigcirc} \\
&\leq& \| o - 2\pi \beta \theta \|_{\bigcirc} + \| 2\pi \beta \theta -  2\pi \beta \theta_{q'} \|_{\bigcirc} \text{~by~triangle~inequality} \\
&\leq & \frac{s\pi}{2} + 2\pi \| \beta \theta - \beta \theta_{q'}\|_{\bigcirc} \\
& \leq & s \pi.
\end{eqnarray*}
Thus, $v_{q'}$ will increase in each round with probability at least $1-s$.

On the other side, consider $q$ with $| q -q'| > 3 $. Then $|\theta - \theta_q| \geq \frac{7\Delta l}{2t}$, and (assuming $\beta \geq \frac{st}{4 \Delta l}$) we have 
\begin{equation*}
2 \pi \beta |\theta - \theta_q| \geq 2\pi \frac{st}{4\Delta l} | \theta - \theta_q | = \frac{s\pi t}{2\Delta l} |\theta - \theta_q |\geq \frac{7s\pi}{4} > \frac{3s\pi}{2}.
\end{equation*}
There are two cases: $| \theta - \theta_q| \leq \frac{\Delta l}{st}$ and $|\theta - \theta_q| > \frac{\Delta l}{st}$.

First, if $|\theta - \theta_q| \leq \frac{\Delta l}{st}$. In this case, from the definition of $\beta$ it follows that
\begin{equation*}
2\pi \beta |\theta - \theta_q | \leq \frac{s \pi t}{\Delta l} |\theta - \theta_q| \leq \pi
\end{equation*}
Combining the above equations implies that
\begin{equation*}
\mathsf{Pr} \bigl[ 2\pi \beta (\theta - \theta_q) \pmod {2\pi } \in  [ -\frac{3s}{4} 2\pi, \frac{3s}{4} 2\pi ] \bigr] =0
\end{equation*}
Second, if $|\theta - \theta_q | > \frac{\Delta l}{st}$. We show this claim is true :
$\mathsf{Pr}[ 2\pi\beta (\theta - \theta_q ) \pmod {2\pi} \in [-\frac{3s}{4} 2\pi, \frac{3s}{4} 2\pi ] ] \lesssim s$. To prove it, we apply Lemma \ref{lem:wrapping}  by setting $\widetilde{T} = 2\pi$, $\widetilde{\sigma} = 2\pi \beta$, $\widetilde{\delta} = 0$, $\epsilon = \frac{3s}{4}2\pi$, $A = 2\pi \widehat{\beta}$, $\Delta f  = | \theta  - \theta_q|$. By upper bound of Lemma \ref{lem:wrapping}, the probability is at most
\begin{equation*}
\frac{ 2\widetilde{\epsilon} }{ \widetilde{T} } + \frac{4 \widetilde{\epsilon} }{A \Delta f} = \frac{3s}{2} + \frac{3s}{\widehat{\beta} \Delta f} \leq \frac{3s}{2} + \frac{3s}{\frac{st}{4 \Delta l} \frac{\Delta  l}{st} } < 15s
\end{equation*}
Then in either case, with probability at least $1-15s$, we have
\begin{equation*}
\| 2\pi \beta \theta_q - 2\pi \beta \theta \|_{\bigcirc} > \frac{3s}{4} 2\pi
\end{equation*}
which implies that $v_q$ will not increase.
\end{proof}

\begin{lemma}\label{lem:locate_1_signal}
Procedure \textsc{Locate1Inner} in Algorithm \ref{alg:locate1signal_locate1inner_frequencyrecovery1cluster} uses $R_{\loc}$ ``legal'' samples, and then after Procedure \textsc{Locate1Signal} in Algorithm \ref{alg:locate1signal_locate1inner_frequencyrecovery1cluster} running Procedure \textsc{Locate1Inner} $D_{\max}$ times, it outputs a frequency $\widetilde{f}_0$ such that
\begin{equation*}
| \widetilde{f}_0 - f_0 | \lesssim \Delta \cdot \sqrt{T\Delta}
\end{equation*}
with arbitrarily large constant probability.
\end{lemma}

\begin{proof}
For each observation, $v_{q'}$ incremented with probability at least $1-p$ and $v_{q}$ is incremented with probability at most $15s+p$ for $|q-q'| > 3$. The probabilities corresponding to different observations are independent. Then after $R_{\loc}$ observations, there exists some constant $c< \frac{1}{2}$, for any $q$ such that $|q-q'|> 3$,
\begin{eqnarray*}
 && \mathsf{Pr}[\text{False~region~gets~more~than~half~votes}] \\
&=&\mathsf{Pr} [ v_{j,q} > R_{\loc} / 2 ]  \\
&\leq & \binom{ R_{\loc}}{R_{\loc}/2 } (15s + p)^{R_{\loc}/2} \\
& \leq & c^{\Omega(R_{\loc})}
\end{eqnarray*}
Similarly, on the other side,
\begin{eqnarray*}
&& \mathsf{Pr} [\text{True~region~gets~less~than~half~votes}] \\
&=&\mathsf{Pr} [ v_{j,q'} < R_{\loc}/2] \\
&\leq & \binom{ R_{\loc}}{R_{\loc}/2 } ( p)^{R_{\loc}/2} \\
&\leq & c^{\Omega(R_{\loc})}
\end{eqnarray*}
Taking the union bound over all the $t$ regions, it gives  with probability at least $1-t f^{\Omega(R_{\loc})}$ we can find some region $q$ such that $|q - q' |< 3$.

If we repeat the above procedure $D_{\max}$ rounds, each round we choose the ``False'' region with probability at most $1-t c^{\Omega(R_{\loc})}$. Thus, taking the union bound over all the $D_{\max}$ rounds, we will report a region has size $\eqsim \Delta \sqrt{\Delta T}$ and contains $f_0$ with probability at least $1- D_{\max} t c^{\Omega(R_{\loc})}$.

The reason for not ending up with region that has size $\eqsim \Delta$ is, the upper bound of the sample range of $\beta$ force us to choose $\beta$ is at most $\lesssim \frac{T}{(\Delta T)^\frac{3}{2}}$ by Claim \ref{cla:close_samples} 

It remains to explain how to set $D_{\max}, t, $ and $R_{\loc}$. At the beginning of the first round, we start with frequency interval of length $2F$, at the beginning of the last round, we start with frequency interval of length $t\cdot \Delta \sqrt{T\Delta}$. Each round we do a $t$-ary search, thus 
\begin{equation*}
D_{\max} = \log_t (\frac{2F}{ t \Delta \sqrt{T\Delta}} ) \leq \log_t (F/\Delta).
\end{equation*}
We can set $R_{\loc} \eqsim \log_{1/c} (t/c)$ and $t > D_{\max}$, e.g. $t = \log (F/\Delta)$. Thus, the probability becomes,
\begin{equation*}
1- D_{\max} t c^{\Omega(R_{\loc})} \geq 1 -t^2 c^{\Omega(R_{\loc})} \geq 1- \poly(1/t,c)
\end{equation*}
which is larger than any constant probability.
\end{proof}

Using the same parameters setting in the proof of Lemma \ref{lem:locate_1_signal}, we show the running time and sample complexity of Procedure \textsc{Locate1Signal},
\begin{lemma}
Procedure \textsc{Locate1Signal}  in Algorithm \ref{alg:locate1signal_locate1inner_frequencyrecovery1cluster} uses \\
$O(\poly(k,\log(1/\delta))) \cdot \log(FT)$ samples and runs in $O(\poly(k,\log(1/\delta))) \cdot \log^2(FT) $ time.
\end{lemma}

\begin{proof}
The number of ``legal'' observations is
\begin{equation*}
D_{\max} R_{\loc} = O( \log_t(F/\Delta) \log_{1/c} (t/c))  = O(\log (F/\Delta))
\end{equation*}
The total number of samples is 
\begin{equation*}
R_{\est}+ R_{\repeats} D_{\max} R_{\loc} = O(T \Delta_h)^2 + (T \Delta_h)^3 \cdot \log (FT)  = \poly(k,\log(1/\delta)) \cdot \log(FT)
\end{equation*}
where the first step follows by Claim \ref{cla:get_empirical_1_energy} and Lemma \ref{lem:get_legal_1_sample} and the last step follows by the setting of $\Delta_h$ in Appendix \ref{sec:parameters_setting_for_filters}.

 The running time includes two parts, one is approximately computing $H(t)$ for all the samples, each sample takes $\poly(k,\log(1/\delta))$ time according to Lemma \ref{lem:approximation_H_in_poly_time}; the other is for each legal sample we need to assign vote to some regions.
 \begin{equation*}
\poly(k,\log(1/\delta)) \cdot (R_{\est} + R_{\repeats} D_{\max} R_{\loc}) + D_{\max}R_{\loc} t = \poly(k,\log(1/\delta)) \log^2(FT)
 \end{equation*}
\end{proof}

Lemma \ref{lem:merge_1_stage} only achieves constant success probability, using median analysis we can boost the success probability,
\begin{lemma}\label{lem:merge_1_stage}
Let $\wt{f_0}$ denote the frequency output by Procedure \textsc{FrequencyRecovery1Cluster} in Algorithm \ref{alg:main_1}, then with probability at least $1-2^{-\Omega(k)}$,
\begin{equation*}
| \wt{f}_0 -f_0 | \lesssim \Delta \sqrt{T\Delta}
\end{equation*}
\end{lemma}
\begin{proof}
Because of Procedure \textsc{FrequencyRecovery1Cluster} taking the median of $O(  k)$ independent results by repeating algorithm \textsc{Locate1Signal} $O(k)$ times. Each sample $L_r$ is close to $\wt{f}_0$ with sufficiently large probability. Thus, using the Chernoff bound will output $\wt{f}_0$ with probability $1-2^{-\Omega(k)}$ such that 
\begin{equation*}
| \wt{f}_0 -f_0 | \lesssim \Delta \sqrt{T\Delta}.
\end{equation*}
\end{proof}
Combining Lemma \ref{lem:merge_1_stage} with the sample complexity and running time in Lemma \ref{lem:locate_1_signal}, we are able to finish the proof of Theorem \ref{thm:frequency_recovery_1_cluster}.

\subsection{The full signal, after multiplying by $H$ and convolving with $G$,  is one-clustered.}
\label{sec:proof_z_satisfies_two_properties}
The goal of this section is to prove Lemma \ref{lem:z_satisfies_two_properties}.
We fix $f^* \in [-F,F]$ satisfying \eqref{eq:heavyfrequency} in this section. We first define a good hashing $(\sigma,b)$ of $f^*$ as follows.
\begin{definition}\label{def:k_signal_recovery_z}
  We say that a frequency $f^*$ is \emph{well-isolated} under the
  hashing $(\sigma, b)$ if, for $j = h_{\sigma, b}(f^*)$, we have
  that the signal
  \[
  \wh{z}^{(j)} = \widehat{x\cdot H} \cdot \widehat{G}^{ (j)}_{\sigma,b}
  \]
  satisfies, over the interval $\overline{I_{f^*}} = (-\infty, \infty) \setminus (f^* -
  \Delta, f^* + \Delta)$,
  \[
  \int_{\overline{I_{f^*}}} \abs{\wh{z}^{(j)}(f)}^2df \lesssim \epsilon \cdot T\N^2/k.
  \]
\end{definition}
For convenience, we simplify $z^{(j)}$ by using $z$ in the rest of this section.
\begin{lemma}\label{lem:often-well-isolated}
  Let $f^*$ be any frequency.  Then $f^*$ is well-isolated by a
  hashing $(\sigma, b)$ with probability $\ge 0.9$ given $B= \Theta(k)$ and $\sigma \in [\frac{1}{B \Delta},\frac{2}{B \Delta}]$ chosen uniformly at random.
\end{lemma}
\begin{proof}
  For any other frequency $f'$ in $x^*$, its contribution in $\wh{z}$ depends on how far
  it is from $f^*$.  Either it is:
  \begin{itemize}
  \item Within $\Delta$ of $f^*$, $f'$ and $f^*$ will be mapped into the same bucket with probability at least $0.99$. 
  \item Between $\Delta$ and $1/\sigma$ far, from Claim \ref{cla:PS15_hash_claims}, $f'$ and $f^*$ will always mapped into different buckets. Hence $f'$ always
    contributes in the $\frac{\epsilon \delta}{ k}$ region of Property \RN{3} in Lemma \ref{lem:property_of_filter_G} about filter function $(G(t), \wh{G}(f))$, i.e., it contributes at most $\frac{\epsilon \delta}{ k} \cdot \int_{f'-\Delta}^{f'+\Delta} |\widehat{x \cdot H}|^2 \mathrm{d} f$. Overall it will contribute \[\frac{\epsilon \delta}{ k} \cdot \int |\widehat{x \cdot H}|^2 \mathrm{d} f=\frac{\epsilon \delta}{ k} \int |x \cdot H|^2 \mathrm{d} t.\]
  \item More than $1/\sigma$ far, in which case they contribute in the
    same region with probability at most $3/B$. By a union bound, it is at most $3k/B \le 0.01$
  \end{itemize}
\end{proof}
Without loss of generality, we assume $\supp(\wh{g \cdot H}) \cap \supp(\wh{x^* \cdot H})=\emptyset$, otherwise we treat it as a part of $x^* \cdot H$. We first consider frequency $f^* \in \wh{x^* \cdot H}$ under $G^{(j)}_{\sigma,b}$.

\begin{lemma}\label{lem:full_proof_of_3_properties_true_for_z} 
Let $f^*$ satisfying $\int_{f^*-\Delta}^{f^*+\Delta} | \widehat{x^*\cdot H}(f) |^2 \mathrm{d} f \geq  T\N^2/k$
 and $\wh{z} = \widehat{x^* \cdot H} \cdot \widehat{G}^{ (j)}_{\sigma,b}$ where $j=h_{\sigma,b}(f^*)$. If $f^*$ is well-isolated, then $z$ and $\wh{z}$ satisfying Property I(in Definition \ref{def:one_cluster}), i.e.,
\[ \int_0^T | z(t) |^2 \mathrm{d} t \geq (1-\epsilon) \int_{-\infty}^{+\infty} |z(t) |^2 \mathrm{d} t.\]
\end{lemma}
\begin{proof}

We first notice that $z(t)=x^*(t) \cdot H(t) * G_{\sigma,b}^{(j)}(t)$ and lower bound $\int_{-\infty}^{+\infty} |z(t) |^2 \mathrm{d} t$ as follows :
\begin{align}\label{eq:xHG_infty_is_at_least_xH_inside_region}
~&\int_{-\infty}^{+\infty} | x^*(t) \cdot H(t) * G_{\sigma,b}^{(j)}(t) |^2 \mathrm{d} t \notag  \\
= \quad &\int_{-\infty}^{+\infty} | \widehat{ x^*\cdot H } (f) \cdot \widehat{G}_{\sigma,b}^{(j)}(f) |^2 \mathrm{d} f &\text{~by~FT} \notag \\
\geq \quad & \int_{f_0-\Delta}^{f_0+\Delta} | \widehat{ x^*\cdot H } (f) \cdot \widehat{G}_{\sigma,b}^{(j)}(f) |^2 \mathrm{d} f \notag \\
\geq \quad&   (1-\delta)^2 \int_{f_0-\Delta}^{f_0+\Delta} | \widehat{ x^*\cdot H } (f) \ |^2 \mathrm{d} f  \notag \\
\geq \quad & (1-\delta)^2 T \N^2/k \notag\\
\geq \quad & 0.9 \frac{\delta }{k}\int_0^T |x^*(t)|^2 \mathrm{d} t
\end{align}
We give an upper bound $\int_{-\infty}^{0} |z(t)|^2 \mathrm{d} t + \int_{T}^{+\infty} |z(t)|^2 \mathrm{d} t \lesssim \epsilon \frac{\delta}{k} \int_0^T |x^*(t) H(t)|^2 \mathrm{d} t$ in the rest of this proof.


Consider the case $t<0$, by definition of Convolution,
\begin{eqnarray*}
z^{(j)}(t) &=& x^*(t) \cdot H(t) * G_{\sigma,b}^{(j)}(t) = \int_{-\infty}^{+\infty} G_{\sigma,b}^{(j)}(t-\tau) \cdot (x^* \cdot H)(\tau) \mathrm{d} \tau
\end{eqnarray*}

Without loss of generality, we can shift the original signal and $H(t)$ from $[0,T]$ to $[-T/2,T/2]$, by Property of $H(t)$, we know that if $s_3T/2\leq  |t| \leq T/2 $, then $H(t)\leq 2^{-O\Theta(\ell)}$.  Note that $G(t)$ is compact and has support $D B$, we also assume its compact region is $[-DB/2, DB/2]$ (Recall that $D=\frac{l}{\alpha \pi}$). 

Thus, by definition of convolution,
\begin{align*}
& ~z(t) \\
=  &~ \int_{-DB\sigma/2}^{DB\sigma/2} {G}^{(j)}_{\sigma,b}(s) \cdot (x \cdot H)(t-\tau) \mathrm{d} \tau \\
=  &~ \frac{1}{\sigma} \int_{-DB\sigma/2}^{DB\sigma/2} {G}(s/\sigma) e^{2\pi\i s(j/B-\sigma b)/\sigma} \cdot (x \cdot H)(t-\tau) \mathrm{d} \tau \\
\leq &~ \frac{1}{\sigma} \int_{-DB\sigma/2}^{DB\sigma/2} | {G}(\tau/\sigma)| \cdot |  (x \cdot H)(t-\tau)| \mathrm{d} \tau \\
\leq &~ \left(\frac{1}{\sigma} \int_{-DB\sigma/2}^{DB\sigma/2} | {G}(\tau/\sigma)| \mathrm{d} \tau \right) \cdot \left( \underset{|\tau|\leq DB\sigma/2}{\max} | (x\cdot H) (t-\tau)|  \right)
\end{align*}
So, if $t\notin [-T/2,T/2]$, then $t-s \notin [-T/2+DB\sigma/2, T/2-DB\sigma/2]$. By Property \RN{5} of $G(t)$, $|G(t)|\leq \poly(k,\log(1/\delta))$. Because of the parameter setting\footnote{We will set $B$ to be $O(k)$, $D$ to be $\poly(k)$ and $\sigma$ to be $T/\poly(k)$.}, we have the fact $[-T s_3/2, T s_3 / 2] \subseteq [-T/2+DB\sigma/2, T/2-DB\sigma/2] \subseteq [-T/2,T/2]$. Thus, we know $ T (1-s_3)/2> DB\sigma/2$, then for any $t-\tau \in [ -T/2, -T/2+DB\sigma/2] \cup [T/2-DB\sigma/2,T/2] =S$, then
 \begin{equation*}
|z(t)|^2 \lesssim \bigl( DB\sigma \cdot \frac{1}{\sigma} \cdot \poly(k,\log(1/\delta)) \bigr)^2 \cdot 2^{-\Theta(\ell)} \cdot k^4  \cdot \| x^*(t) \|_T^2 \lesssim \poly(k,\log(1/\delta) ) \cdot 2^{-\Theta(\ell)} \cdot \| x^*(t) \|_T^2.
 \end{equation*} 

 Thus, taking the integral over $S$,
 \begin{equation*}
\int_S |z(t)|^2 \mathrm{d} t \lesssim  |S| \cdot 2^{-\Theta(\ell)} \poly(k,\log(1/\delta) ) \cdot \| x^*(t) \|^2 \lesssim 2^{-\Theta(\ell)} T \| x^*(t) \cdot H(t)\|_T^2
 \end{equation*}
By property of filter function $H(t),\widehat{H}(f)$, we have
\begin{equation*}
|(x\cdot H)(t) |^2 \leq ( \frac{t}{T} )^{-\ell} \| x^*(t) \cdot H(t) \|_T^2 \text{~if~}t\geq 3T
\end{equation*}
Thus for any constant $\epsilon$,
\begin{equation}\label{eq:xHG_outside_T_is_at_most_xH_inside_region}
\int_{-\infty}^{-T/2} |z(t)|^2 \mathrm{d} t + \int_{T/2}^{+\infty} |z(t)|^2 \mathrm{d} t \lesssim 2^{-\ell} T\| x^*(t) \cdot H(t) \|_T^2 \leq 0.9 \epsilon \cdot \frac{\delta}{k} \int_{-T/2}^{T/2} |x^*(t)|^2 \mathrm{d} t
\end{equation}
where the last inequality follows by $\ell \gtrsim k\log (k/\delta)$.
Shifting the interval from $[-T/2,T/2]$ to $[0,T]$, the same result is still holding. Combining Equation (\ref{eq:xHG_infty_is_at_least_xH_inside_region}) and (\ref{eq:xHG_outside_T_is_at_most_xH_inside_region}) completes the proof of Property \RN{2}.

\end{proof}

We consider frequency $f^* \in \wh{g \cdot H}$ under $G^{(j)}_{\sigma,b}$ and show the energy of noise $g(t)$ is evenly distributed over $B$ bins on expectation.
\begin{lemma}\label{lem:noise_hashing}
Given any noise $g(t) : [0,T]\rightarrow \mathbb{C}$ and $g(t)=0, \forall t \notin [0,T]$. We have, $\forall j\in [B],$
\begin{equation*}
\underset{\sigma,b}{\mathbb{E}} \left[ \int_{-\infty}^{+\infty} | g(t) H(t) * G_{\sigma,b}^{(j)}(t) |^2 \mathrm{d} t \right] \lesssim \frac{1}{B} \int_{-\infty}^{+\infty} |g(t) H(t)|^2 \mathrm{d} t 
\end{equation*}
\end{lemma}

\begin{proof}
Because of Fourier Transform preserves $\ell_2$ norm, it suffices to prove
\begin{equation*}
\underset{\sigma,b}{\mathbb{E}} \left[ \int_{-\infty}^{+\infty} | \widehat{ g \cdot H}(f) \cdot \widehat{G}^{(j)}_{\sigma,b}(f)  |^2 \mathrm{d} f \right] \lesssim \frac{1}{B} \int_{-\infty}^{+\infty} | \widehat{ g\cdot H} (f)|^2 \mathrm{d} f
\end{equation*}
Since $\widehat{G}_{\sigma,b}^{(j)}(f)$ is a periodic function  and outputs at most $1$ on $O(1/B)$ fraction of the period, and outputs $\le \delta$ on other part. Thus, for any frequency $f$, we have
\begin{equation*}
\underset{\sigma,b}{ \mathbb{E}} \left[  | \widehat{G}_{\sigma,b}^{(j)} (f) |^2 \right] \lesssim \frac{1}{B}
\end{equation*}
Thus, we have
\begin{eqnarray*}
&&\underset{\sigma,b}{\mathbb{E}} \left[ \int_{-\infty}^{+\infty} | \widehat{ g \cdot H}(f) \cdot \widehat{G}_{\sigma,b}^{(j)}(f)  |^2 \mathrm{d} f \right] \\
& \leq &\underset{\sigma,b}{\mathbb{E}} \left[ \int_{-\infty}^{+\infty} | \widehat{ g \cdot H}(f) |^2 \cdot |\widehat{G}^{(j)}_{\sigma,b}(f)  |^2 \mathrm{d} f \right] \\
& = &  \int_{-\infty}^{+\infty} | \widehat{ g \cdot H}(f) |^2 \cdot  \underset{\sigma,b}{\mathbb{E}} [|\widehat{G}_{\sigma,b}^{(j)}(f)  |^2] \mathrm{d} f  \\
&\leq &\int_{-\infty}^{+\infty} | \widehat{ g \cdot H}(f) |^2  \mathrm{d} f \cdot \underset{f}{\max }\left[\underset{\sigma,b}{\mathbb{E}} \left|\widehat{G}_{\sigma,b}^{(j)}(f) \right|^2 \right] \\
& \lesssim & \frac{1}{B} \int_{-\infty}^{+\infty} |  \widehat{ g \cdot H}(f) |^2 \mathrm{d} f,
\end{eqnarray*}
which completes the proof.
\end{proof}

\begin{proofof}{Lemma \ref{lem:z_satisfies_two_properties}}
  Let $j = h_{\sigma, b}(f^*)$, signal 
  \begin{equation}\label{eq:def_z}
  \wh{z} = \widehat{x\cdot H} \cdot \widehat{G}^{ (j)}_{\sigma,b},
  \end{equation}
  and region $I_{f^*} = (f^* - \Delta, f^* + \Delta)$ with complement
  $\overline{I_{f^*}} = (-\infty, \infty)\setminus I_{f^*}$.  
  From Property I of $G$ in Lemma \ref{lem:property_of_filter_G}, we have that
  \[
  \widehat{G}^{ (l)}_{\sigma,b}(f) \gtrsim 1
  \]
  for all $f \in I_{f^*}$, so by~\eqref{eq:heavyfrequency}
  \begin{equation*}
    \int_{I_{f^*}} | \widehat{z}(f) |^2 \mathrm{d} f \geq  T\N^2/k.
  \end{equation*}
  On the other hand, $f^*$ is will-isolated with probability $0.9$:
  \begin{equation*}
    \int_{\overline{I_{f^*}}} | \widehat{z}(f) |^2 \mathrm{d} f \lesssim  \epsilon T\N^2/k.
  \end{equation*}
  Hence, $\wh{z}$ satisfies the
  Property I(in Definition \ref{def:one_cluster}) of one-mountain recovery. Combining Lemma~\ref{lem:full_proof_of_3_properties_true_for_z} and Lemma~\ref{lem:noise_hashing}, we know that $(x^*\cdot H)*G^{(j)}_{\sigma,b}$ always satisfies Property \RN{2}(in Definition \ref{def:one_cluster}) and $\int_{-\infty}^{+\infty} | g(t) H(t) * G_{\sigma,b}^{(j)}(t) |^2 \mathrm{d} t$ is less than $20 T \N^2/B \le \epsilon T \N^2/k$ with probability at least $0.95$, which indicates that $z=(x^*+g)\cdot H * G^{(j)}_{\sigma,b}$ satisfies Property \RN{2}(in Definition \ref{def:one_cluster}) with probability $0.95$. \end{proofof}

\subsection{Frequency recovery of \texorpdfstring{$k$}{k}-clustered signals}\label{sec:k_mountain_frequency_recovery}
The goal of this section is to prove that the frequencies found by Procedure \textsc{FrequencyRecoveryKCluster} in Algorithm \ref{alg:main_k} have some reasonable guarantee.

We first notice that Lemma~\ref{lem:z_satisfies_two_properties} and Lemma~\ref{lem:findfrequency} imply the following 
lemma by a union bound.
\begin{lemma}\label{lem:one-round-full}
  Let $x^*(t) = \overset{k}{ \underset{j=1}{\sum} } v_j e^{2\pi\i f_j
    t}$.  We observe $x(t)= x^*(t) +g(t)$, where $\|g(t) \|_T^2 \le
  c\|x^*(t)\|_T^2$ for a sufficiently small constant $c$ and define $\N^2 :
  = \| g(t) \|_T^2 + \delta \| x^*(t) \|_T^2$. Then
  Procedure \textsc{OneStage}
  returns a set $L$ of $O(k)$
  frequencies that covers the heavy frequencies of $x^*$.  In
  particular, for any $f^*$ with
  \begin{equation}\label{eq:heavyband}
    \int_{f^*-\Delta}^{f^*+\Delta} | \widehat{x\cdot H}(f) |^2 \mathrm{d} f \geq  T\N^2/k,
  \end{equation}
  there will exist an $\wt{f} \in L$ satisfying $|f^*-\widetilde{f} |
  \lesssim \sqrt{T \Delta} \cdot \Delta T$ with probability $0.99$.
\end{lemma}

\begin{lemma}
  Let $x^*(t) = \overset{k}{ \underset{j=1}{\sum} } v_j e^{2\pi\i f_j
    t}$ and $R=O(  k)$.  We observe $x(t)= x^*(t) +g(t)$, where $\|g(t) \|_T^2 \le
  c\|x^*(t)\|_T^2$ for a sufficiently small constant $c$ and choose $\N^2 :
  = \| g(t) \|_T^2 + \delta \| x^*(t) \|_T^2$. Then
  Algorithm $\textsc{MultipleStages}$
  returns a set $L$ of $O(k)$
  frequencies that approximates the heavy frequencies of $x^*$.  In
  particular, with probability $1-2^{-\Omega(k)}$, for any $f^*$ such that
  \begin{equation}
    \int_{f^*-\Delta}^{f^*+\Delta} | \widehat{x\cdot H}(f) |^2 \mathrm{d} f \geq  T\N^2/k,
  \end{equation}
  there will exist an $\wt{f} \in L$ satisfying $|f^*-\widetilde{f} |
  \lesssim \sqrt{T\Delta} \Delta$.
\end{lemma}
\begin{proof}
  Let $A \subset [-F, F]$ denote the set of frequencies $f^*$
  satisfying Equation~\eqref{eq:heavyband}.  Let $A'\subset [-F, F]$ denote a
  net of $A$ of distance $2\Delta$, so the intervals used
  in Equation~\eqref{eq:heavyband} for each $f^* \in A'$ are disjoint.  Then
  \[
  \abs{A'} \leq 2k + k =3k
  \]
  because each frequency in $x^*$ contributes to at most two of the
  intervals, and the total mass of $\wh{g}$ is at most $k$ times the
  threshold $T\N^2$.

  Let $L_1, \dotsc, L_R$ be the results of $R$ rounds of
  Algorithm \textsc{OneStage}.
  We say that a frequency $f
  \in A'$ is \emph{successfully recovered in round $r$} if there
  exists an $\wt{f} \in L_r$ such that $\abs{f - \wt{f}} \leq
  \Delta_a$, where 
  \[ 
  \Delta_a = \Delta \sqrt{T \Delta} \lesssim \sqrt{T\Delta} \Delta.
  \]  By
  Lemma~\ref{lem:one-round-full}, each frequency is successfully
  recovered with $0.8$ probability in each round.  Then by the Chernoff
  bound, with $1 - 2^{-\Omega(k)}$ probability, every $f \in A'$ will be
  successfully recovered in at least $0.6 R$ rounds.

  Then, by Lemma~\ref{lem:merging}, we output a set $L$ of $O(B)$
  frequencies such that every $f \in A'$ is within $\Delta_a$ of some
  $\wt{f} \in L$.  Hence every $f \in A$ is within $2\Delta_a$ of some
  $\wt{f} \in L$.
\end{proof}

\begin{lemma}\label{lem:merging}
  Let $L_1, \dotsc, L_R$ by sets of frequencies and $f^*$ be any
  frequency.  Then $L = \textsc{MergedStages}(L_1,$ $\dotsc,$ $L_R)$ is a set of
  $2\frac{\sum \abs{L_r}}{R}$ frequencies satisfying
  \[
  \min_{\wt{f} \in L} \abs{f^* - \wt{f}} \leq \median_{r \in [R]} \min_{f \in L_r} \abs{f^* - f}.
  \]
\end{lemma}
\begin{proof}
  The algorithm is to take the union, sort, and take every
  $\frac{R}{2}$th entry of the sorted list.

  Let $\Delta = \median_{r \in [R]} \min_{f \in L_r} \abs{f^* - f}$.
  We have that at least $R/2$ different $f \in \bigcup_r L_r$ lie
  within $\Delta$ of $f^*$.  This set forms a sequential subsequence
  of the sorted list of frequencies, so our output will include one.
\end{proof}

\subsection{Time and sample complexity of frequency recovery of $k$-clustered signals}
The goal of this section is to show that Procedure \textsc{FrequencyRecoveryKCluster} takes\\ $\poly(k,\log(1/\delta)) \log(FT)$ samples, and runs in $\poly(k,\log(1/\delta)) \log^2(FT)$ time.

In order to analyze the running time and sample complexity. We need to extend the one-cluster version Procedure \textsc{GetLegal1Sample} and \textsc{GetEmpirical1Energy} (in Algorithm \ref{alg:getempirical1enery_getlegal1sample}) to $k$-cluster version \textsc{GetLegalKSample} and \textsc{GetEmpiricalKEnergy}(in Algorithm \ref{alg:getempiricalkenergy_getlegalksample_onestage})\footnote{We omitted the proofs here, because the proofs are identical to the one-cluster situation.},
\begin{lemma}\label{lem:get_legal_k_sample}
Procedure \textsc{GetLegalKSample} in Algorithm \ref{alg:getempiricalkenergy_getlegalksample_onestage} runs Procedure \textsc{HashToBins} $R_{\repeats} = O((T \Delta)^3)$ times  to output two vectors $\wh{v}, \wh{v}'\in \mathbb{C}^B$ such that, for each $j\in [B]$,
\begin{equation*} 
|\wh{v}_j - \wh{v}'_j e^{2 \pi \i f_j \beta}| \le 0.08 (|\wh{v}_j |+|\wh{v}'_j|),
\end{equation*}
holds with probability at least 0.6.
\end{lemma}

Using the definition of $z$ in Definition \ref{def:k_signal_recovery_z}.
\begin{claim}\label{cla:get_empirical_k_energy}
Procedure \textsc{GetEmpiricalKEnergy} in Algorithm \ref{alg:getempiricalkenergy_getlegalksample_onestage} runs Procedure \textsc{HashTobins} $R_{\est}O( (T \Delta)^2 )$ times to output a vector $z_{\emp}\in \mathbb{R}^B$ such that, for each $j\in [B]$,
\begin{equation*}
z_{\emp}^j  \in [0.8\|z^{(j)} \|_T, 1.2\|z^{(j)}\|_T],
\end{equation*}
holds with probability at least 0.9.
\end{claim}


\begin{claim}\label{cla:locate_k_signal}
Algorithm \textsc{LocateKSignal} in Algorithm \ref{alg:locateksignal_locatekinner_hashtobins} uses $O(\poly(k,\log(1/\delta)) \cdot \log(FT) )$, and runs in $O(\poly(k,\log(1/\delta)) \cdot \log^2(FT) )$. 
\end{claim}
\begin{proof}
We first calculate the number of samples. All the samples is basically all the Fourier samples, each time needs $B\log(k/\delta)$. In total it calls $\textsc{HashToBins}$ $O(R_{\est} + R_{\repeats} D_{\max} R_{\loc})$ times where $D_{\max} R_{\loc} = \Theta(\log(FT) )$ by similar analysis as one-cluster frequency recovery. Thus, the total number of samples is 
\begin{eqnarray*}
(R_{\est} + R_{\repeats} D_{\max} R_{\loc}) B \log(k/\delta)  = \poly(k,\log(1/\delta)) \cdot \log(FT).
\end{eqnarray*}

Then, we analyze the running time.

The expected running time includes the following parts: the first part is running Procedure \textsc{HashToBins} $O(R_{\est} + R_{\repeats} D_{\max} R_{loc})$ times, each run takes $O(B \log(k/\delta) + B\log B)$ samples. For each such sample we need $\poly(k,\log(1/\delta))$ time to compute $H(t)$ according to Lemma \ref{lem:approximation_H_in_poly_time} and there are $\poly(k,\log(1/\delta)) \log(FT))$ many samples; the second part is updating the counter $v$,which takes $O(D_{\max}R_{loc} Bt)$ time. Thus, in total
\begin{align*}
 &\poly(k,\log(1/\delta)) \cdot O(R_{\est} + R_{\repeats} D_{\max} R_{\loc}) \cdot O(B \log(k/\delta) + B\log B) + O(D_{\max}R_{loc} Bt) \\
= &\poly(k,\log(1/\delta)) \cdot \log^2(FT),
\end{align*}
where by similar analysis as one-cluster recovery, $t= \Theta(\log(FT))$ and $D_{\max} R_{\loc} = \Theta(\log(FT))$.
\end{proof}

To boost the success probability, Procedure \textsc{MultipleStages} reruns Procedure \textsc{LocateKSignal} $O(k)$ times. At the end, Procedure \textsc{FrequencyRecoveryKCluster} combining Procedure \textsc{MultipleStages} and \textsc{MergedStages} directly, and the running time and sample complexity of \textsc{MultipleStages} are dominating \textsc{MergedStages}. Thus we have
\begin{lemma}
Procedure \textsc{FrequencyRecoveryKCluster} in Algorithm \ref{alg:main_k} uses $O(\poly(k,\log(1/\delta)) \cdot \log(FT) )$, and runs in $O(\poly(k,\log(1/\delta)) \cdot \log^2(FT) )$.
\end{lemma}

\section{One-cluster Signal Recovery}\label{sec:one_cluster_recovery}

\subsection{Overview}
In this section, we consider $x^*$ whose frequencies in $\widehat{x}^{*}$ are in the range $[f_0 - \Delta', f_0 + \Delta']$ for some frequency $f_0$ and $\Delta'>0$ and provide an algorithm to approximate it by a polynomial. 

We fix $T$ in this section and recall that $\langle f(t),g(t) \rangle_T := \frac{1}{T} \int_0^T  f(t) \overline{g}(t) \mathrm{d}t $ such that $\|e^{2\pi \i f_i t}\|_T=\sqrt{\langle e^{2\pi \i f_i t},  e^{2\pi \i f_i t} \rangle_T} =1$. For convenience, given $\overset{k}{ \underset{j=1}{ \sum } } v_j e^{2\pi \i f_j t}$, we say the frequency gap of this signal is $ \underset{i\neq j}{\min}|f_i-f_j|$.

For simplicity, we first consider frequencies clustered around 0. The main technical lemma in this section is that any signal $x^*$ with bounded frequencies in $\widehat{x}^{*}$ can be approximated by a low-degree polynomial on $[0,T]$.

\restate{lem:low_degree_approximates_concentrated_freq}

One direct corollary is that when $\widehat{x}^{*}$ are in the range $[f_0+\Delta', f_0+\Delta']$, we can approximate $x^*$ by $P(t) \cdot e^{2 \pi \i f_0 t}$ for some low degree polynomial $P$. 

We give an overview of this section first.
We first show some technical tools in Section \ref{sec:bounding_det_of_gram}, \ref{sec:bounded_and_unboundedly}. In Section \ref{sec:existence_of_bounded_gap_signal}, using those tools, we can show for any $k$-Fourier-sparse signal, there exists another $k$-Fourier-sparse signal with bounded frequency gap close to the original signal.  In Section \ref{sec:transfer_bounded_gap_signal_to_polynomial}, we show that for any $k$-Fourier-sparse signal with bounded frequency gap, then there exists a low degree polynomial close to it. In Section \ref{sec:transfer_polynomial_to_fourier}, we show how to transfer low degree polynomial back to a Fourier-sparse signal. Combining all the above steps finishes the proof of Lemma \ref{lem:low_degree_approximates_concentrated_freq}.

We apply Theorem \ref{thm:frequency_recovery_1_cluster} of frequency estimation on $x^*$ to obtain an estimation $\wt{f_0}$ of $f_0$ and use Theorem~\ref{thm:accurate_poly_learning} on the approximation $Q(t)e^{2\pi \i \wt{f_0}t}$ of $x^*$ to recover the signal. We summarize this result as follows.
\begin{theorem}[One-cluster Signal Recovery]\label{thm:cft1cluster}
Let $x^*(t) = \overset{k}{ \underset{j=1}{\sum}} v_j e^{2\pi \i f_j t}$ where $\forall j \in [k], |f_j -f_0|\leq \Delta$ and $x(t) = x^*(t) + g(t)$ be our observable signal. For any $\delta>0$ and any $T>0$, let $\N^2 : = \|g\|^2_T + \delta \|x^*\|_T^2$.  Procedure \textsc{CFT1Culster} in Algorithm \ref{alg:main_1} finds a polynomial $P(t)$ of degree at most $d=O \left( (T \Delta_h + T \Delta)^{1.5}  + k^3 \log k + k \log 1/\delta \right)$ and a frequency $\widetilde{f}_0$ such that
\begin{equation}
\| P(t) \cdot e^{2\pi \i \widetilde{f}_0 t} - x^*(t)  \|^2_T \lesssim \N^2
\end{equation}
The algorithm uses $O(kd) + \poly(k,\log(1/\delta)) \log(FT) $ samples, run in $O(k d^\omega ) + \poly(k, \log(1/\delta))\log^2(FT)$ time, and succeeds with probability at least $1-2^{-\Omega(k)}$.
\end{theorem}

\begin{proof}
We apply the algorithm in Theorem \ref{thm:frequency_recovery_1_cluster} to obtain an estimation $\widetilde{f}_0$ with $\poly(k)\log(FT)$ samples and $\poly(k)\log^2(FT)$ running time such that $|\widetilde{f}_0 - f_0| \lesssim (\Delta_h+ \Delta) \sqrt{T (\Delta_h+ \Delta)}$ holds with probability at least $1-2^{-\Omega(k)}$. Notice that $|f_j-\widetilde{f}_0| \le |f_j-f_0| + |\widetilde{f}_0 - f_0|\lesssim (T (\Delta_h+ \Delta))^{1.5}$.


We consider $x'(t)=e^{-2\pi \i \widetilde{f}_0 t}x(t)=\overset{k}{ \underset{j=1}{\sum} } v_j e^{2\pi \i (f_j-\widetilde{f}_0) t}$ . By Lemma \ref{lem:low_degree_approximates_concentrated_freq}, there exists a polynomial $P(t)$ of degree at most \begin{equation*}
d=O\left((T \Delta_h + T \Delta)^{1.5} + k^3 \log k + k \log 1/\delta\right)
\end{equation*} such that it approximates $x'$ by
\begin{equation*}
\|P(t)-x'(t)\|_T \le \frac{\delta}{4} \|x'(t)\|_T = \frac{\delta}{4} \|x^*(t)\|_T.
\end{equation*}
which indicates $\|Q(t) -e^{-2 \pi \i \widetilde{f}_0 t} \cdot x^*(t) \|_T \le \frac{\delta}{4} \|x^*(t) \|_T$.

Because we can sample $x(t)$, we can also sample  $e^{-2 \pi \i \widetilde{f}_0 t} \cdot x(t)=Q(t) + g'(t)$ for $g'(t)=e^{-2 \pi \i \widetilde{f}_0 t} \cdot g(t) + (e^{-2 \pi \i \widetilde{f}_0 t} \cdot x^*(t)-Q(t))$. Hence we apply the algorithm in Theorem \ref{thm:accurate_poly_learning} and choose $R=O(k)$ in that proof. Then Procedure \textsc{RobustPolynomialLearning$^+$} takes $O(kd)$ samples and $O(kd^\omega)$ time to find a degree $d$ polynomial $P(t)$ approximating $Q(t)$  such that 
\begin{equation*}
\|P(t)-Q(t)\|_T \lesssim  \|g'(t)\|_T,
\end{equation*} 
holds with probability at least $1-2^{-\Omega(k)}$. It indicates \[\|P(t)-e^{-2 \pi \i \widetilde{f}_0 t} \cdot x^*(t)\|_T \lesssim \|P(t) - Q(t)\|_T + \|Q(t) - x^*(t)\| \lesssim \delta \|x^*(t)\|_T + \|g(t)\|_T \eqsim \N.\] Therefore we know $\|e^{2 \pi \i \widetilde{f}_0 t}\cdot P(t)-x^*(t)\|_T^2 \lesssim \N^2$.
\end{proof}


\subsection{Bounding the Gram matrix determinant}\label{sec:bounding_det_of_gram}
We define Gram matrix for $e^{2\pi \i f_1 t}, e^{2\pi \i f_2 t}, \cdots, e^{2\pi \i f_k t}$ and provide lower/upper bounds for its determinant. 
\begin{definition}[Gram matrix]\label{def:gram_matrix}
We define $\Gram_{f_1,\cdots,f_k}$ to be
\begin{equation*}
\begin{bmatrix}
\langle e^{2\pi \i f_1 t}, e^{2\pi \i f_1 t} \rangle_T & \langle e^{2\pi \i f_1 t}, e^{2\pi \i f_2 t} \rangle_T  & \cdots & \langle e^{2\pi \i f_1 t}, e^{2\pi \i f_k t} \rangle_T \\
\langle e^{2\pi \i f_2 t}, e^{2\pi \i f_1 t} \rangle_T & \langle e^{2\pi \i f_2 t}, e^{2\pi \i f_2 t} \rangle_T  & \cdots & \langle e^{2\pi \i f_2 t}, e^{2\pi \i f_k t} \rangle_T \\
\cdots & \cdots & \cdots & \cdots \\
\langle e^{2\pi \i f_k t}, e^{2\pi \i f_1 t} \rangle_T & \langle e^{2\pi \i f_k t}, e^{2\pi \i f_2 t} \rangle_T  & \cdots & \langle e^{2\pi \i f_k t}, e^{2\pi \i f_k t} \rangle_T \\
\end{bmatrix}
\end{equation*}
Note that the above matrix is a Hermitian matrix with complex entries, thus both its determinant and all eigenvalues are in $\mathbb{R}$.
\end{definition}
We defer the proof of the following Theorem to Appendix \ref{sec:bounds_gram_determinant}.

\define{thm:bounds_gram_determinant}{Theorem}{
For real numbers $\xi_1,\ldots,\xi_k$, let $G_{\xi_1,\ldots,\xi_k}$ be the matrix whose $(i,j)$-entry is
$$
\int_{-1}^1 e^{2\pi \i(\xi_i-\xi_j)t}dt.
$$
Then
$$
\det(G_{\xi_1,\ldots,\xi_k}) = 2^{\tilde O(k^2)}\prod_{i<j} \min(|\xi_i-\xi_j|^2,1).
$$}
\state{thm:bounds_gram_determinant}

We use the following corollary in this section.
\begin{corollary}\label{cor:determinant_corollary}
There exists a universal constant $\alpha>0$ such that, for any $T>0$ and real numbers $f_1,\cdots,f_k$, the $k \times k$ Gram matrix of $e^{2\pi \i f_1 t}, e^{2\pi \i f_2 t}, \cdots, e^{2\pi \i f_k t}$ whose $(i,j)$-entry is
\[\Gram_{f_1,\cdots,f_k}(i,j) = \langle e^{2\pi \i f_i t}, e^{2\pi \i f_j t} \rangle_T=\frac{1}{T}\int_{0}^T e^{2\pi \i(f_i-f_j)t}\mathrm{d} t.\] 
satisfies 
\[k^{-\alpha k^2} \prod_{i<j} \min((|f_i-f_j|T)^2,1) \le \det\left(\Gram_{f_1,\cdots,f_k} \right)\le k^{\alpha k^2} \prod_{i<j} \min((|f_i-f_j|T)^2,1).\]
\end{corollary}
Based on Corollary \ref{cor:determinant_corollary}, we show the coefficients of a $k$-Fourier-sparse signal can be upper bounded by the energy $\|x\|^2_T$.
\restate{lem:relation_energy_coef}
\begin{proof}
Let $\vec{v_i}$ denote the vector $e^{2 \pi \i f_i t}$ and $V=\{\vec{v_1},\cdots,\vec{v_k}\}$. Notice that $\|\vec{v_i}\|_T^2=\langle \vec{v_i},\vec{v_i}\rangle=1$. For each $\vec{v_i}$, we define $\vec{v}_i^{\parallel}$ to be the projection of $\vec{v_i}$ into the linear subspace $\mathrm{span}\{V\setminus \vec{v_i}\}=\mathrm{span}\{\vec{v}_1,\cdots,\vec{v}_{i-1},\vec{v}_{i+1},\cdots,\vec{v}_k\}$ and $\vec{v}_i^{\perp}=\vec{v}_i-\vec{v}_i^{\parallel}$ which is orthogonal to $\mathrm{span}\{V\setminus \vec{v_i}\}$ by the definition.

Therefore from the orthogonality, 
\[\|x(t)\|_T^2 \ge \max_{j \in [k]} \{|v_j|^2 \cdot \|\vec{v}_j^{\perp}\|_T^2\} \ge \frac{1}{k}  \overset{k}{\underset{j=1}{\sum}} |v_j|^2 \cdot \|\vec{v}_j^{\perp}\|_T^2 .\] 
It is enough to estimate $\|\vec{v}_j^{\perp}\|_T^2$ from Claim \ref{cla:orthogonal_distance}:
\[ \|\vec{v}_j^{\perp}\|_T^2=\frac{\det(\Gram(V))}{\det(\Gram(V \setminus \vec{v}_i))} \ge k^{-2\alpha k^2} \prod_{j\neq i} \min\left((f_j-f_i)T,1\right)^2 \ge k^{-2\alpha k^2} (\eta T)^{2k-2},\] 
where we use Corollary \ref{cor:determinant_corollary} to lower bound it in the last step.
\end{proof}

\subsection{Perturbing the frequencies does not change the subspace much}\label{sec:bounded_and_unboundedly}
We show that for a $k$-Fourier-sparse signal with unboundedly close frequency gap, there always exists another $k$-Fourier-sparse signal with slightly separated gap.
\begin{lemma}[Slightly Shifting one Frequency]\label{lem:shifting_one_frequency_preserve_the_energy}
There is a universal constant $C_0>0$ such that for any $x(t) = \overset{k} { \underset{j=1}{ \sum } } v_j e^{2\pi \i f_j t}$ and any frequency $f_{k+1}$, there always exists 
\[x'(t) = \overset{ k-1}{ \underset{j=1 }{ \sum} } v'_j e^{2\pi \i f_j t} + v'_{k+1} e^{2\pi\i f_{k+1} t}\] 
with $k$ coefficients $v'_1, v'_2, \cdots, v'_{k-1}, v_{k+1}'$ satisfying 
\begin{equation*}
\|x'(t)-x(t)\|_T \leq k^{C_0 k^2} \cdot \left(|f_k-f_{k+1}|T\right) \cdot \|x(t)\|_T
\end{equation*}
\end{lemma}
\begin{proof}
We abuse the notation $e^{2\pi \i f_j t}$ to denote a vector in the linear subspace. We plan to shift $f_k$ to $f_{k+1}$ and define
\begin{eqnarray*}
V & = & \{ e^{2\pi \i f_1 t}, \cdots,e^{2\pi \i f_{k-1} t}, e^{2\pi \i f_k t} \} \\
V'& = & \{ e^{2\pi \i f_1 t}, \cdots,e^{2\pi \i f_{k-1} t}, e^{2\pi \i f_{k+1} t} \} \\
U & = & \{ e^{2\pi \i f_1 t}, \cdots,e^{2\pi \i f_{k-1} t}\}\\
W & = & \{ e^{2\pi \i f_1 t}, \cdots,e^{2\pi \i f_{k-1} t}, e^{2\pi \i f_{k} t}, e^{2\pi \i f_{k+1} t} \}
\end{eqnarray*}
where $f_1, f_2, \cdots, f_k$ are original frequencies in $x$. The idea is to show that any vector in the linear subspace $\mathrm{span}\{V\}$ is close to some vector in the linear subspace $\mathrm{span}\{V'\}$. 

For convenience, we use $\vec{u}^{\parallel}$ to denote the projection of vector $e^{2 \pi \i f_k t}$ to the linear subspace $\mathrm{span} \{U\}=\mathrm{span}\{e^{2 \pi \i f_1 t},\cdots,e^{2 \pi \i f_{k-1} t}\}$ and $\vec{w}^{\parallel}$ denote the projection of vector $e^{2 \pi \i f_{k+1} t}$ to this linear subspace $ \mathrm{span}\{U\}$. Let $\vec{u}^{\perp}=e^{2 \pi \i f_k t} - \vec{u}^{\parallel}$ and $\vec{w}^{\perp}=e^{2 \pi \i f_{k+1} t} - \vec{w}^{\parallel}$ be their orthogonal part to $ \mathrm{span}\{U\}$. 

From the definition $e^{2 \pi \i f_k t}=\vec{u}^{\parallel}+\vec{u}^{\perp}$ and $\vec{u}^{\parallel} \in \mathrm{span}\{U\}=\mathrm{span}\{e^{2 \pi \i f_1 t},\cdots,e^{2 \pi \i f_{k-1} t}\}$, we rewrite the linear combination \[x(t)=\sum_{j=1}^k v_j e^{2\pi \i f_j t}=\sum_{j=1}^{k-1} \alpha_j e^{2\pi \i f_j t} + v_k \cdot \vec{u}^{\perp}\] for  some scalars $\alpha_1,\cdots,\alpha_{k-1}$. 

We will substitute $\vec{u}^{\perp}$ by $\vec{w}^{\perp}$ in the above linear combination and find a set of new coefficients. Let $\vec{w}^{\perp}=\vec{w}_1 + \vec{w}_2$ where $\vec{w}_1 = \frac{\langle \vec{u}^{\perp}, \vec{w}^{\perp}\rangle}{\|\vec{u}^{\perp}\|^2_T} \vec{u}^{\perp}$ is the projection of $\vec{w}^{\perp}$ to $\vec{u}^{\perp}$. Therefore $\vec{w}_2$ is the orthogonal part of the vector $e^{2 \pi \i f_{k+1}t}$ to $\mathrm{span}\{V\}=\mathrm{span}\{ e^{2\pi \i f_1 t}, \cdots,e^{2\pi \i f_{k-1} t}, e^{2\pi \i f_k t} \}$. We use $\delta=\frac{\|\vec{w}_2\|_T}{\|\vec{w}^{\perp}\|_T}$ for convenience.

Notice that the $\underset{\beta\in \C}{\min} \frac{\|\vec{u}^{\perp} - \beta \cdot \vec{w}^{\perp}\|_T}{\|\vec{u}^{\perp}\|_T}=\delta$ and $\beta^*=\frac{\langle \vec{u}^{\perp}, \vec{w}^{\perp}\rangle}{\|\vec{w}^{\perp}\|^2_T}$ is the optimal choice. Therefore we set 


\begin{equation*}
x'(t)=\sum_{j=1}^{k-1} \beta_j e^{2\pi \i f_j t} +  v_k \cdot \beta^* \cdot \vec{w}^{\perp} \in \mathrm{span}\{e^{2 \pi \i f_1 t}, \cdots, e^{2 \pi \i f_{k-1} t}, e^{2 \pi \i f_{k+1}t }\}
\end{equation*}
where the coefficients $\beta_1,\cdots, \beta_{k-1}$ guarantee that the projection of $x'$ onto $ \mathrm{span}\{U\}$ is as same as the projection of $x$ onto $\mathrm{span}\{U\}$. From the choice of $\beta^*$ and the definition of $x'$, 
\begin{equation*}
\|x(t)-x'(t)\|_T^2 = \delta^2 \cdot |v_k|^2 \cdot \|\vec{u}^{\perp}\|_T^2 \le \delta^2 \cdot \|x(t)\|_T^2.
\end{equation*}

Eventually, we show an upper bound for $\delta^2$ from Claim \ref{cla:orthogonal_distance}. 
\begin{eqnarray*}
\delta^2 & = & \frac{\|\vec{w}_2\|_T^2}{\|\vec{w}^{\perp}\|_T^2}\\
&=&\frac{ \det(\Gram_W )}{ \det(\Gram_V )} / \frac{\det(\Gram_{V'}) }{ \det(\Gram_U) } \text{~by~Claim~\ref{cla:orthogonal_distance}}\\
& = & \frac{\det(\Gram_W ) }{\det(\Gram_V ) } \cdot \frac{\det(\Gram_U ) }{\det(\Gram_{V'} ) } \text{~by~Corollary~\ref{cor:determinant_corollary} }  \\
& \leq & k^{4\alpha k^2} \cdot \frac{ \overset{k+1}{ \underset{i=1}{\prod} } \overset{k+1}{ \underset{ \substack{j=1 \\ j \neq i} }{\prod} }  \min(|f_i -f_j|T,1)}{ \overset{k}{ \underset{i=1}{\prod} } \overset{k}{ \underset{ \substack{j=1 \\ j \neq i} }{\prod} }   \min(|f_i -f_j|T,1)} \cdot \frac{ \overset{k-1}{ \underset{i=1}{\prod} } \overset{k-1}{ \underset{ \substack{j=1 \\ j \neq i} }{\prod} } \min(|f_i -f_j|T,1) }{   \overset{k-1}{ \underset{i=1}{\prod} } \overset{k-1}{ \underset{ \substack{j=1 \\ j \neq i} }{\prod} } \min(|f_i -f_j|T,1) \cdot \overset{k-1}{ \underset{i=1}{\prod}}   \min(|f_i - f_{k+1}|^2 T^2,1) } \\
& = & k^{4\alpha k^2} |f_k - f_{k+1}|^2 T^2 \\
\end{eqnarray*}
\end{proof}

\begin{lemma}\label{lem:exists_k_separated_frequency_in_small_region}
For any $k$ frequencies $f_1 < f_2 < \cdots < f_k$, there exists $k$ frequencies $f'_1,\cdots,f'_k$  such that $\underset{i \in [k-1]}{\min} f'_{i+1} - f'_i \geq \eta$ and for all $i \in [k]$, $|f'_i-f_i| \le k\eta$.
\end{lemma}
\begin{proof}
We define the new frequencies $f_i'$ as follows: $f_1' =  f_1$ and $f_i'=\max \{f'_{i-1} + \eta, f_i \}$ for $i\in \{ 2, 3, \cdots, k \}$.
\end{proof}


\subsection{Existence of  nearby $k$-Fourier-sparse signal with frequency gap bounded away from zero}\label{sec:existence_of_bounded_gap_signal}
We combine the results in the above section to finish the proof of Lemma \ref{lem:low_degree_approximates_concentrated_freq}. We first prove that for any $x^*(t) = \overset{k}{ \underset{j=1}{\sum}} v_j e^{2\pi \i f_j t}$, there always exists another $k$-Fourier-sparse signal $x'$  close to $x^*(t)= \overset{k}{ \underset{j=1}{\sum }} v_j e^{2\pi \i f_j t}$ such that the frequency gap in $x'$ is at least $\eta \ge 2^{-\poly(k)}$. Then we show how to find a low degree polynomial $P(t)$ approximating $x'(t)$.

\restate{lem:existence_gap}
\begin{proof}
Using Lemma \ref{lem:exists_k_separated_frequency_in_small_region} on frequencies $f_1,\cdots,f_k$, we obtain $k$ new frequencies $f'_1,\cdots,f'_k$ such that their gap is at least $\eta$ and $\max_i |f_i-f'_i|\le k\eta$. Next we use the hybrid argument to find $x'$. 

Let $x^{(0)}(t)=x^*(t)$. For $i=1,\cdots,t,$ we apply Lemma \ref{lem:shifting_one_frequency_preserve_the_energy} to shift $f_i$ to $f'_i$ and obtain 
\begin{equation*}
x^{(i)}(t)=\overset{k}{ \underset{j=i+1}{\sum}} v^{(i)}_j e^{2\pi \i f_j t}+\overset{i}{ \underset{j=1}{\sum}} v^{(i)}_j e^{2\pi \i f'_j t}.
\end{equation*}
From Lemma \ref{lem:shifting_one_frequency_preserve_the_energy}, we know $\|x^{(i)}(t)-x^{(i-1)}(t)\|_T \le k^{C_0 k^2} (|f_i - f'_{i}|T) \|x^{(i-1)}\|_T$. 
Thus  we obtain 
\begin{equation*}
\left(1-k^{C_0 k^2} (k \eta T)\right)^i \|x^{(0)}(t)\|_T \le \|x^{(i)} (t)\|_T \le \left(1+k^{C_0 k^2} (k \eta T)\right)^i \|x^{(0)}(t)\|_T,
\end{equation*}
which is between $\left[\left(1 - i \cdot k^{C_0 k^2} (k \eta T) \right)\|x^{(0)}(t)\|_T,\left(1 + 2 i \cdot k^{C_0 k^2} (k \eta T)\right)\|x^{(0)}(t)\|_T\right]$ for $\eta \le \frac{1}{5 T} \cdot k^{-C_1 k^2}$ with some $C_1>C_0$.

At last, we set $x'(t)=x^{(k)}(t)$ and bound the distance between $x'(t)$ and $x^*(t)$ by 
\begin{align*}
 \|x^{(k)}(t)-x^{(0)}(t)\|_T \le &~ \sum_{i=1}^k \|x^{(i)}(t)-x^{(i-1)}(t)\|_T &\text{~by~triangle~inequality}\\
 \le &~ \sum_{i=1}^k  k^{C_0 k^2} (|f_i - f_i'| T) \|x^{(i-1)}(t)\|_T &\text{~by~Lemma~\ref{lem:shifting_one_frequency_preserve_the_energy}}\\
 \le &~ \sum_{i=1}^k  2k^{C_0 k^2} (k \eta T) \|x^{(i-1)}(t)\|_T & \text{~by~$\underset{i}{\max} |f_i -f_i'| \leq k\eta$}\\
 \le &~ k \cdot 2k^{C_0 k^2} (k \eta T) \|x^*(t)\|_T \\
 \le &~ \delta \|x^*(t)\|_T 
\end{align*}
where the last inequality follows by the sufficiently small $\eta$.
\end{proof}

\subsection{Approximating $k$-Fourier-sparse signals by polynomials}\label{sec:transfer_bounded_gap_signal_to_polynomial}
For any $k$-Fourier-sparse signal with frequency gap bounded away from
zero, we show that there exists a low degree polynomial which is close
to the original $k$-Fourier-sparse signal in $\| \cdot \|_T$
distance.

\begin{lemma}[Existence of low degree polynomial]\label{lem:low_degree_polynomial}
Let $x^*(t) = \overset{k}{ \underset{j=1}{\sum} }  v_j e^{2\pi \i f_j t}$, where $\forall j \in [k], |f_j  |\leq \Delta$ and $ \underset{i\neq j}{\min}  | f_{i} - f_j| \geq \eta$. There exists a polynomial $Q(t)$ of degree 
\begin{equation*}
d=O\left(T\Delta + k \log 1/(\eta T) + k^2 \log k + k \log (1/\delta)\right)
\end{equation*} such that,
\begin{equation}
\|Q(t)-x^*(t)\|_T^2 \le \delta \|x^*(t)\|_T^2
\end{equation}
\end{lemma}
\begin{proof}
For each frequency $f_j$, let $Q_j(t)= \overset{d-1}{ \underset{k=0}{\sum}}  \frac{(2\pi \i f_j t)^k}{k!}$ be the first $d$ terms in the Taylor Expansion of $e^{2 \pi \i f_j t}$. For any $t \in [0,T]$, we know the difference between $Q_j(t)$ and $e^{2 \pi \i f_j t}$ is at most $$|Q_j(t)-e^{2 \pi \i f_j t}| \le |\frac{(2\pi \i f_j T)^{d}}{d!}| \le (\frac{2\pi T \Delta \cdot e}{d})^d.$$

We define $$Q(t)=\sum_{j=1}^k v_j Q_j(t)$$ and bound the distance between $Q$ and $x^*$ from the above estimation:
\begin{align*}
 \|Q(t)-x^*(t)\|_T^2 =  & ~\frac{1}{T} \int_0^T |Q(t) - x^*(t)|^2 \mathrm{d} t \\
=  &~ \frac{1}{T} \int_0^T |\sum_{j=1}^k v_j(Q_j(t) - e^{2 \pi \i f_j t})|^2 \mathrm{d} t\\
\le  &~ 2k \sum_{j=1}^k \frac{1}{T} \int_0^T |v_j|^2 \cdot |Q_j(t) - e^{2 \pi \i f_j t}|^2 \mathrm{d} t & \text{~by~triangle~inequality}\\
\le  &~ k \sum_{j=1}^k |v_j|^2 \cdot (\frac{2\pi T \Delta \cdot e}{d})^{2d} &\text{~by~Taylor~expansion}
\end{align*}
On the other hand, from Lemma \ref{lem:relation_energy_coef}, we know 
\begin{equation*}
\|x^*(t)\|_T^2 \ge (\eta T)^{2k} \cdot k^{-c k^{2}} \sum_j |v_j|^2.
\end{equation*}
Because $d = 10 \cdot \pi e  ( T\Delta + k \log 1/(\eta T) + k^2 \log k + k \log (1/\delta))$ is large enough, we have $k(\frac{2\pi T \Delta \cdot e}{d})^{2d} \le \delta (\eta T)^{2k} \cdot k^{-c k^{2}}$, which indicates that $\|Q(t)-x^*(t)\|_T^2\le \delta \|x^*\|_T^2$ from all discussion above. 
\end{proof}

\subsection{Transferring degree-$d$ polynomial to ($d$+$1$)-Fourier-sparse signal}\label{sec:transfer_polynomial_to_fourier}
In this section, we show how to transfer a degree-$d$ polynomial to ($d$+$1$)-Fourier-sparse signal. 
\begin{lemma}\label{lem:polynomial_to_FT}
For any degree-$d$ polynomial $Q(t) = \overset{d}{\underset{j=0}{\sum}} c_j t^j$, any $T>0$ and any $\epsilon>0$, there always exist $\gamma>0$ and $$x^*(t)=\sum_{i=1}^{d+1} \alpha_i e^{2\pi \i (\gamma i) t}$$ with some coefficients $\alpha_0,\cdots,\alpha_d$ such that
\begin{equation*}
\forall t \in [0,T], |x^*(t) - Q(t)| \le \epsilon. 
\end{equation*}
\end{lemma}

\begin{proof}
We can rewrite $x^*(t)$,
\begin{align*}
x^*(t) = ~ &\sum_{i=1}^{d+1} \alpha_i e^{2\pi \i \gamma i t} \\
= ~& \sum_{i=1}^{d+1} \alpha_i \sum_{j=0}^{\infty} \frac{ (2\pi \i \gamma i t)^j }{j!} \\
=~ & \sum_{j=0}^{\infty} \frac{(2\pi \i \gamma t)^j}{j!}  \sum_{i=1}^{d+1} \alpha_i \cdot i^j \\
=~ & \sum_{j=0}^{d} \frac{(2\pi \i \gamma t)^j}{j!}  \sum_{i=1}^{d+1} \alpha_i \cdot i^j  +  \sum_{j=d+1}^{\infty} \frac{(2\pi \i \gamma t)^j}{j!}  \sum_{i=1}^{d+1} \alpha_i \cdot i^j \\
=~ & Q(t) + \underbrace{ \left( \sum_{j=0}^{d} \frac{(2\pi \i \gamma t)^j}{j!}  \sum_{i=1}^{d+1} \alpha_i \cdot i^j - Q(t)\right) }_{C_1} + \underbrace{ \left( \sum_{j=d+1}^{\infty} \frac{(2\pi \i \gamma t)^j}{j!}  \sum_{i=1}^{d+1} \alpha_i \cdot i^j \right) }_{C_2}.
\end{align*}
Our goal is to show there exists some parameter $\gamma$ and coefficients $\{\alpha_0, \alpha_1, \cdots, \alpha_d\}$ such that the term $C_1=0$ and $|C_2| \leq \epsilon$. Let's consider $C_1$,
\begin{align*}
C_1 = \sum_{j=0}^d (\frac{t}{T})^j\left( \frac{(2\pi \i \gamma T)^j}{j!}\sum_{i=1}^{d+1} \alpha_i i^j -c_j \right)
\end{align*}
To guarantee $C_1=0$, we need to solve a linear system with $d+1$ unknown variables and $d+1$ constraints, 
\begin{align*}
\text{Find} ~& ~\alpha_1, \alpha_2, \cdots \alpha_{d+1} \\
\text{s.t.} ~& ~\frac{(2\pi \i \gamma T)^j}{j!}\sum_{i=1}^{d+1} \alpha_i i^j -c_j  = 0, \forall j \in \{0,1,\cdots,d\}
\end{align*}
Define $c'_j = c_j j! /(2\pi \i \gamma)^j$, let $\alpha$ and $c'$ be the length-($d+1$) column vectors with $\alpha_i$ and $c'_j$. Let $A\in \mathbb{R}^{d+1\times d+1}$ denote the Vandermonde matrix where $A_{i,j} = i^j, \forall i,j\in [d+1]\times \{0,1,\cdots,d\}$. Then we need to guarantee $A \alpha = c'$. Using the definition of determinant, $\det(A) = \underset{i<j}{\prod} |i-j| \leq 2^{O(d^2\log d)}$. Thus $\sigma_{\max}(A) \leq 2^{O(d^2\log d)}$ and then 
\begin{align*}
\sigma_{\min}(A) = \frac{\det(A)}{\prod_{i=1}^{d-1} \sigma_i} \geq  2^{-O(d^3\log d)}  .
\end{align*}
We show how to upper bound $|\alpha_i|$,
\begin{align*}
~ \max_{i\in [d+1]} |\alpha_i|  \leq \| \alpha \|_2  = \| A^\dagger c'\|_2 \leq \| A^\dagger \|_2 \cdot \|c'\|_2 \leq \frac{1}{\sigma_{\min}(A)} \sqrt{d+1}\max_{0\leq j\leq d} \frac{ |c_j| j!}{ (2\pi \gamma T)^j}
\end{align*} 
Plugging the above equation into $C_2$, we have
\begin{align*}
|C_2| &~= \left| \sum_{j=d+1}^{\infty} \frac{(2\pi \i \gamma t)^j}{j!}  \sum_{i=1}^{d+1} \alpha_i \cdot i^j \right| \\
&~ \leq\sum_{j=d+1}^{\infty} \frac{(2\pi  \gamma t)^j}{j!}  \sum_{i=1}^{d+1} |\alpha_i| \cdot i^j \\
& ~\leq \sum_{j=d+1}^{\infty} \frac{(2\pi  \gamma t)^j}{j!}  (d+1)^{d+1}\max_{i\in[d+1]} |\alpha_i|\\
& ~\leq \sum_{j=d+1}^{\infty} \frac{(2\pi  \gamma t)^j}{j!}  (d+1)^{d+2} \frac{1}{\sigma_{\min}(A)} \frac{d!}{(2\pi \gamma T)^d} \max_{ 0\leq j\leq d} |c_j| \\
&~ \leq \epsilon
\end{align*}
where the last step follows by choosing sufficiently small 
\begin{align*}
\gamma \lesssim \epsilon / \left( T 2^{\Theta(d^3\log d)} \underset{0\leq j\leq d}{\max} |c_j| \right).
\end{align*}
\end{proof}

\section{$k$-cluster Signal Recovery}\label{sub:kmountains}
\subsection{Overview}
In this section, we prove Lemma \ref{lem:signal_recovery_k_cluster} as the main technical lemma to finish the proof of main Theorem~\ref{thm:main}, which shows how to learn $x^*(t) = \sum_{j=1}^k v_j e^{2\pi \i f_j t}$ with noise. 
\begin{lemma}\label{lem:signal_recovery_k_cluster}
Let $x^*(t) = \sum_{j=1}^k v_j e^{2\pi \i f_j t}$ and $x(t) = x^*(t) + g(t)$ be our observation. For any $\delta>0$ and $T>0$, let $\N^2 : = \frac{1}{T} \int_0^T |g(t) |^2 \mathrm{d} t + \delta \cdot \frac{1}{T} \int_0^T |x^*(t) |^2 \mathrm{d} t$. For $\Delta=\polydelta$,  Procedure \textsc{SignalRecoveryKCluster$^+$} in Algorithm \ref{alg:main_k} takes $l=O(k)$ frequencies $\widetilde{f}_1,\cdots,\widetilde{f}_l$ as input and finds $l$ polynomials $Q_1,\cdots,Q_l$ of degree $d=O( (T\Delta)^{1.5} + k^3 \log k + k \log 1/\delta)$ such that
\begin{equation}\label{eq:key_eq_signal_recovery_k_cluster}
\widetilde{x}(t)=\sum_{j \in [l]}Q_j(t) e^{2 \pi \i \widetilde{f}_j t} \mathrm {~satisfies~} \|\widetilde{x}(t)-x^*(t)\|_T^2 \lesssim \N^2.
\end{equation}
The procedure succeeds with probability at least $1-2^{-\Omega(k)}$, uses $\poly(k,\log(1/\delta)) \cdot \log(FT)$ samples, and runs in $\poly(k,\log(1/\delta)) \cdot \log^2(FT)$ time.
\end{lemma}
For any set $W=\{t_1,\cdots,t_m\}$ where each $t_i \in [0,T]$, we use \[\|\vec{v}\|_W=\sqrt{\frac{\sum_{i \in W}|\vec{v}(t_i)|^2}{|W|}} \textit{ for any }\vec{v}:[0,T] \rightarrow \mathbb{C}\] in this section. 
We first show that Procedure \textsc{SignalRecoveryKCluster} succeeds with constant probability, then prove that Procedure \textsc{SignalRecoveryKCluster$^+$} succeeds with probability at least $1-2^{-\Omega(k)}$.

\subsection{Heavy clusters separation}
Recall the definition of ``heavy'' clusters.


\restate{def:heavy_clusters}
By reordering $C_i$, we can assume $\{C_1, C_2,\cdots, C_l\}$ are heavy clusters, where $l\leq n \leq k$.
\restate{cla:guarantee_removing_x**_x*}
\begin{proof}
Let $x^{(\overline{S})}(t) = \underset{j\in [k] \backslash S}{\sum} v_j e^{2 \pi \i f_j t}$. Notice that $\|x^*-x^{(S)}\|_T^2=\|x^{(\overline{S})}\|^2_T$.

From the property \RN{6} of filter function $(H, \widehat{H})$ in Appendix \ref{sec:properties_of_H}, we have \[\int_{-\infty}^{+\infty} |x^{(\overline{S})}(t) \cdot H(t)|^2 \mathrm {d} t \ge 0.9 \int_{0}^T |x^{(\overline{S})}(t)|^2 \mathrm{d} t = 0.9 \cdot T \|x^{\overline{S}}\|_T^2.\]

From Definition \ref{def:heavy_clusters}, we have 
\begin{align*}
\int_{-\infty}^{+\infty} |x^{ (\overline{S}) }(t) \cdot H(t)|^2 \mathrm {d} t   =& ~ \int_{-\infty}^{+\infty} |\wh{x^{ (\overline{S}) } \cdot H}(f)|^2 \mathrm {d} f \\
= & ~ \int_{[-\infty,+\infty]\setminus C_1 \cup \cdots \cup C_l} |\widehat{x^{*} \cdot H}(f)|^2 \mathrm {d} f\\
 \le& ~ k \cdot T\N^2/k.
\end{align*}

Overall, we have $\|x^{(\overline{S})}\|^2_T \lesssim \N^2.$
\end{proof}
From the guarantee of Theorem \ref{thm:frequency_recovery_k_cluster}, for any $j \in S$, $ \underset{i\in [l]}{\min} |f_j - \widetilde{f}_i| \le \Delta \sqrt{\Delta T}$. From now on, we focus on the recovery of $x^{(S)}$, which is enough to approximate $x^*$ from the above claim. Because we are looking for $\tilde{x}$ approximating $x^{(S)}$ within distance $O(\N^2)$, from Lemma \ref{lem:existence_gap}, we can assume there is a frequency gap $\eta \ge \frac{\delta}{10T} k^{- O(k^2)}$ among $x^{(S)}$. 

\subsection{Approximating clusters by polynomials}
In this section, we show how to approximate $x^{(S)}$ by  $x'(t)=\sum_{i \in [l]}e^{2 \pi \i \wt{f}_i t}P_i(t)$ where $P_1,\cdots,P_l$ are low degree polynomials.
\begin{claim}
For any $x^{(S)}(t) = \sum_{j \in S} v_j e^{2\pi \i f_j t}$ with a frequency gap $\eta=\underset{i\neq j}{\min}  | f_{i} - f_j|$ and $l$ frequencies $\widetilde{f}_1,\cdots,\widetilde{f}_l$ with the property $\forall j \in S, \min_{i \in [l]}|f_j - \widetilde{f}_i| \le \Delta \sqrt{\Delta T}$, let 
\[d=5 \pi \left( (T \Delta)^{1.5} + k^3 \log k + \log 1/\delta \right) \textit{ and } V=\left\{t^j e^{2 \pi \i \wt{f}_i t}| i \in [l], j\in \{0,\cdots,d\}\right\}.\] There exists $x'(t) \in \mathrm{span}\{V\}$ that approximates $x^{(S)}(t)$ as follows:
\begin{equation*}
\forall t \in [0,T], |x'(t)-x^{(S)}(t)| \le \delta \|x^{(S)}\|_T.
\end{equation*}
\end{claim}
\begin{proof}
From Lemma \ref{lem:relation_energy_coef}, we know 
\begin{equation*}
\|x^{(S)}\|^2_T \ge (\eta T)^{2k} \cdot k^{-c k^{2}} \sum_{j \in S} |v_j|^2.
\end{equation*}
For each frequency $f_j$, we use $p_j$ to denote the index in $[l]$ such that $|f_j - \widetilde{f}_{p_j}| \le \Delta \sqrt{\Delta T}$. We rewrite 
\begin{equation*}
x^{(S)}(t)=\sum_{i=1}^l e^{2\pi \i \widetilde{f}_i} \left(\sum_{j\in S:p_j=i} v_j e^{2 \pi \i (f_j-\widetilde{f}_i)t}\right).
\end{equation*}
For $d=5 \pi ( (T\Delta)^{1.5} + k^3 \log k + \log 1/\delta)$ and each $e^{2 \pi \i (f_j-\widetilde{f}_{p_j})t}$, let $Q_j(t)=\sum_{i=0}^{d-1} \frac{\left(2\pi \i (f_j-\widetilde{f}_{p_j}) t\right)^i}{i!}$ be the first $d$ terms in the Taylor Expansion of $e^{2 \pi \i (f_j-\widetilde{f}_{p_j})t}$. For any $t \in [0,T]$, we know the difference between $Q_j(t)$ and $e^{2 \pi \i (f_j-\widetilde{f}_{p_j}) t}$ is at most $$\forall t\in [0,T], |Q_j(t)-e^{2 \pi \i (f_j-\widetilde{f}_{p_j}) t}| \le |\frac{(2\pi \i (f_j-\widetilde{f}_{p_j}) T)^{d}}{d!}| \le (\frac{8 \pi (\Delta T)^{1.5}}{d})^d.$$

Let $x'=\sum_{i=1}^l e^{2\pi \i \widetilde{f}_i t} \left(\sum_{j\in S:p_j=i} v_j Q_j(t) \right).$ From all discussion above, we know for any $t \in [0,T]$,
\begin{align*}
|x'(t)-x^{(S)}(t)|^2 \le&~ \left(\sum_{j \in S} |v_j|(\frac{8 \pi (T\Delta)^{1.5}}{d})^d\right)^2 \\
\le&~ k(\frac{8 \pi (T\Delta)^{1.5}}{d})^{2d} \sum_j |v_j|^2 \\
\le&~ \frac{k( \frac{8 \pi (T\Delta)^{1.5}}{d})^{2d}}{(\eta T)^{2k} \cdot k^{-c k^{2}}} \|x^{(S)}\|_T^2 \\
\le&~ \delta^2 \|x^{(S)}\|_T^2.
\end{align*}
\end{proof}

We provide a property of functions in $\mathrm{span}\{V\}$ such that we can use the Chernoff bound and the $\epsilon$-net argument on vectors in $\mathrm{span}\{V\}$.
\restate{cla:max_bounded_Q}
\begin{proof}
From Lemma \ref{lem:polynomial_to_FT}, we can approximate each polynomial in $\vec{u}$ by a linear combination of $\{1,e^{2 \pi \i \cdot \gamma t},\cdots, e^{2 \pi \i \cdot (\gamma d) t}\}$ such that we obtain $u^* \in \mathrm{span}\left\{e^{2 \pi \i \cdot (\gamma j)t}\cdot e^{2 \pi \i \tilde{f}_i t}|i \in [l],j \in \{0,\cdots,d+1\}\right\}$ for some small $\gamma$ such that $\forall t\in [0,T], |\vec{u}(t)-u^*(t)| \le 0.01 \|\vec{u}\|_T$.

From Lemma \ref{lem:max_is_at_most_polyk_times_l2}, we know 
\begin{equation*}
\max_{t \in [0,T]} |u^*(t)|^2 \le C \cdot \left( (ld+1)^{4} \cdot \log^{3} (ld+1)\right) \|u^*\|_T^2.
\end{equation*}
For some constant $C'$, we have
\begin{equation*}
\max_{t \in [0,T]} |\vec{u}(t)|^2 \le C' \left((k d)^{C_1} \log^{C_2} d\right) \|\vec{u}\|_T^2.
\end{equation*}
\end{proof}

\subsection{Main result, with constant success probability}
In this section, we show that the output $\wt{x}$ is close to $x'$ with high probability using the $\epsilon$-net argument, which is enough to prove $\|\wt{x} - x\|_T \lesssim \N^2$ from all discussion above. Because we can prove Lemma \ref{lem:tildex_is_close_to_x}(which is the main goal of this section), then 
combining $\|x'-x^*\|_T \le \|x' - x^{(S)}\|_T + \|x^{(S)}-x^*\|_T\lesssim \delta \|x^*\|_T$ and Lemma \ref{lem:tildex_is_close_to_x}, we have $\|x^*-\widetilde{x}\|_T \lesssim \|g\|_T + \delta \|x^*\|_T$, which finishes the proof of 
 Procedure \textsc{SignalRecoveryKCluster} in Algorithm \ref{alg:main_k} achieving the Equation (\ref{eq:key_eq_signal_recovery_k_cluster}) with constant success probability but not $1-2^{-\Omega(k)}$. We will boost the success probability in Section \ref{sec:signal_recovery_boosting_success_probability}.


We first provide an $\epsilon$-net $\mathcal{P}$ for the unit vectors $\mathcal{Q}=\{\vec{u} \in \mathrm{span}\{V\}\big{|}\|\vec{u}\|_T^2=1\}$ in the linear subspace $\mathrm{span}\{V\}$ where $V=\left\{t^j \cdot e^{2 \pi \i \wt{f}_i t} \big{|}j \in \{0,1,\cdots,d\},i \in [l]\right\}$ from the above discussion. Notice that the dimension of $\mathrm{span}\{V\}$ is at most $l(d+1)$.
\begin{claim}
There exists an $\epsilon$-net $\mathcal{P} \subset \mathrm{span}\{V\}$ such that
\begin{enumerate}
\item $\forall \vec{u} \in Q, \exists \vec{w} \in \mathcal{P}, \|\vec{u}-\vec{w}\|_T \le \epsilon.$
\item $|\mathcal{P}| \le \left(5\frac{l(d+1)}{\epsilon}\right)^{2l(d+1)}.$
\end{enumerate}
\end{claim}
\begin{proof}
Let $\mathcal{P'}$ be an $\frac{\epsilon}{l(d+1)}$-net in the unit circle of $\mathbb{C}$ with size at most $(4\frac{l(d+1)}{\epsilon}+1)^2$, i.e., 
\begin{equation*}\mathcal{P'}=\left\{\frac{\epsilon}{2l(d+1)} j_1 + \i \frac{\epsilon}{2l(d+1)} j_2\bigg{|}j_1, j_2 \in \mathbb{Z}, |j_1| \le \frac{2l(d+1)}{\epsilon}, |j_2| \le \frac{2l(d+1)}{\epsilon}\right\}.
\end{equation*} 
Observe that the dimension of $\mathrm{span}\{V\}$ is at most $l(d+1)$.  Then we take an orthogonal basis $\vec{w}_1,\cdots,\vec{w}_{l(d+1)}$ in $\mathrm{span}\{V\}$ and set 
\begin{equation*} 
P=\{\sum_{i=1}^{l(d+1)} \alpha_i \vec{w}_i\big{|} \forall i\in [l(d+1)], \alpha_i \in \mathcal{P'}\} . 
\end{equation*}

Therefore $\mathcal{P}$ is an $\epsilon$-net for $Q$ and $|\mathcal{P}| \le \left(5\frac{l(d+1)}{\epsilon}\right)^{2l(d+1)}$.
\end{proof}

We first prove that $W$ is a good estimation for all functions in the $\epsilon$-net $\mathcal{P}$.
\begin{claim}
For any $\epsilon>0$, there exists a universal constant $C_3 \le \finalC$ such that for a set $S$ of i.i.d. samples chosen uniformly at random over $[0,T]$ of size $|S| \ge \frac{3 (kd)^{C_3} \log^{C_3}d/\epsilon}{\epsilon^2}$,then with probability at least $1-k^{-k}$, for all $\vec{w} \in \mathcal{P}$, we have
\begin{equation*}
\|\vec{w} \|_W \in \left[(1-\epsilon)\|\vec{w} \|_T,(1+\epsilon)\|\vec{w}\|_T\right].
\end{equation*}
\end{claim}
\begin{proof}
From Claim \ref{cla:max_bounded_Q} and Lemma \ref{lem:max_is_bounded_imply_sample_is_small}, for each $\vec{w} \in \mathcal{P}$, $$\mathsf{Pr}\left[\|\vec{w}(t)\|_W \notin \big[(1-\epsilon)\|\vec{w}\|_T,(1+\epsilon)\|\vec{w}\|_T\big]\right]\le 2^{-\frac{|W| \epsilon^2}{3 (k d)^{C_1} \log^{C_2+0.5} d}}\le 2^{-k d \log^{1.5} \frac{d}{\epsilon}}.$$ 

From the union bound, $\|\vec{w} \|_W \in [(1-\epsilon)\|\vec{w}\|_T,(1+\epsilon)\|\vec{w}\|_T]$ for any $\vec{w} \in \mathcal{P}$ with probability at least $1-(\frac{d}{\epsilon})^{-k d \log^{0.5} d} \cdot |\mathcal{P}|\ge 1- d^{-d}$.
\end{proof}
Then We prove that $W$ is a good estimation for all functions in $\mathrm{span}\{V\}$ using the property of $\epsilon$-nets.
\begin{claim}\label{cla:concentration_for_any_polynomial_signal}
For any $\epsilon>0$, there exists a universal constant $C_3 \le \finalC$ such that for a set $W$ of i.i.d. samples chosen uniformly at random over $[0,T]$ of size $|W| \ge \frac{3 (kd)^{C_3} \log^{C_3} d/\epsilon}{\epsilon^2}$,then with probability at least $1-d^{-d}$, for all $u \in \mathrm{span}\{V\}$, we have
\begin{equation*}
\|\vec{u}\|_W \in \left[ (1-3\epsilon)\|\vec{u}\|_T,(1 + 3\epsilon)\|\vec{u}\|_T \right]
\end{equation*} 
\end{claim}
\begin{proof}
We assume that the above claim is true for any $\vec{w} \in \mathcal{P}$. Without loss of generality, we consider $\vec{u} \in \mathcal{Q}$ such that $\|\vec{u}\|_T=1$.

Let $\vec{w}_0$ be the vector in $\mathcal{P}$ that minimizes $\|\vec{w}-\vec{u}\|_T$ for all $\vec{w} \in \mathcal{P}$, i.e., $\vec{w}_0 =  \underset{ \vec{w} \in \mathcal{P} } {\arg\min} \|\vec{w}-\vec{u}\|_T$. Define $\vec{u}_1=\vec{u}-\vec{w}_0$ and notice that $\|\vec{u}_1\|_T\le \epsilon$ because $\mathcal{P}$ is a $\epsilon$-net. If $\|\vec{u}_1\|_T=0$, then we skip the rest of this procedure. Otherwise, we define $\alpha_1=\|\vec{u}_1\|_T$ and normalize $\widetilde{u}_1 = \vec{u}_1 /\alpha_1$. 

Then we choose $\vec{w}_1$ to be the vector in $\mathcal{P}$ that minimizes $\|\vec{w}-\widetilde{u}_1\|_T$ for all $\vec{w} \in \mathcal{P}$. Similarly, we set $\vec{u}_2=\widetilde{u}_1-\vec{w}_1$ and $\alpha_2 =\|\vec{u}_2\|_T$. Next we repeat this process for $\widetilde{u}_2=\vec{u}_2/\alpha_2$ and so on. The recursive definition can be summarized in the following sense, 
\begin{eqnarray*}
\text{initial~:}&& \widetilde{u}_0 = \vec{u} \text{ and } m=10 \log_{1/\epsilon} (ld) + 1, \\
\text{For }i \in \{0,1,2,\cdots,m\}~ :&& \vec{w}_{i} = \underset{ \vec{w} \in \mathcal{P} } {\arg\min} \| \vec{w} - \widetilde{u}_i \|_T, \\
&& \vec{u}_{i+1} = \widetilde{u}_i-\vec{w}_i \text{ and } \alpha_{i+1} = \| \vec{u}_{i+1} \|_T, \\
 && \text{if }\alpha_{i+1}=0, \text{~stop.}\\
&& \text{if }\alpha_{i+1}\neq 0, \widetilde{u}_{i+1} =  \vec{u}_{i+1} / \alpha_{i+1}  \text{~and~continue}, 
\end{eqnarray*}

Eventually, we have $\vec{u}=\vec{w}_0 + \alpha_1 \vec{w}_1 + \alpha_1 \alpha_2 \vec{w}_2 + \cdots + \prod_{j=1}^{m} \alpha_j (\vec{w}_m+ \vec{u}_{m+1})$ where each $|\alpha_i|\le \epsilon$ and each $\vec{w}_i$ is in the $\epsilon$-net $\mathcal{P}$. Notice that $\|\vec{u}_{m+1}\|_T \le 1$ and $\|\vec{u}_{m+1}\|_W \le (ld+1)^{3} \cdot \|\vec{u}_{m+1}\|_T$ from Claim \ref{cla:max_bounded_Q}. We prove a lower bound for $\| \vec{u} \|_W$,
\begin{eqnarray*}
\|\vec{u}\|_W&=&\|\vec{w}_0 + \alpha_1 \vec{w}_1 + \alpha_1 \alpha_2 \vec{w}_2 + \cdots + \prod_{j=1}^{m} \alpha_j (\vec{w}_m + \vec{u}_{m+1}) \|_W\\
&\geq & \|\vec{w}_0 \|_W-\|\alpha_1 \vec{w}_1 \|_W - \|\alpha_1 \alpha_2 \vec{w}_2 \|_W - \cdots - \|\prod_{j=1}^{m}  \alpha_j \vec{w}_m \|_W -  \|\prod_{j=1}^{m} \alpha_j \vec{u}_{m+1}\|_W\\ 
&\ge & (1-\epsilon) - \epsilon (1+\epsilon) - \epsilon^2 (1+\epsilon) - \cdots - \epsilon^m (1+\epsilon) - \epsilon^m \|\vec{u}_{m+1}\|_W \\ 
&\ge & 1-\epsilon - \frac{(1+\epsilon)\epsilon}{1-\epsilon} - \epsilon^m \cdot (ld+1)^{3} \ge 1- 3 \epsilon.
\end{eqnarray*}
Similarly, we have $\|\vec{u}\|_W \le 1 + 3 \epsilon$.
\end{proof}


\begin{lemma}\label{lem:tildex_is_close_to_x}
With probability at least $0.99$ over the $m$ i.i.d samples in $W$, 
\begin{equation*}
\| x'(t) - \wt{x}(t) \|_T \leq 2200\left(\|g(t)\|_T + \| x^{(S)}(t) - x'(t)\|_T\right).
\end{equation*}
\end{lemma}
\begin{proof}
Let $g'(t)=g(t) + x^*(t) - x'(t)$ such that $x(t)=x'(t) + g'(t)$. Then we choose $\epsilon=0.03$ and bound:
\begin{align*}
& ~\| x'(t) - \widetilde{x}(t) \|_T  \\
\leq &~ (1+3\epsilon) \| x'(t) - \widetilde{x}(t) \|_W \quad & \text{with~prob.~}1-2^{-\Omega(d\log d)}~\text{by~Claim~\ref{cla:concentration_for_any_polynomial_signal}}\\
=  &~ 1.09 \| x'(t) - \widetilde{x}(t) \|_W &\text{~by~}\epsilon=0.03 \\
 = &~1.09 \| x(t) - g'(t) - \widetilde{x}(t) \|_W & \text{~by~}x'(t)=x(t) - g'(t)\\
 \leq &~1.09\| x(t) - \widetilde{x}(t) \|_W + 1.09 \| g'(t) \|_W \quad & \text{by~triangle~inequality} \\
 \leq &~1.09\| x(t) - x'(t) \|_W + 1.09 \| g'(t) \|_W \quad & \text{by~}\wt{x} = \underset{y \in \mathrm{span}\{V\}}{\arg\min}\|x-y\|_W \\
 = &~2.18 \| g'(t) \|_W.& \text{~by~} x(t) - x'(t) = g(t)
\end{align*}
From the fact that $\E_{W}[\|g'\|_W]=\|g'\|_T$, $\|g'\|_W \le 1000 \|g'\|_T$ with probability at least .999. It indicates $\| x'(t) - \widetilde{x}(t) \|_T \le 2200 \|g'\|_T$ with probability at least $0.99$ from all discussion above.
\end{proof}


\subsection{Boosting the success probability}\label{sec:signal_recovery_boosting_success_probability}
In order to achieve $1-2^{-\Omega(k)}$ for the main theorem, we cannot combine Procedure \textsc{SignalRecoveryKCluster} with \textsc{FrequencyRecoveryKCluster} directly. However, using the similar proof technique in Theorem \ref{thm:accurate_poly_learning}, we are able to boost the success probability by using Procedure \textsc{SignalRecoveryKCluster$^+$} in Algorithm \ref{alg:main_k}.  It runs Procedure \textsc{SignalRecoveryKCluster} $R=O(k)$ times in parallel for independent fresh samples and report $R$ different $d$-Fourier-sparse signals $\wt{x}_i(t)$. Then, taking $m=\poly(k)$ new locations $\{t_1,t_2,\cdots, t_m\}$, and computing $\wt{A}$ as before and $\wt{b}_j$ by taking the median of $\{ \wt{x}_1(t_j), \cdots,  \wt{x}_{R}(t_j) \} $. At the end, solving the linear regression for matrix $\wt{A}$ and vector $\wt{b}$. Thus, we complete the proof of Lemma \ref{lem:signal_recovery_k_cluster}.

 Because we can transfer a degree-$d$ polynomial to a $d$-Fourier-sparse signal by Lemma \ref{lem:polynomial_to_FT}, the output of Procedure \textsc{CFTKCluster} in Algorithm \ref{alg:main_k} matches the main theorem,

\restate{thm:main}

\newpage
\bibliographystyle{alpha}
\bibliography{ref}
\newpage
\appendix
\section{Technical Proofs}\label{sec:proofs}

\subsection{Proof of Theorem~\ref{thm:bounds_gram_determinant}}\label{sec:bounds_gram_determinant}

We prove the following Theorem
\restate{thm:bounds_gram_determinant}
First, we note by the Cauchy-Binet formula that the determinant in question is equal to
\begin{equation}
\int_{-1}^1 \int_{-1}^1 \ldots \int_{-1}^1 \left|\det([e^{2\pi i \xi_i t_j}]_{i,j}) \right|^2 dt_1dt_2\ldots dt_k.
\end{equation}

We next need to consider the integrand in the special case when $\sum |\xi_i| \leq 1/8$.
\begin{lemma}
If $\xi_i\in R$ and $t_j\in \R$, $\sum_i |\xi_i|(\max_i |t_i|) \leq 1/8$ then
$$
|\det([e^{2\pi i \xi_i t_j}]_{i,j})| = \Theta\left(\frac{(2\pi)^{\binom{k}{2}}\prod_{i<j}|t_i-t_j||\xi_i-\xi_j|}{1!2!\cdots k!} \right).
$$
\end{lemma}
\begin{proof}
Firstly, by adding a constant to all the $t_j$ we can make them non-negative. This multiplies the determinant by a root of unity, and at most doubles $\sum_i |\xi_i|(\max_i |t_i|)$.

By continuity, it suffices to consider the $t_i$ to all be multiples of $1/N$ for some large integer $N$. By multiplying all the $t_j$ by $N$ and all $\xi_i$ by $1/N$, we may assume that all of the $t_j$ are non-negative integers with $t_1\leq t_2\leq \ldots \leq t_k$.

Let $z_i = \exp(2\pi i \xi_i)$. Then our determinant is 
$$
\det\left( \left[z_i^{t_j} \right]_{i,j}\right),
$$
which is equal to the Vandermonde determinant times the Schur polynomial $s_\lambda(z_i)$ where $\lambda$ is the partition $\lambda_j = t_j-(j-1)$.

Therefore, this determinant equals
$$
\prod_{i<j}(z_i-z_j) s_\lambda(z_1,z_2,\ldots,z_k).
$$
The absolute value of
$$
\prod_{i<j}(z_i-z_j)
$$
is approximately $\prod_{i<j} (2\pi i)(\xi_i-\xi_j)$, which has absolute value $(2\pi)^{\binom{k}{2}}\prod_{i<j}|\xi_i-\xi_j|$. We have left to evaluate the size of the Schur polynomial.

By standard results, $s_\lambda$ is a polynomial in the $z_i$ with non-negative coefficients, and all exponents at most $\max_j |t_j|$ in each variable. Therefore, the monomials with non-zero coefficients will all have real part at least $1/2$ and absolute value $1$ when evaluated at the $z_i$. Therefore,
$$
|s_\lambda(z_1,\ldots,z_k)| = \Theta(|s_\lambda(1,1,\ldots,1)|).
$$
On the other hand, by the Weyl character formula
$$
s_\lambda(1,1,\ldots,1) = \prod_{i<j} \frac{t_j-t_i}{j-i} = \frac{\prod_{i<j} |t_i-t_j|}{1!2!\ldots k!}.
$$
This completes the proof.
\end{proof}

Next we prove our Theorem when the $\xi$ have small total variation.
\begin{lemma}
If there exists a $\xi_0$ so that $\sum |\xi_i-\xi_0|<1/8$, then 
$$
\det(G_{\xi_1,\ldots,\xi_k}) = \Theta\left( \frac{2^{3k(k-1)/2}\pi^{k(k-1)}\prod_{i<j}|\xi_i-\xi_j|^2}{(k!)^3\prod_{n=0}^{k-1}(2n)!}\right).
$$
\end{lemma}
\begin{proof}
By translating the $\xi_i$ we can assume that $\xi_0=0$.

By the above we have 
$$
\Theta(\frac{(2\pi)^{k(k-1)}\prod_{i<j}|\xi_i-\xi_j|^2}{(1!2!\cdots k!)^2}) \int_{-1}^1 \ldots \int_{-1}^1 \prod_{i<j} |t_i-t_j|^2 dt_1\ldots dt_k.
$$
We note that by the Cauchy-Binet formula the latter term is the determinant of the matrix $M$ with $M_{i,j} = \int_{-1}^1 t^{i+j}dt$. This is the Graham matrix associated to the polynomials $t^i$ for $0\leq i \leq k-1$. Applying Graham-Schmidt (without the renormalization step) to this set yields the basis $P_n \alpha_n$ where $\alpha_n = \frac{2^n (n!)^2}{(2n)!}$ is the inverse of the leading term of $P_n$. This polynomial has norm $\alpha_n^2 2/(2n+1)$. Therefore, the integral over the $t_i$ yields
$$
\prod_{n=0}^{k-1} \frac{2^{n+1} (n!)^2}{(n+1)(2n)!}.
$$
This completes the proof.
\end{proof}

Next we extend this result to the case that all the $\xi$ are within $\poly(k)$ of each other.
\begin{proposition}
If there exists a $\xi_0$ so that $|\xi_i-\xi_0|=\poly(k)$ for all $i$, then
$$
\det(G_{\xi_1,\ldots,\xi_k}) = 2^{\tilde O(k^2)}\prod_{i<j} \min(|\xi_i-\xi_j|^2,1).
$$
\end{proposition}
\begin{proof}
We begin by proving the lower bound. We note that for $0<x<1$,
$$
\det(G_{\xi_1,\ldots,\xi_k}) \geq \int_{-x}^x \int_{-x}^x \ldots \int_{-1}^1 \left|\det([e^{2\pi i \xi_i t_j}]_{i,j}) \right|^2 dt_1dt_2\ldots dt_k = x^k\det(G_{\xi_1/x,\xi_2/x,\ldots,\xi_k/k}).
$$
Taking $x=1/\poly(k)$, we may apply the above Lemma to compute the determinant on the right hand side, yielding an appropriate lower bound.

To prove the lower bound, we note that we can divide our $\xi_i$ into clusters, $\mathcal{C}_i$, where for any $i,j$ in the same cluster $|\xi_i-\xi_j|<1/k$ and for $i$ and $j$ in different clusters $|\xi_i-\xi_j|\geq 1/k^2$. We then note as a property of Graham matrices that
$$
\det(G_{\xi_1,\ldots,\xi_k}) \leq \prod_{\mathcal{C}_i} \det(G_{\{\xi_j\in \mathcal{C}_i\}}) = 2^{\tilde O(k^2)}\prod_{i<j,\textrm{ in same cluster}}|\xi_i-\xi_j|^2 = 2^{\tilde O(k^2)}\prod_{i<j}|\xi_i-\xi_j|^2.
$$
This completes the proof.
\end{proof}

Finally, we are ready to prove our Theorem.
\begin{proof}
Let $I(t)$ be the indicator function of the interval $[-1,1]$.

Recall that there is a function $h(t)$ so that for any function $f$ that is a linear combination of at most $k$ complex exponentials that $|h(t)f(t)|_2 = \Theta(|I(t)f(t)|_2)$ and so that $\hat{h}$ is supported on an interval of length $\poly(k)<k^C$ about the origin.

Note that we can divide our $\xi_i$ into clusters, $\mathcal{C_i}$, so that for $i$ and $j$ in a cluster $|\xi_i-\xi_j| < k^{C+1}$ and for $i$ and $j$ in different clusters $|\xi_i-\xi_j| > k^C$.

Let $\tilde G_{\xi_1,\xi_2,\ldots,\xi_k'}$ be the matrix with $(i,j)$-entry $\int_\R |h(t)|^2 e^{(2\pi i)(\xi_i-\xi_j)t} dt.$

We claim that for any $k'\leq k$ that
$$
\det(\tilde G_{\xi_1,\xi_2,\ldots,\xi_k'}) = 2^{O(k')}\det( G_{\xi_1,\xi_2,\ldots,\xi_k'}).
$$
This is because both are Graham determinants, one for the set of functions $I(t)\exp((2\pi i)\xi_j t)$ and the other for $h(t)\exp((2\pi i)\xi_j t)$. However since any linear combination of the former has $L^2$ norm a constant multiple of that the same linear combination of the latter, we have that
$$
\tilde G_{\xi_1,\xi_2,\ldots,\xi_k'} = \Theta(G_{\xi_1,\xi_2,\ldots,\xi_k'})
$$
as self-adjoint matrices. This implies the appropriate bound.

Therefore, we have that
$$
\det(G_{\xi_1,\ldots,\xi_k}) = 2^{O(k)}\det(\tilde G_{\xi_1,\ldots,\xi_k}).
$$
However, note that by the Fourier support of $h$ that
$$
\int_\R |h(t)|^2 e^{(2\pi i)(\xi_i-\xi_j)t} dt = 0 
$$
if $|\xi_i-\xi_j|>k^C$, which happens if $i$ and $j$ are in different clusters. Therefore $\tilde G$ is block diagonal and hence its determinant equals
$$
\det(\tilde G_{\xi_1,\ldots,\xi_k}) = \prod_{\mathcal{C_i}} \det(\tilde G_{\{\xi_j \in \mathcal{C}_i\}}) = 2^{O(k)}\prod_{\mathcal{C_i}} \det(G_{\{\xi_j \in \mathcal{C}_i\}}).
$$
However the Proposition above shows that
$$
\prod_{\mathcal{C_i}} \det(G_{\{\xi_j \in \mathcal{C}_i\}}) = 2^{\tilde O(k^2)}\prod_{i<j} \min(1,|\xi_i-\xi_j|^2).
$$
This completes the proof.
\end{proof}

\subsection{Proofs of Lemma \ref{lem:bound_residual_polynomial_coefficients} and Lemma \ref{lem:existence_poly_k_roots}}\label{sec:supplement_proof_poly}
We fix $z_1,\cdots,z_k$ to be complex numbers on the unit circle and use $Q(z)$ to denote the degree-$k$ polynomial $\overset{ k } { \underset{i=1}{ \prod } } (z- z_i)$.
\restate{lem:bound_residual_polynomial_coefficients}
\begin{proof}
By definition, $r_{n,k}(z_i)=z_i^n$. From the polynomial interpolation, we have
\begin{equation*}
r_{n,k}(z) = \sum_{i=1}^k \frac{  \underset{ j\in [k] \backslash i }{ \prod}  (z-z_j) z_i^n }{ \underset{j \in [k] \backslash i}{ \prod } (z_i - z_j) }.
\end{equation*}
Let $\Sym_{S,i}$ be the symmetry polynomial of $z_1,\cdots,z_k$ with
degree $i$ among subset $S \subseteq [k]$, i.e., $\Sym_{S,i}=
\underset{ S' \subseteq \binom{S}{i} }{ \sum} \underset{j \in S'}{
  \prod} z_j$. Then \[r_{n,k}^{(l)}= (-1)^{k-1-l} \sum_{i=1}^k
\frac{\Sym_{[k]\setminus i,k-1-l} \cdot z_i^n }{ \underset{j \in [k]
    \backslash i}{ \prod } (z_i - z_j) }.\] We omit $(-1)^{k-1-l}$ in
the rest of proof and use induction on $n,k,$ and $l$ to prove
$|r_{n,k}^{(l)}|\le \binom{k-1}{l}\binom{n}{k-1}$.

Base Case of $n$: For any $n<k$, from the definition, $r(z) = z^n$ and $|r_{n,k}^{(l)}| \le 1$.

Suppose it is true for any $n<n_0$. We consider $r_{n_0,k}^l$ from now on. When $k=1$, $r_{n,0}=z_1^n$ is bounded by $1$ because $z_1$ is on the unit circle of $\mathbb{C}$.

Given $n_0$, suppose the induction hypothesis is true for any $k<k_0$ and any $l<k$. For $k=k_0$, we first prove that $|r_{n_0,k_0}^{(k_0-1)}| \le \binom{n_0 }{ k_0  - 1}$ then prove that $|r_{n_0,k_0}^{(l)}|\le \binom{k_0-1 }{ l} \binom{n_0 }{ k_0-1}$ for $l=k_0-2, \cdots,0$.
\begin{eqnarray*}
r_{n_0,k_0}^{(k_0-1)} & = &  \sum_{i=1}^{k_0} \frac{  z_i^{n_0} }{\underset{j\in [k_0] \backslash i }{\prod} (z_i - z_j)} \\
& = & \sum_{i=1}^{k_0-1 } \frac{  z_i^{n_0} }{\underset{j\in [k_0] \backslash i }{\prod} (z_{i} - z_j)} + \frac{  z_{k_0}^{n_0} }{\underset{j\in [k_0] \backslash {k_0} }{\prod} (z_{k_0} - z_j)} \\
& = & \sum_{i=1}^{k_0-1 } \frac{  z_i^{n_0} - z_i^{n_0-1}z_{k_0} + z_i^{n_0-1}z_{k_0}  }{\underset{j\in [k_0] \backslash i }{\prod} (z_i - z_j)} + \frac{  z_{k_0}^{n_0} }{\underset{j\in k_0 \backslash {k_0} }{\prod} (z_{k_0} - z_j)} \\
& = & \sum_{i=1}^{k_0-1 } \left( \frac{    z_i^{n_0-1}   }{\underset{j\in [k_0-1] \backslash i }{\prod} (z_i - z_j)} + \frac{  z_i^{n_0-1}z_{k_0}  }{\underset{j\in k_0 \backslash i }{\prod} (z_i - z_j)} \right) + \frac{  z_{k_0}^{n_0} }{\underset{j\in k_0 \backslash {k_0} }{\prod} (z_{k_0} - z_j)} \\
& =&  \left( \sum_{i=1}^{k_0-1 } \frac{    z_i^{n_0-1}   }{\underset{j\in [k_0-1] \backslash i }{\prod} (z_i - z_j)} \right) +  \left( z_{k_0} \sum_{i=1}^{k_0 } \frac{    z_i^{n_0-1}   }{\underset{j\in k_0 \backslash i }{\prod} (z_i - z_j)} \right) \\
& = & r_{n_0 -1, k_0-1}^{(k_0-2)} + z_{k_0} \cdot r_{n_0 -1, k_0}^{(k_0-1)}
\end{eqnarray*}
Hence $|r_{n_0,k_0}^{(k_0-1)}| \le |r_{n_0 -1, [k_0-1]}^{(k_0-2)}| + |r_{n_0 -1, k_0}^{(k_0-1)}| \le \binom{n_0 -2 }{ k_0 - 2} + \binom{n_0-2 }{ k_0 -1} = \binom{n_0-1 }{ k_0 - 1}$. For $l<k_0-1$, we have
\begin{eqnarray*}
r_{n_0,k_0}^{(l)} 
& = &\sum_{i=1}^{k_0} \frac{ \Sym_{[k_0]\setminus i,k_0-1-l} \cdot z_i^{n_0} }{\underset{j\in [k_0] \backslash i }{\prod} (z_i - z_j)} \quad  \text{let~} l'=k_0-1-l\\ 
& = & \sum_{i=1}^{k_0-1} \frac{ \left(\Sym_{[k_0 - 1]\setminus i,l'} + \Sym_{[k_0 - 1]\setminus i,l'-1} \cdot z_{k_0}\right) z_i^{n_0} }{\underset{j\in [k_0] \backslash i }{\prod} (z_i - z_j)} + \frac{\Sym_{[k_0-1],l'}\cdot z_{k_0}^{n_0} }{ \underset{j<k_0}{\prod} (z_{k_0}-z_j)}\\
& = & \sum_{i=1}^{k_0-1} \frac{ \Sym_{[k_0 - 1]\setminus i,l'} \cdot (z_i - z_{k_0}) z_i^{n_0-1}+ \Sym_{[k_0 - 1]\setminus i,l'} \cdot  z_{k_0} z_i^{n_0-1}+ \Sym_{[k_0 - 1]\setminus i,l'-1} \cdot z_{k_0} z_i^{n_0} }{\underset{j\in [k_0] \backslash i }{\prod} (z_i - z_j)} \\
& & + \frac{\Sym_{[k_0-1],l'}\cdot z_{k_0}^{n_0} }{ \underset{j<k_0 }{ \prod } (z_{k_0}-z_j)}\\
& = & \sum_{i=1}^{k_0-1} \frac{ \Sym_{[k_0 - 1]\setminus i,l'} \cdot (z_i - z_{k_0}) z_i^{n_0-1}+ \Sym_{[k_0 - 1],l'} \cdot z_{k_0} z_i^{n_0-1} }{\underset{j\in [k_0] \backslash i }{\prod} (z_i - z_j)} + \frac{\Sym_{[k_0-1],l'}\cdot z_{k_0}^{n_0} }{  \underset{j<k_0}{ \prod }(z_{k_0}-z_j)}\\
& = & \sum_{i=1}^{k_0-1} \frac{ \Sym_{[k_0 - 1]\setminus i,l'} z_i^{n_0-1}}{\underset{j\in [k_0-1] \backslash i }{\prod} (z_i - z_j)} + \sum_{i=1}^{k_0-1} \frac{\Sym_{[k_0 - 1],l'} \cdot z_{k_0} z_i^{n_0-1}}{\underset{j\in [k_0] \backslash i }{\prod} (z_i - z_j)} + \frac{\Sym_{[k_0-1],l'}\cdot z_{k_0}^{n_0} }{ \underset{ j<k_0} { \prod }(z_{k_0}-z_j)}\\
& = & r_{n_0-1,k_0-1}^{(l-1)} + \Sym_{[k_0-1],k_0-1-l} \cdot z_{k_0} \cdot r_{n_0-1,k_0}^{(k_0-1)}
\end{eqnarray*}
By induction hypothesis, $|r_{n_0,k_0}^{(l)}| \le \binom{k_0 -2 }{ l-1} \binom{n_0 -1 }{ k_0 -2 } + \binom{k_0 - 1 }{ l}\binom{n_0 - 1 }{ k_0 - 1} \le \binom{k_0 -1 }{ l}\binom{n_0 }{ k_0-1}.$
\end{proof}
Now we finish the proof of Lemma \ref{lem:existence_poly_k_roots}.
\restate{lem:existence_poly_k_roots}

Let $m=10 k^2 \log k$ and $\mathcal{P}$ denote a set of polynomials that has degree at most $m$, and all the coefficients are integers chosen from $\{-5,\cdots,-1, 0, 1, \cdots,5\}$, i.e., 
\begin{equation*}
\mathcal{P}:=\left\{P(z)=\sum_{i=0}^{m} \alpha_i z^i ~|~ \forall i \in\{ 0,1,\cdots,m\}, |\alpha_i| \le 2 \right\}.
\end{equation*}
\begin{claim}
There exists $P^*(z)= \overset{m }{ \underset{i=0} { \sum}}   \alpha_i z^i$  with coefficient $|\alpha_i| \le 10$ for any $i \in \{0,1,\cdots,m\}$,  such that every coefficient of $P^*(z) \mod Q(z)$ is bounded by $2^{-m}$.
\end{claim}
\begin{proof}
For $P(z)= \overset{m}{ \underset{i=0}{ \sum} } \alpha_i z^i \in \mathcal{P}$, $P(z) \mod Q(z) \equiv \overset{m}{ \underset{i=0}{\sum} } \alpha_i r_{n,k}(z)$ from the definition $r_{n,k}(z)$. Hence \[P(z) \mod Q(z) = \sum_{i=0}^m \alpha_i \sum_{l=0}^{k-1} r_{i,k}^{(l)}z^l = \sum_{l=0}^{k-1} z^l\sum_{i=0}^m \alpha_i r_{i,k}^{(l)}.\]
Each coefficient in $P(z) \mod Q(z)$ is bounded by $| \overset{m }{ \underset{i=0}{ \sum}} \alpha_i r_{i,k}^{(l)}| \le 5 \overset{m }{ \underset{i=0}{ \sum}} 2^k i^{k-1} \le 2^k m^k.$ 

At the same time, $|\mathcal{P}|=11^m$. From the pigeonhole principle and $\left(\frac{2^k m^k}{(2^{-m})^2}\right)^k<11^m$, there exists $P_1 , P_2 \in \mathcal{P}$ such that for $P^*(z)=P_1(z)-P_2(z)$, $P^*(z) \mod Q(z) = \overset{k-1}{\underset{i=0}{\sum}} \gamma_i z^i$ where each coefficient $|\gamma_i| \le 2^{-m}$.
\end{proof}
Let $r(z)= \overset{k-1}{\underset{i=0}{\sum}} \gamma_i z^i=P^*(z) \mod Q(z)$ for convenience. If $P^*(0)$ (the constant term of $P^*$) is nonzero, then $|P^*(0) - r(0)| \ge 0.99$ from the above lemma. Therefore the polynomial $\frac{P^*(z)-r(z)}{P^*(0) - r(0)}$ satisfies the three properties in Lemma \ref{lem:existence_poly_k_roots}.

Otherwise, we assume $z^l$ is the first term in $P^*(z)$ with a non-zero coefficient. Let 
\[
r_{-l,k}(z) = z^{-l} \mod Q(z) =\sum_{i=1}^k \frac{  \underset{ j\in [k] \backslash i }{ \prod}  (z-z_j) z_i^{-l} }{ \underset{j \in [k] \backslash i}{ \prod } (z_i - z_j) }.\]
For convenience, we use $z_{S}= \underset{ i \in S}{\prod}  z_i$ for any subset $S \subseteq [k]$. Notice that $z_i^{-l}=\frac{z^l_{[k]\setminus i}}{z_{[k]}^l}$. Hence $r_{-l,k}(z)=r'_{l,k}(z) / z_{[k]}^l$ where $r'$ is the polynomial for $k$ units roots $z_{[k] \setminus 1},\cdots, z_{[k] \setminus k}$. So each coefficients of $r$ is still bounded by $2^k l^k$, which is less than $2^{-m/2}$.

Eventually we choose $P^*(z)/z^l - r(z) \cdot r_{-l,k}(z)$ and renormalize it to satisfy the three properties in Lemma \ref{lem:existence_poly_k_roots}.


\subsection{Proof of Lemma \ref{lem:derivative_poly}}\label{sec:proof_lem_derivative_poly}
\restate{lem:derivative_poly}
Given a degree $d$ polynomial $P(x)$, we rewrite $P(x)$ as a linear combination of the Legendre polynomials: $$P(x)=\sum_{i=0}^d \alpha_i L_i(x).$$
We use $F_i(x)=(1-x^2)L'_i(x)$ for convenience. From the definition of the Legendre polynomials in the Equation \eqref{eq:Legendre_differential}, $F'_i(x)=-i(i+1) \cdot L_i(x)$ and $F''_i(x)=-i(i+1) \cdot L'_i(x)$.

Hence we have 
\begin{align*}
\int_{1}^{-1} (1-x^2) |P'(x)|^2 \mathrm{d} x  =\quad & \int_{1}^{-1} (1-x^2) P'(x) \cdot \overline{P'}(x) \mathrm{d} x\\
 =\quad & \int_1^{-1} \left( \sum_{i \in [d]} \alpha_i F_i(x) \right) \cdot \left( \sum_{i \in [d]} \overline{\alpha_i} \frac{-F''_i(x)}{i(i+1)} \right) \mathrm{d} x\\
 =\quad & \left( \sum_{i \in [d]} \alpha_i F_i(x) \right) \cdot \left( \sum_{i \in [d]} \overline{\alpha_i} \frac{-F'_i(x)}{i(i+1)} \right) \bigg{|}_{-1}^1 \\
  +\quad & \int_1^{-1} \left( \sum_{i \in [d]} \alpha_i F'_i(x) \right) \cdot \left( \sum_{i \in [d]} \overline{\alpha_i} \frac{F'_i(x)}{i(i+1)} \right) \mathrm{d} x\\
 = \quad& \int_1^{-1} \left( \sum_{i \in [d]} \alpha_i \cdot i(i+1) \cdot L_i(x) \right) \cdot \left( \sum_{i \in [d]} \overline{\alpha_i} \frac{i(i+1) \cdot L_i(x)}{i(i+1)} \right) \mathrm{d} x\\
 =\quad & \sum_{i\in [d]} |\alpha_i|^2 i(i+1) \|L_i\|_T^2\\
 \le \quad & d(d+1) \|P\|^2_T
\end{align*}

\subsection{Proof of Lemma \ref{lem:permutation}}\label{sec:proof_permutation}
\restate{lem:permutation}
\begin{proof}
Let's compute the Fourier Transform of $(P_{\sigma,a,b}x)(t)$,
\begin{align*}
& \widehat{P_{\sigma,a,b} x}(f) \\
= & \int_{-\infty}^{+\infty} (P_{\sigma,a,b}(x))(t) e^{-2\pi \i ft} \mathrm{d} t \\
= & \int_{-\infty}^{+\infty} x(\sigma(t-a))e^{-2\pi \i \sigma b t} e^{-2\pi \i f t} \mathrm{d} t \\
= & e^{-2\pi\i (\sigma ab + fa)} \int_{-\infty}^{+\infty} x(\sigma(t-a))e^{-2\pi \i \sigma b (t-a)} e^{-2\pi \i f (t-a)} \mathrm{d} t &\text{~by~shifting~$t$~by~$a$}\\
= & e^{-2\pi\i (\sigma ab + fa)} \int_{-\infty}^{+\infty} x(\sigma t)e^{-2\pi \i \sigma b t} e^{-2\pi \i f t} \mathrm{d} t &\text{~by~replacing~$t-a$~by~$t$}\\
= & \frac{1}{\sigma} e^{-2\pi\i (\sigma ab + fa)} \int_{-\infty}^{+\infty} x(\sigma t)e^{-2\pi \i b \sigma t} e^{-2\pi \i f \sigma t /\sigma} \mathrm{d} \sigma t \\
= & \frac{1}{\sigma} e^{-2\pi\i (\sigma ab + fa)} \int_{-\infty}^{+\infty} x( t)e^{-2\pi \i (b+f/\sigma)   t} \mathrm{d}   t  &\text{~by~replacing~$t\sigma$~by~$t$} \\
= & \frac{1}{\sigma} e^{-2\pi\i a\sigma (f/\sigma + b )} \widehat{x}(f/\sigma+b)    &\text{~by~definition~of~FT}
\end{align*}
The first result follows immediately by replacing $f/\sigma +b$ by $f'$, which gives
\begin{equation*}
\widehat{P_{\sigma,a,b} x}(\sigma(f'-b)) = \frac{1}{\sigma}e^{-2\pi \i a \sigma f'} \widehat{x}(f').
\end{equation*}
Thus, we complete the proof of this Lemma.
\end{proof}

\subsection{Proof of Lemma \ref{lem:max_is_bounded_imply_sample_is_small} }\label{sec:proo_max_is_bounded_imply_sample_is_small}
\restate{lem:max_is_bounded_imply_sample_is_small}
\begin{proof}
Let $M$ denote  $\underset{t\in [0,T] }{\max} |x(t)|^2$. Replacing $X_i$ by $ \frac{ |x(t_i)|^2}{ M}$ and $n$ by $|S|$ in Lemma \ref{lem:chernoff_bound}, we obtain that
\begin{eqnarray*}
&&\mathsf{Pr}[ | X - \mu | > \epsilon \mu ] \leq 2 \exp(-\frac{\epsilon^2}{3} \mu) \\
&\implies & \mathsf{Pr} \left[ \left| \sum_{i\in S} \frac{ |x(t_i)|^2 }{ M } -  |S|\frac{\| x(t) \|_T^2}{ M}    \right| > \epsilon  |S|\frac{\| x(t) \|_T^2}{  M} \right] \leq 2\exp(-\frac{\epsilon^2}{3}\mu) \\
&\implies & \mathsf{Pr} \left[ \left| \frac{1}{|S|} \sum_{i\in S} |x(t_i)|^2- \| x(t) \|_T^2  \right|  \geq \epsilon \|x(t)\|_T^2 \right] \leq 2\exp(-\frac{\epsilon^2}{3}\mu) \\
&\implies & \mathsf{Pr} \left[ \left| \frac{1}{|S|} \sum_{i\in S} |x(t_i)|^2- \| x(t) \|_T^2  \right|  \geq \epsilon \|x(t)\|_T^2 \right] \leq 2\exp(-\frac{\epsilon^2}{3} |S|\frac{ \|x(t)\|_T^2}{  M}  )  
\end{eqnarray*}
which is less than $2\exp(-\frac{\eps^2}{3} |S| /d)$, thus completes the proof.
\end{proof}

\subsection{Proof of Lemma \ref{lem:max_average_square}}\label{sec:proof_max_average_square}
\restate{lem:max_average_square}
\begin{proof}
Let $t^*= \underset{ t \in [S,T] }{ \arg\max } |P(t)|^2$. If $t^* \in (S,T)$, then it is enough to prove that \[|P(t^*)|^2 \le (d+1)^2 \frac{1}{t^*-S}\int_{S}^{t^*} |P(x)|^2 \mathrm{d} x \text{ and } |P(t^*)|^2 \le (d+1)^2 \frac{1}{T-t^*}\int_{t^*}^T |P(x)|^2 \mathrm{d} x\] on the two intervals $[S,t^*]$ and $[t^*,T]$ separately. 

Without loss of generality, we will prove the inequality for $S=-1$ and $t^*=T=1$. We find the minimum $\|P(x)\|^2_T$ assuming $|P(1)|^2 = 1$. Because the first $(d+1)$ Legendre polynomials provide a basis of polynomials of degree at most $d$ and their evaluation $L_n(1) = 1$ for any $n$,  we consider:
\begin{eqnarray*}
\underset{ \alpha_0, \alpha_1, \cdots, \alpha_{d}\in \mathbb{C} }{\min} &&  \int_{-1}^{1} |P(x) |^2 \mathrm{d} x\\
\text{s.t.} && P(x)=\sum_{i=0}^{d} \alpha_i L_i(x)\\
			&& |P(1)|=|\sum_{i=0}^{d} \alpha_i|= 1.
\end{eqnarray*}
We simplify the integration of $P(x)^2$ over $[-1,1]$ by the orthogonality of Legendre polynomials:
\begin{eqnarray*}
\int_{-1}^1 |P(x)|^2 \mathrm{d} x 
& = & \int_{-1}^1 \left( \sum_{i=0}^{d} \alpha_i L_i(x) \right) \cdot \left( \sum_{j=0}^{d} \overline{\alpha_j} \overline{L_j}(x)  \right) \mathrm{d} x \\
& = & \int_{-1}^1 \sum_{i=0}^{d} |\alpha_i|^2  L_i(x)^2 + \sum_{i\neq j} \alpha_i \overline{\alpha_j} L_i(x) \overline{L_j}(x) \mathrm{d} x\\
& = &  \sum_{i=0}^{d} |\alpha_i|^2 \frac{2}{2i+1}  \text{~by~Fact~\ref{fac:legendre_inner_product}}
\end{eqnarray*}

Using $\int_{-1}^1 |P(x)|^2 \mathrm{d} x = \overset{d}{ \underset{i=0}{ \sum} }  |\alpha_i|^2 \frac{2}{2i+1}$, we simplify the optimization problem to
\begin{eqnarray*}
\underset{ \alpha_0, \alpha_1, \cdots, \alpha_{d}\in \C }{\min} &&  \sum_{i=0}^{d} |\alpha_i|^2 \frac{2}{2i+1} \\
\text{s.t.} && \left|\sum_{i=0}^{d} \alpha_i \right|= 1 
\end{eqnarray*}
From the Cauchy-Schwarz inequality, we have 
\begin{equation*}
\left|\sum_{i=0}^{d} \alpha_i \right|^2 \le \left(\sum_{i=0}^{d} |\alpha_i|^2 \frac{2}{2i+1} \right) \left(\sum_{i=0}^{d} \frac{2i+1}{2} \right).
\end{equation*}
Therefore $\overset{d}{ \underset{i=0}{\sum}} |\alpha_i|^2 \frac{2}{2i+1} \ge \frac{2}{(d+1)^2}$ and $|P(1)|^2 \le (d+1)^2 \cdot \frac{1}{2} \int_{-1}^{1} |P(x) |^2 \mathrm{d} x$.
\end{proof}

\section{Known Facts}\label{sec:known_facts}
This section provides a list of well-known facts existing in literature.

\subsection{Inequalities}
We state the H\"older's inequality for complex numbers. We will use the corresponding version $p=q=2$ of Cauchy-Schwarz inequality for complex numbers.
\begin{lemma}[H\"older's inequality]
If $S$ is a measurable subset of $\mathbb{R}^n$ with the Lebesgue measure, and $f$ and $g$ are measurable complex-valued functions on $S$, then
\begin{equation*}
\int_S |f(x) g(x) | \mathrm{d} x \leq (\int_S |f(x)|^p \mathrm{d}x )^{\frac{1}{p} } ( \int_S |g(x)|^q \mathrm{d} x )^{\frac{1}{q}}
\end{equation*}
\end{lemma}


\begin{lemma}[Chernoff Bound \cite{tarjan09},\cite{chernoff1952} ]\label{lem:chernoff_bound}
Let $X_1, X_2, \cdots, X_n$ be independent random variables. Assume that $0\leq X_i \leq 1$ always, for each $i \in [n]$. Let $X= X_1+X_2+\cdots+X_n$ and $\mu = \mathbb{E}[X] = \overset{n}{  \underset{i=1}{\sum} } \mathbb{E}[X_i]$. Then for any $\epsilon>0$,
\[\mathsf{Pr} [ X \geq (1+\epsilon) \mu ] \leq \exp(-\frac{\epsilon^2 }{2+\epsilon} \mu) \textit{ and } \mathsf{Pr} [ X \geq (1-\epsilon) \mu ] \leq \exp(-\frac{\epsilon^2 }{2} \mu).\]
\end{lemma}

\subsection{Linear regression}\label{sec:l2_regression}
Given a linear subspace $\mathrm{span}\{\vec{v}_1,\cdots,\vec{v}_d\}$ and $n$ points, we always use $\ell_2$-regression to find a vector as the linear combination of $\vec{v}_1,\cdots,\vec{v}_d$ that minimizes the distance of this vector to those $n$ points.
\begin{fact}\label{fac:basic_l2_regression}
Given an $n\times d$ matrix $A$ and an $n\times 1$ column vector $b$ , it takes $O(nd^{\omega-1})$ time to output an $x'$ such that
\begin{equation*}
x' = \underset{x}{ \arg \min } \| Ax - b\|_2.
\end{equation*}
where $\omega$ is the exponent of matrix multiplication\cite{W12}.
\end{fact}
Notice that weighted linear regression can be solved by linear regression solver as a black-box.

\begin{algorithm}[t]
\caption{Linear regression algorithms}\label{alg:linear_regression}
\begin{algorithmic}[1]
\Procedure{\textsc{LinearRegression}}{$A,b,$} --- Fact \ref{fac:basic_l2_regression}
	\State $x' \leftarrow \underset{x}{\arg \min} \| Ax - b\|_2$.
	\State \Return $x'$
\EndProcedure
\Procedure{\textsc{LinearRegressionW}}{$A,b,w$} --- Fact \ref{fac:basic_l2_regression}
	\State $x' \leftarrow \underset{x}{\arg \min} \overset{d}{\underset{i=1}{\sum}} w_i | (Ax)_i - b_i |^2$.
	\State \Return $x'$
\EndProcedure
\end{algorithmic}
\end{algorithm}

\begin{algorithm}[t]
\caption{Multipoint evaluation of a polynomial}\label{alg:multipoint_evaluation}
\begin{algorithmic}[1]
\Procedure{\textsc{MultipointEvaluation}}{$P,\{t_1,t_2,\cdots, t_d\}$} --- Fact \ref{fac:multipoint_evaluation_of_polynomial}
	\State \Return $P(t_1), P(t_2), \cdots, P(t_d)$
\EndProcedure
\end{algorithmic}
\end{algorithm}

\subsection{Multipoint evaluation of a polynomial}
Given a degree-$d$ polynomial, and $n$ locations. The naive algorithm of computing the evaluations at those $n$ locations takes $O(nd)$. However, the running time can be improved to $O(n \poly(\log d))$ by using this well-known result,
\begin{fact}[\cite{BS12}]\label{fac:multipoint_evaluation_of_polynomial}
Given a degree-$d$ polynomial $P(t)$, and a set of $d$ locations $\{t_1,t_2,\cdots, t_d\}$. There exists an algorithm that takes $O(d\log^c d)$ time to output the evaluations $\{ P(t_1), P(t_2), \cdots, P(t_d)\}$, for some constant $c$.
\end{fact}

\section{Analysis of Hash Functions and Filter Functions }\label{sec:proof_hashfilter}

\subsection{Analysis of filter function \texorpdfstring{$(H(t), \widehat{H}(f))$} \quad}\label{sec:properties_of_H}
We construct the Filter function $(H(t), \widehat{H}(f))$ in this section. 

We fix the interval to be $\supp(\rect_1)=[-1/2,1/2]$ instead of $[0,T]$ for convenience. We first define the filter function $H_1(t)$ which preserves the energy of a $k$-Fourier-sparse signal $x^*$ on $[-1/2,1/2]$ to the signal $H_1 \cdot x^*$ on $[-\infty,+\infty]$. 
\begin{definition}
Let $s_1 = \Theta( k^4 \log^4 k )$, $\ell = \Omega( k \log k /\delta)$ be a even number, and $s_0=C_0 s_1\sqrt{\ell}$ for some constant $C_0$ that will normalize $H_1(0)=1$. Recall that $\rect_s(t)=1$ iff $|t| \le s/2$ and $\wh{\rect_s(f)}=\sinc(fs)=\frac{\sin(\pi fs)}{\pi fs}$. 

We define the filter function $H_1(t)$ and its Fourier transform $\wh{H}_1(f)$ as follows:
\begin{eqnarray}
\widehat{H}_1(f) & = & s_0 \cdot \left( \rect_{s_1}(f) \right)^{*\ell} \cdot \sinc\left( f s_2\right), \notag \\ 
 & =& s_0 \cdot \bigl( {\vcenter{\hbox{\includegraphics[width=0.2\textwidth]{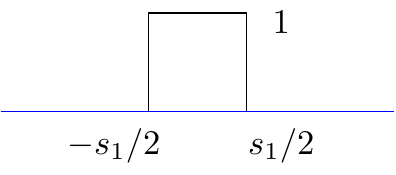}}}} \bigr)^{*\ell} \cdot {\vcenter{\hbox{\includegraphics[width=0.4\textwidth]{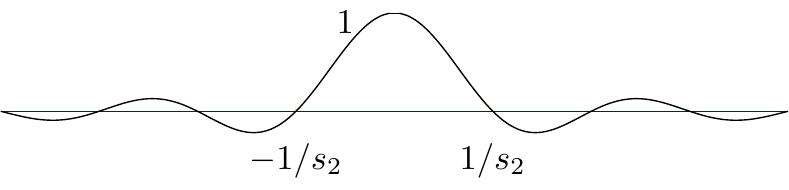}}}}  \notag\\
H_1(t) &= & s_0 \cdot \left(\sinc( s_1 t) \right)^{\cdot \ell} * \rect_{s_2}(t) \notag \\
&= & s_0 \cdot \bigl( {\vcenter{\hbox{\includegraphics[width=0.4\textwidth]{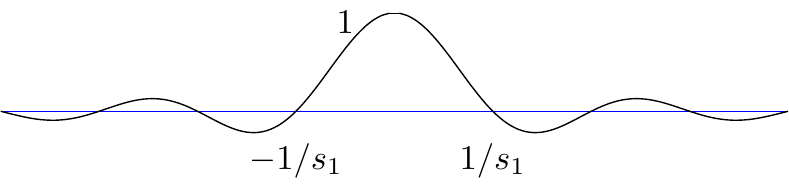}}}} \bigr)^{\cdot \ell} * {\vcenter{\hbox{\includegraphics[width=0.2\textwidth]{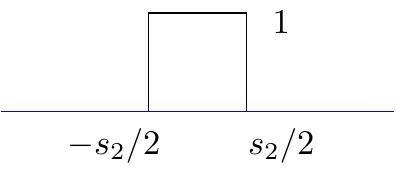}}}}  \notag
\end{eqnarray}
where $s_0$ is a fixed parameter s.t. $H_1(0)=1$. 
\end{definition}
We provide some basic properties about our filter function. Notice that $\sinc(t)=\frac{\sin(\pi t)}{\pi t}$ ( $\sinc(0)$ is defined to be $1$ ) has the following properties (shown in Figure \ref{fig:sinc_property}):

\begin{enumerate}
\item $\forall t \in \mathbb{R}, 1 - \frac{(\pi t)^2}{3!} \le \sinc(t) \le 1$.
\item $\forall |t| \le 1.2/\pi, \sinc(t) \le 1 - \frac{t^2}{8}$.
\item $\forall |t|>1.2/\pi, |\sinc(t)| \le \frac{1}{\pi |t|}$.
\end{enumerate}

\begin{figure}
 \centering
    \includegraphics[width=0.8\textwidth]{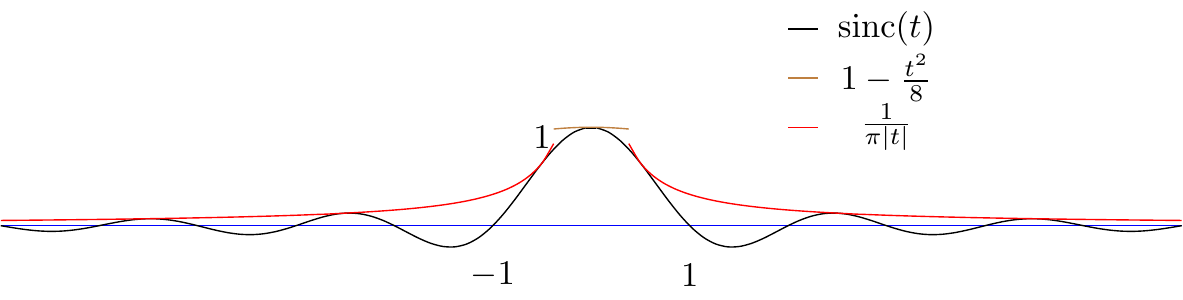}
    \caption{The Property of $\sinc(t)$.}\label{fig:sinc_property}
\end{figure}

\begin{claim}\label{clm:concentrate_around_0}
$\int_{-\frac{1.2}{\pi s_1}}^{\frac{1.2}{\pi s_1}} ( \sinc(s_1 t) )^{\ell} \mathrm{d} t \eqsim \frac{1}{s_1\sqrt{\ell}}$.
\end{claim}
\begin{proof}
We use the above properties for the sinc function to prove the upper bound:
\begin{align*}
\int_{-\frac{1.2}{\pi s_1}}^{+\frac{1.2}{\pi s_1}} ( \sinc(s_1 t) )^{\ell} \mathrm{d} t & = \frac{1}{s_1} \int_{-1.2/\pi}^{+1.2/\pi} ( \sinc( t) )^{\ell} \mathrm{d} t \\
& = \frac{2}{s_1} \left( \int_{0}^{\sqrt{8/\ell} }( \sinc{t} )^\ell \mathrm{d} t + \int_{\sqrt{8/\ell}}^{1.2/\pi} (\sinc(t))^{\ell} \mathrm{d} t \right)\\
& \le \frac{2}{s_1} \left( \sqrt{8/\ell} + \sum_{i=1}^{\frac{1.2/\pi}{\sqrt{8/\ell}}-1}\int_{i\sqrt{8/\ell}} ^{(i+1) \sqrt{8/\ell} }(1- x^2/8 )^\ell \mathrm{d} x\right) \\
& \le \frac{2}{s_1} \left( \sqrt{8/\ell} + \sum_{i=1}^{\frac{1.2/\pi}{\sqrt{8/\ell}}-1} \sqrt{8/\ell} \cdot 2^{-i^2} \right) \\
& \lesssim \frac{1}{s_1\sqrt{\ell}}.
\end{align*}
We prove the lower bound:
\begin{align*}
\int_{-\frac{1.2}{\pi s_1}}^{\frac{1.2}{\pi s_1}} ( \sinc(s_1 t) )^{\ell} \mathrm{d} t & = \frac{2}{s_1} \left( \int_{0}^{\sqrt{8/\ell} }( \sinc{t} )^\ell \mathrm{d} t + \int_{\sqrt{8/\ell}}^{1.2/\pi} (\sinc(t))^{\ell} \mathrm{d} t \right)\\
& \ge \frac{2}{s_1} \left( \int_{0}^{\sqrt{8/\ell} }( 1 - \frac{\pi^2 t^2}{6})^\ell \mathrm{d} t \right)  \\
& \gtrsim \frac{1}{s_1\sqrt{\ell}}.
\end{align*}
\end{proof}
We bound the integration outside $[-\frac{1.2}{\pi s_1},\frac{1.2}{\pi s_1}]$ from the last property of the sinc function.
\begin{claim}\label{clm:the_rest_is_tiny}
$\int_{\frac{1.2}{\pi s_1}}^{+\infty} ( \sinc(s_1 t) )^{\ell} \mathrm{d} t = O(1.2^{-\ell})$.
\end{claim}
From these two claims, we have the existence of $s_0$.
\begin{claim}
There exists a universal constant $C_0$ and $s_0=C_0 s_1\sqrt{\ell}$ such that $H_1(0)=1$.
\end{claim}
\begin{proof}
Because $\ell$ is a large even number, $\int_{\rect_{1-2/s_1}} \sinc(s_1 t)^{\ell} \mathrm{d} t \eqsim \frac{1}{s_1\sqrt{\ell}}$ from all discussion above.
\end{proof}



We show several useful properties about the Filter functions $\big(H_1(t),\wh{H_1}(f)\big)$.
\begin{lemma}\label{lem:property_of_filter_H1}
Given $s_0, s_1, s_2, \ell$, where $ \frac{s_2}{2}+\frac{1}{s_1} \leq 1/2$ and $s_0 = C_0 s_1 \sqrt{\ell}$ for some constant $C_0$. The filter function $(H_1(t),\widehat{H}_1(f))[s_0,s_1,s_2,\ell]$ has the following properties,
\begin{eqnarray*}
&\mathrm{Property~\RN{1}} : & H_1(t) \in [1- \frac{s_0}{s_1} \cdot  \frac{2 \pi^{-\ell}}{\ell-1} , 1 ], \text{~if~} |t| \leq  \frac{s_2}{2} - \frac{1}{s_1}. \\
&\mathrm{Property~\RN{2}} : & H_1(t) \in [0,1], \text{~if~}   \frac{s_2}{2} -\frac{1}{s_1} \leq |t| \leq \frac{1}{2} \\
&\mathrm{Property~\RN{3}} : & H_1(t) \leq s_0  s_2\left(  (s_1 |t|-  s_1  + 2 )^2 + 1\right)^{-\ell} ,\forall |t| > \frac{1}{2} \\
&\mathrm{Property~\RN{4}} : & \supp(\widehat{H}_1(f) ) \subseteq [-\frac{s_1  \ell}{2}, \frac{s_1 \ell }{2}] 
\end{eqnarray*}
\end{lemma}

\begin{figure}[t]
  \centering
    \includegraphics[width=1.0\textwidth]{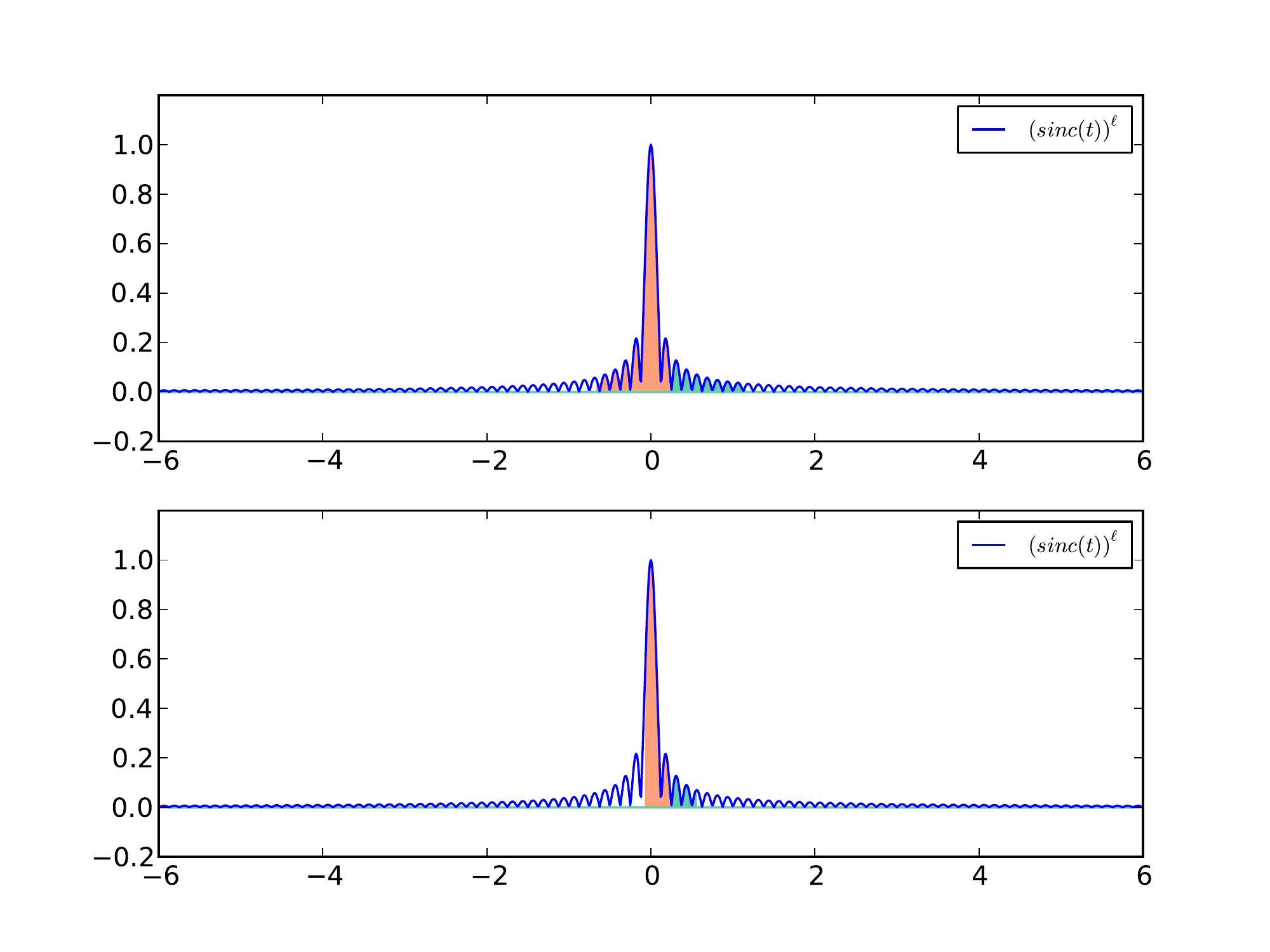}
    \caption{The light red area represents $\int_{( 1/2-2/s_1) }^{  ( 1/2-2/s_1)-t } s_0 \cdot \sinc \left(s_1 (\tau)\right)^{\ell} \mathrm{d}\tau$ and the light green area represents $\int_{ 1/2-2/s_1 }^{ 1/2-2/s_1 +t} s_0 \cdot \sinc \left(s_1 (\tau)\right)^{\ell} \mathrm{d}\tau$.}\label{fig:H1_property_1_proof}
\end{figure}

\begin{proofof}{Property \RN{1}}
First,  $H_1(0) = 1$ follows by definition of $s_0$,
 then we can prove the upper bound for $H_1(t)$ by showing for any $t>0$, $H_1(0) - H_1(t) >0$ always holds  ,

 By definition of sinc function, we know that $\sinc(s_1 t)^\ell$ reaches $0$ at all the points $\{ \frac{1}{s_1} + i \frac{2}{s_1} | i\in \mathbb{N} \}$. By definition of $s_1$, we know that $\frac{1}{s_1} \ll \frac{1}{2}-\frac{1}{s_1}$. For any $t>0$,
\begin{align*}
~&  H_1(0) - H_1(t) \\
 = ~& \int s_0 \cdot \sinc \left(s_1 (\tau)\right)^{\ell}\cdot \rect_{1-2/s_1}(0-\tau) \mathrm{d}\tau - \int s_0 \cdot \sinc \left(s_1 (\tau)\right)^{\ell}\cdot \rect_{1-2/s_1}(t-\tau) \mathrm{d}\tau \\
  = ~& \int_{- ( 1/2-2/s_1) }^{ 1/2-2/s_1 } s_0 \cdot \sinc \left(s_1 (\tau)\right)^{\ell} \mathrm{d}\tau - \int_{- ( 1/2-2/s_1)+t }^{ 1/2-2/s_1 +t} s_0 \cdot \sinc \left(s_1 (\tau)\right)^{\ell} \mathrm{d}\tau \\
   = ~& \int_{- ( 1/2-2/s_1) }^{ - ( 1/2-2/s_1)+t } s_0 \cdot \sinc \left(s_1 (\tau)\right)^{\ell} \mathrm{d}\tau - \int_{ 1/2-2/s_1 }^{ 1/2-2/s_1 +t} s_0 \cdot \sinc \left(s_1 (\tau)\right)^{\ell} \mathrm{d}\tau & \text{~shown~in~Figure~\ref{fig:H1_property_1_proof}} \\
   \geq ~& 0,
\end{align*} 
where the last inequality follows by choosing $s_1$ to be an integer. Thus, we prove an upper bound for $H_1(t)$.
Third, we show the lower bound for $H_1(t)$,
\begin{eqnarray*}
H_1(t) & = & \int_{-\infty}^{+\infty} s_0 \cdot \sinc(s_1 \tau)^{\cdot \ell} \rect_{s_2}(t-\tau) \mathrm{d} \tau  \\
& = & \int_{t-\frac{s_2}{2}}^{t+\frac{s_2}{2}} s_0 \cdot \sinc(s_1 \tau)^{\cdot \ell} \mathrm{d} \tau \\
& = & 1 - \underbrace{ \int_{ t+\frac{s_2}{2} }^{ + \infty } s_0 \cdot \sinc(s_1 \tau)^{\cdot \ell} \mathrm{d} \tau }_{A} - \underbrace{ \int_{ - \infty }^{ t-\frac{s_2}{2}  } s_0 \cdot \sinc(s_1 \tau)^{\cdot \ell} \mathrm{d} \tau }_{B} \\
\end{eqnarray*}
Thus, as long as we can upper bound the term $A$ and $B$, then we will have a lower bound for the $H_1(t)$, for any $|t| \leq \frac{s_2}{2} - \frac{1}{s_1}$.
\begin{eqnarray*}
A & = & \int_{ t+\frac{s_2}{2} }^{ + \infty } s_0 \cdot \sinc(s_1 \tau)^{\cdot \ell} \mathrm{d} \tau \\
& \leq & \int_{\frac{1}{s_1}}^{+\infty} s_0 \cdot \sinc(s_1 \tau)^{\cdot \ell} \mathrm{d} \tau \\
& \leq &  \int_{\frac{1}{s_1}}^{+\infty} s_0 \cdot (s_1 \pi \tau)^{-\ell} \mathrm{d} \tau \\
& = & s_0 \cdot (s_1 \pi)^{-\ell} \frac{1}{\ell-1} (1/s_1)^{-\ell+1} \\
& = & \frac{s_0}{s_1} \cdot (\pi)^{-\ell} \frac{1}{\ell-1}
\end{eqnarray*}
Similarly, we can bound the term $B$ in the same way.
\end{proofof}

\begin{proofof}{Property \RN{2}}
In the proof of Property \RN{1}, we already show that $\forall t$, $H_1(t) \leq 1$. Thus, the upper bound of Property \RN{2} is also holding. The lower bound follows by both $\sinc(s_1t)^{\cdot \ell}$ and $\rect_{s_2}(t)$ are always nonnegative, thus the convolution of these two functions has to be nonnegative.
\end{proofof}

\begin{proofof}{Property \RN{3}}
Let's prove the case when $t>1$, since $H_1(t)$ is symmetric, then the case $t<-1$ will also hold. By definition of $(H_1(t), \widehat{H}_1(f))$, we have
\begin{eqnarray*}
H_1(t) & = & s_0 \cdot \int_{-\infty}^{+\infty} \sinc(s_1 (t-\tau))^{\cdot \ell} \rect_{s_2}(\tau) \mathrm{d} \tau \\
& = & s_0 \cdot \int_{-\frac{s_2}{2}}^{\frac{s_2}{2}} \sinc(s_1 (t-\tau))^{\cdot \ell}  \mathrm{d} \tau \\
& = & s_0 \cdot \int_{-\frac{s_2}{2}}^{\frac{s_2}{2}} \sinc(s_1 (\tau-t))^{\cdot \ell}  \mathrm{d} \tau \\
\end{eqnarray*}
We'd like to choose a middle point $\tau_0$, and then separated the interval into two parts, one is $[-\frac{s_2}{2},\tau_0]$ and the other is $[\tau_0, \frac{s_2}{2}]$. To choose a reasonable $\tau_0$, we need to use the following simple facts,
\begin{eqnarray*}
( \frac{\sin(x)}{x} )^\ell & \leq & x^{-\ell}  \text{~if~} x\geq 1.2\\
( \frac{\sin(x)}{x} )^\ell & \leq & (1- \frac{x^2}{8})^{\ell} \text{~if~} x < 1.2\\
\end{eqnarray*}
Thus, $ | \pi s_1(\tau_0 - t) | = 1.2$, which implies that $\tau_0^+ = t+ \frac{1.2}{\pi s_1}$ or $\tau_0^- = t - \frac{1.2}{\pi s_1}$. By relationship between $s_1$ and $s_2$, we know $\tau_{0}^- > \frac{1}{2}- \frac{1.2}{\pi s_1} > \frac{1}{2} -\frac{1}{s_1} \geq \frac{s_2}{2}$. Thus, we can use the case $x<1.2$ to upper bound the $H_1(t)$,
\begin{align*}
 &~H_1(t) \\
 \leq ~ & s_0 \cdot s_2 \cdot \underset{\tau \in [-s_2/2, s_2/2] }{\max} \sinc(s_1(\tau-t))^{\cdot \ell} \\
 \leq ~ & s_0 \cdot s_2 \underset{\tau \in [-s_2/2, s_2/2] }{\max}  (1- \frac{(s_1\pi(\tau-t))^2 }{8})^{\ell} \\
 = ~ & s_0 \cdot s_2  (1- \frac{(s_1\pi(\frac{s_2}{2}-t))^2 }{8})^{\ell} \\
 \leq ~ & s_0 \cdot s_2  \cdot ( e^{ - \frac{(s_1\pi(\frac{s_2}{2}-t))^2 }{8} } )^{\ell} &\text{~by~} 1-x \leq e^{-x} \\
 \leq ~ & s_0 \cdot s_2  \cdot ( e^{(s_1 (t-s_2/2))^2 } )^{-\ell} &\text{~by~} 1< \pi^2/8 \\
 \leq ~ & s_0 \cdot s_2  \cdot (1+(s_1 (t-s_2/2))^2  )^{-\ell} &\text{~by~} 1+ x\leq e^x \\
 \leq ~& s_0 \cdot  (1+(s_1 (t-s_2/2))^2  )^{-\ell} &\text{~by~} s_1 \leq 1.
\end{align*}
Thus, we complete the proof.
\end{proofof}

\begin{proofof}{Property \RN{4}} Because of the support of $\rect_{s_1} (f)$ is $s_1$, then the support of $( \rect_{s_1} (f) )^{*\ell} = s_1 \ell$. Since $\widehat{H}_1(f)$ is defined to be the  $( \rect_{s_1} (f) )^{*\ell}$ multiplied by $\sinc( f s_2)$, thus $\supp(\widehat{H}_1 (f )) \subseteq [- \frac{s_1 \ell}{2} , \frac{s_1 \ell}{2} ]$.
\end{proofof}

\begin{definition}\label{def:def_of_filter_H}
Given any $ 0< s_3 < 1$, $0<\delta <1$, we define $(H(t),\widehat{H}(f))$ to be the filter function $(H_1(t), \widehat{H_1}(f))$ by doing the following operations
\begin{itemize}
\item Setting $\ell = \Theta( k \log ( k /\delta ) )$,
\item Setting $ s_2  = 1 -\frac{2}{s_1}$,
\item Shrinking by a factor $s_3$ in time domain,
\end{itemize}
\begin{eqnarray}
H(t) &= & H_1( t/s_3) \label{eq:definition_H} \\
\widehat{H}(f) & = & s_3 \widehat{H_1}(s_3 f) \label{eq:definition_hatH}
\end{eqnarray}
We call the ``heavy cluster" around a frequency $f_0$ to be the support of $\delta_{f_0}(f) * \wh{H}(f)$ in the frequency domain and use 
\begin{equation}
\Delta_h=|\supp(\widehat{H}(f) )| = \frac{ s_1 \cdot \ell }{s_3} 
\end{equation} to denote the width of the cluster.
\end{definition}

We show several useful properties about the Filter functions $\big(H(t),\wh{H}(f)\big)$.
\restate{lem:property_of_filter_H}

The Property \RN{1}, \RN{2}, \RN{3} and \RN{4} follow by filter function $H(t), \widehat{H}(f)$ inheriting $H_1(t), \widehat{H}_1(f)$.

\begin{figure}[t]
  \centering
    \includegraphics[width=1.0\textwidth]{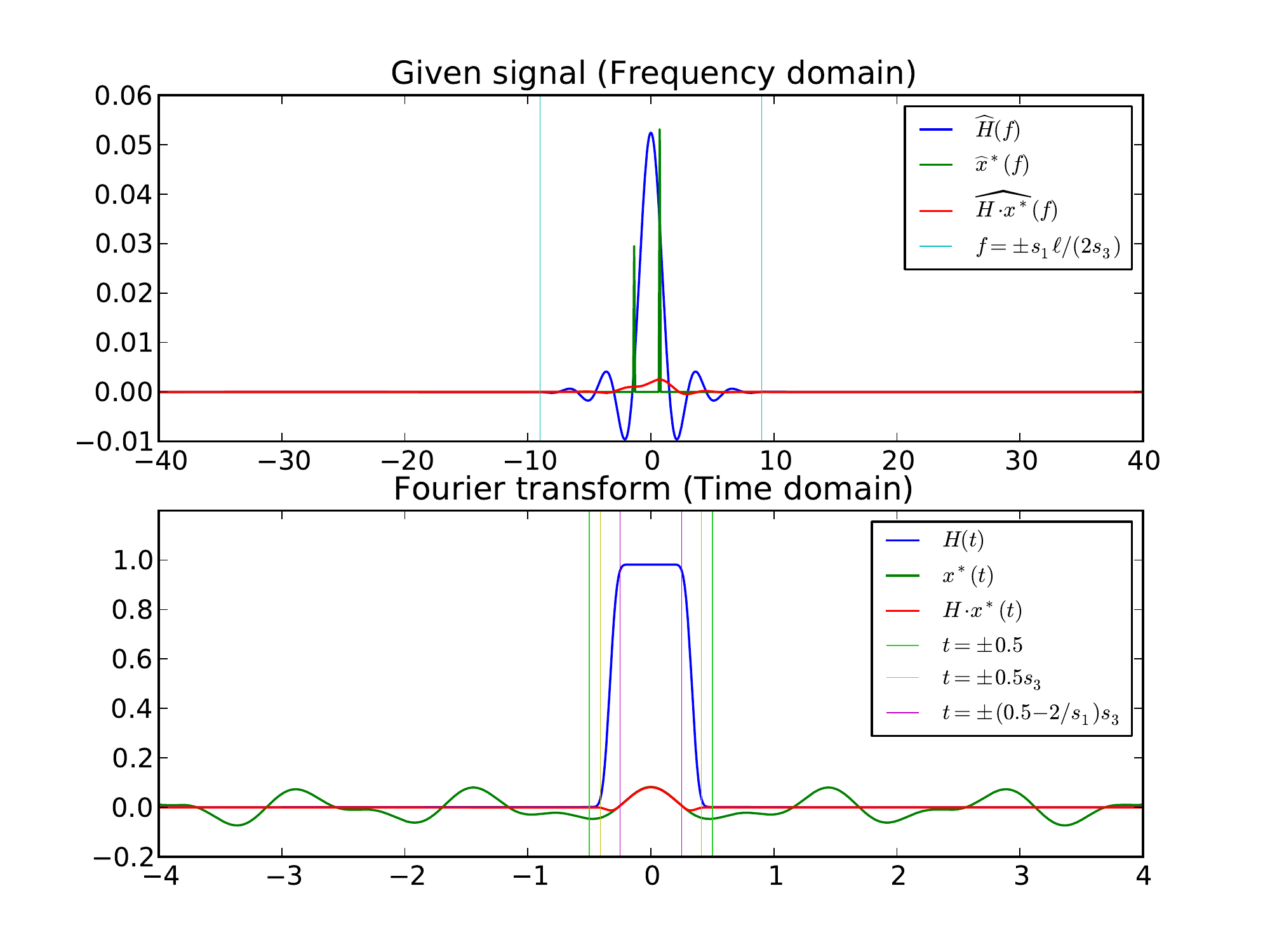}
    \caption{$\widehat{H \cdot x^*}(f)$ and $H \cdot x^*(t)$.}
\end{figure}

\begin{proofof}{Property \RN{5}}
$\forall t \notin [-1/2, 1/2]$, we have, 
\begin{align}\label{eq:why_ell_at_least_k_log_k_over_delta}
  ~& |x^*(t) \cdot H(t)|^2 \notag \\
 \leq ~ & |x^*(t)|^2 \cdot | H (t)|^2 \notag  \\
\leq ~ & |x^*(t)|^2 \cdot (s_1(|t|-1/2)+2)^{-\ell} &\text{~by~Property~\RN{3} of $H_1(t)$} \notag \\
\leq  ~& k^7 \cdot (2kt)^{2.5k} \cdot \int_{-\infty}^{+\infty} | x^*(t) \cdot \rect_1 (t) |^2 \mathrm{d} t \cdot (s_1( \frac{|t|}{s_3}- \frac{1}{2} )+2)^{-\ell} ~&\text{~by~Lemma~\ref{lem:x_dot_H_is_small_outside_T}}\notag \\
\le ~& t^{O(k\log k)}\cdot \int_{-\infty}^{+\infty} | x^*(t) \cdot \rect_1 (t) |^2 \mathrm{d} t   \cdot (s_1( \frac{|t|}{s_3}- \frac{1}{2})+2)^{-\ell/2}.
\end{align}
Thus taking the integral finishes the proof because $\ell \gtrsim k \log (k/\delta)$.
\end{proofof}

\begin{proofof}{Property \RN{6}}
 First, because of for any $t$, $|H_1(t)|\leq 1$, thus we prove the upper bound for $\text{LHS}$,
\begin{equation*}
\int_{-\infty}^{+\infty} |x^*(t) \cdot H(t) \cdot \rect_{1}(t)|^2 \mathrm{d} t \leq \int_{-\infty}^{+\infty} |x^*(t) \cdot 1 \cdot \rect_{1}(t)|^2 \mathrm{d} t.
\end{equation*}
Second, as mentioned early, we need to prove the general case when $s_3 = 1 - 1/\poly(k)$. Define interval $S = [-s_3( \frac{1}{2}- \frac{1}{s_1}), s_3( \frac{1}{2}- \frac{1}{s_1}) ]$, by definition, $S \subset [-1/2,1/2]$. Then define $\overline{S} =[-1/2,1/2]\setminus S$, which is $ [-1/2, -s_3( \frac{1}{2}- \frac{1}{s_1}) ) \cup (s_3( \frac{1}{2}- \frac{1}{s_1}) , 1/2]$. By Property \RN{1}, we have
\begin{equation}\label{eq:eq1_proof_of_property_6}
\int_S |x^*(t) \cdot H(t)|^2 \mathrm{d} t \geq (1-\delta)^2 \int_S |x^*(t)|^2 \mathrm{d} t
\end{equation}
Then we can show
\begin{align}\label{eq:eq2_proof_of_property_6} 
  ~ & \int_{\overline{S}} |x^*(t) |^2 \mathrm{d} t \notag \\ 
 \leq ~ & |\overline{S}|\cdot \underset{t\in [-1/2,1/2]}{\max} |x^*(t)|^2   \notag \\
 \leq ~ & (1-s_3(1-\frac{2}{s_1})) \cdot \wt{O}(k^4) \int_{-\frac{1}{2}}^{\frac{1}{2}} |x^*(t)|^2 \mathrm{d} t &\text{~by~Lemma~\ref{lem:max_is_at_most_polyk_times_l2}}\notag \\
 \lesssim~ & \int_{-\frac{1}{2}}^{\frac{1}{2}} |x^*(t)|^2 \mathrm{d} t & \text{~by~} \min( \frac{1}{1-s_3}, s_1 ) \geq \wt{O}(k^4) 
\end{align}
Combining Equations (\ref{eq:eq1_proof_of_property_6}) and (\ref{eq:eq2_proof_of_property_6}) gives a lower bound for $\text{LHS}$,
\begin{align*}
 & \int_{-\infty}^{+\infty} |x^*(t) \cdot H(t) \cdot \rect_{1}(t)|^2 \mathrm{d} t\\
 \geq ~ & \int_S |x^*(t) H(t)|^2 \mathrm{d} t \\
 \geq ~ & (1-2\delta) \int_S |x^*(t) |^2 \mathrm{d} t &\text{~by~Equation~(\ref{eq:eq1_proof_of_property_6})} \\
 \geq ~ & (1-2\delta) \int_{S\cup \overline{S}} |x^*(t) |^2 \mathrm{d} t - (1-2\delta) \int_{\overline{S} } |x^*(t) |^2 \mathrm{d} t \\
 \geq ~ & (1-2\delta) \int_{S\cup \overline{S}} |x^*(t) |^2 \mathrm{d} t - (1-2\delta) \epsilon \int_{S \cup \overline{S} } |x^*(t)|^2 \mathrm{d} t & \text{~by~Equation~(\ref{eq:eq2_proof_of_property_6})}\\
 \geq ~ & (1-2\delta - \epsilon ) \int_{-\frac{1}{2}}^{\frac{1}{2}} |x^*(t)|^2 \mathrm{d} t \\
 \geq ~ & (1-2\epsilon) \int_{-\infty}^{+\infty} | x^*(t) \cdot \rect_1(t) |^2 \mathrm{d} t  &\text{~by~}\epsilon \gg \delta
\end{align*}

\end{proofof}

\begin{figure}[!th]
  \centering
  
    \includegraphics[width=0.8\textwidth]{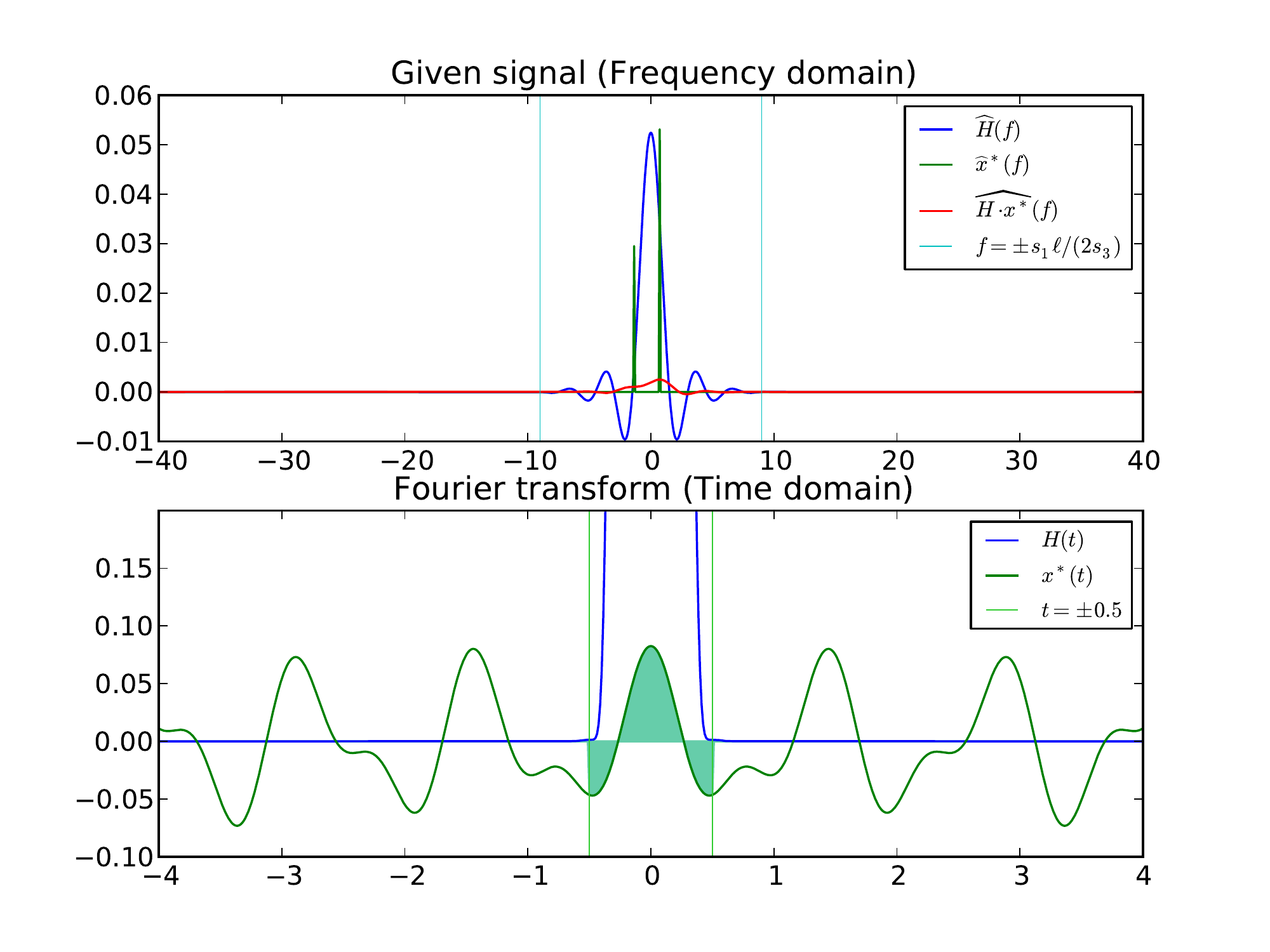}
    \vspace{-0.2cm}
    \includegraphics[width=0.8\textwidth]{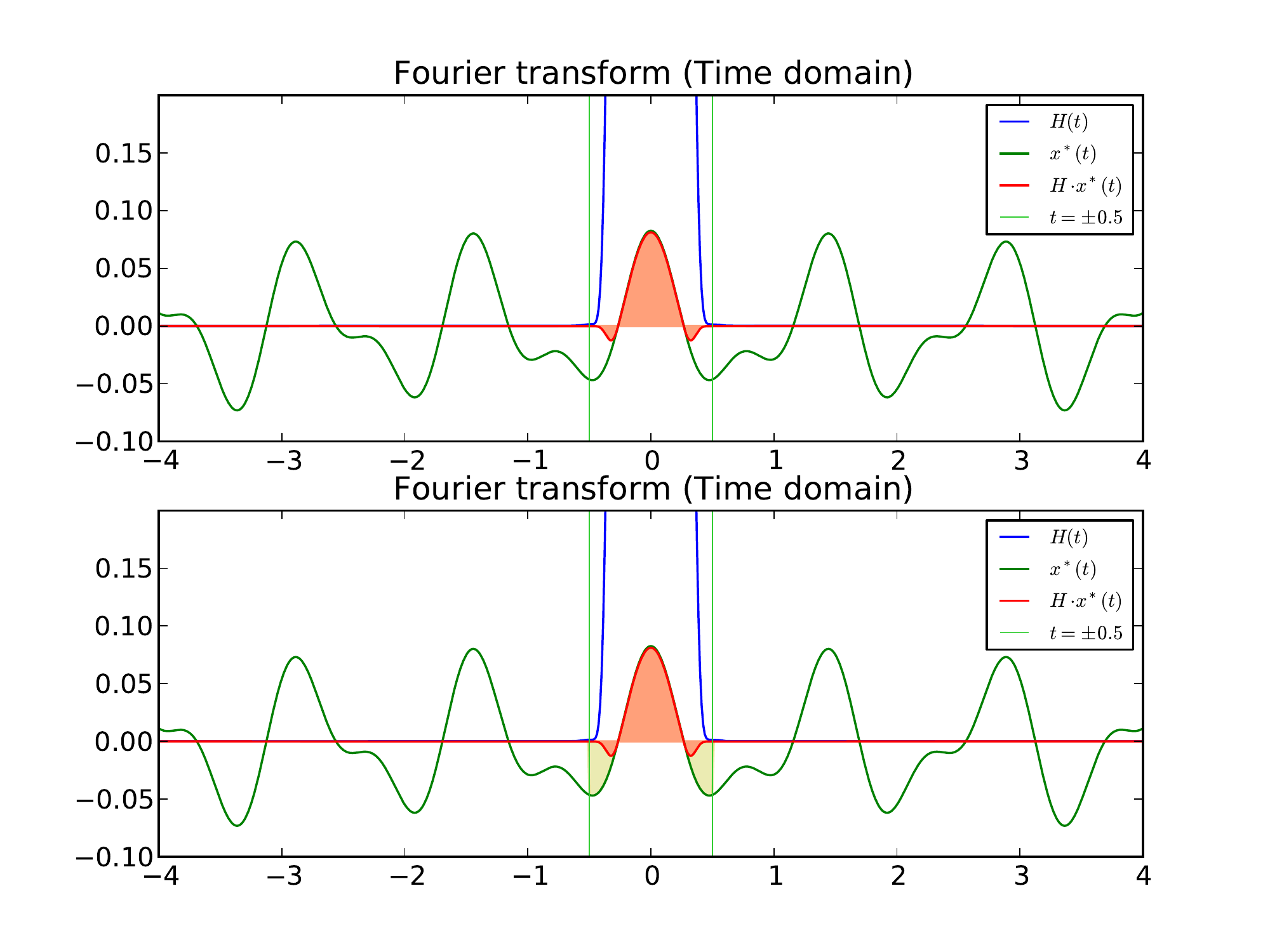}
    \caption{Property \RN{6} of filter function $H(t)$, the light green area represents $\RHS$(without scalar) of Property \RN{6} of filter $H$, the light red area represents $\LHS$ of Property \RN{6} of filter $H$, the light yellow area represents the difference. Property \RN{6} says the light yellow area is only a small constant fraction of the light green area. }
    
\end{figure}


\begin{remark}

To match $(H(t),\widehat{H}(f))$ on $[-1/2,1/2]$ with signal $x(t)$ on $[0,T]$, we will scale the time domain from $[-1/2,1/2]$ to $[-T/2,T/2]$ and shift it to $[0,T]$. For example, the rectangle function in Property \RN{5} and \RN{6} will be replaced by $\rect_T(t-T/2)$. For the parameters $s_0, s_1,s_3,\delta,\ell$ in the definition of $H$, we always treat them as numbers. We assume $T$ has seconds as unit and $\Delta_h$ has Hz as unit . For example, in time domain, the Property \RN{1} becomes that given $T>0$,
\begin{equation*}
H(t) \in [ 1 - \delta, 1] \text{~if~} |t-\frac{T}{2}| \leq (\frac{1}{2} -\frac{1}{s_1}) s_3 \cdot T
\end{equation*}
In frequency domain, the Property \RN{4} becomes
\begin{equation}\label{eq:def_of_delta_h}
\supp( \widehat{H}(f) ) \subseteq [ -\frac{\Delta_h}{2} , \frac{\Delta_h}{2} ], \text{~where~} \Delta_h = \frac{s_1\ell}{s_3 T}.
\end{equation}
\end{remark}

\begin{lemma}\label{lem:approximation_H_in_poly_time}
Let $H(t)$ denote the function defined in Definition \ref{def:def_of_filter_H}. For any $t \in [-\frac{1}{2},\frac{1}{2}]$, there exists an algorithm that takes $O( s_1+ \ell \log(s_1) + \log(1/\epsilon))$ time to output a value $\wt{H}(t)$ such that
\begin{equation*}
(1-\epsilon) H(t) \leq \wt{H}(t) \leq (1+\epsilon) H(t).
\end{equation*}
\end{lemma}
\begin{proof}
We will show that using a low degree polynomial with sufficiently large degree is able to approximate the sinc function. 
By definition of filter function, 
\begin{align*}
H(t) & =  s_0 \cdot \int_{-\infty}^{+\infty} \sinc( s_1 \tau  )^{\cdot \ell} \rect_{s_2}(t-\tau) \mathrm{d} \tau \\ 
& = s_0 \cdot \int_{t-\frac{s_2}{2}}^{t+\frac{s_2}{2}} ( \frac{\sin (\pi s_1 \tau )}{\pi s_1 \tau } )^{\ell} \mathrm{d} \tau \\
& = \frac{s_0}{\pi s_1} \int_{ (t-\frac{s_2}{2}) \pi s_1} ^{ (t+\frac{s_2}{2}) \pi s_1} (\frac{\sin(\tau)}{\tau})^\ell \mathrm{d} \tau \\
& = \frac{s_0}{\pi s_1} \int_{ (t-\frac{s_2}{2}) \pi s_1} ^{ (t+\frac{s_2}{2}) \pi s_1} \left( \sum_{i=0}^{\infty} (-1)^i \frac{\tau^{2i}}{ (2i+1)!} \right)^\ell \mathrm{d} \tau & \text{~by~Taylor~expansion} \\
& = \frac{s_0}{\pi s_1} \int_{ (t-\frac{s_2}{2}) \pi s_1} ^{ (t+\frac{s_2}{2}) \pi s_1} (A+B)^\ell \mathrm{d} \tau
\end{align*}
where the last step follows by setting $A=\sum_{i=0}^{d} (-1)^i \frac{\tau^{2i}}{(2i+1)!}$, and $B=\sum_{i=d+1}^{\infty} (-1)^i \frac{\tau^{2i}}{(2i+1)!}$.

Denote $I^+ = (t+\frac{s_2}{2})\pi s_1$ and $I^- = (t-\frac{s_2}{2})\pi s_1$. Because of $t\in [-1/2,1/2]$, then $\max(|I^+|, |I^-|) = O(s_1)$. The goal is to show that
\begin{equation*}
(1-\eps) \int_{I^-}^{I^+} (A+B)^\ell \mathrm{d} \tau  \leq \int_{I^-}^{I^+} A^\ell \mathrm{d} \tau \leq (1+\eps) \int_{I^-}^{I^+} (A+B)^\ell \mathrm{d} \tau 
\end{equation*} 
Let's prove an upper first,
\begin{align*}
& \int_{I^-}^{I^+} (A+B -B )^\ell \mathrm{d} \tau  \\
= & \int_{I^-}^{I^+} (A + B)^\ell \mathrm{d} \tau  + \sum_{j=1}^\ell \int_{I^-}^{I^+} {\ell \choose j}(A+B)^{\ell-j} (-B)^j \mathrm{d} \tau \\
\leq & \int_{I^-}^{I^+} (A + B)^\ell \mathrm{d} \tau  + \sum_{j=1}^\ell \int_{I^-}^{I^+} {\ell \choose j}|A+B|^{\ell-j} |B|^j \mathrm{d} \tau \\
\leq & \int_{I^-}^{I^+} (A+B)^\ell \mathrm{d} \tau  + \sum_{j=1}^\ell \int_{I^-}^{I^+} {\ell \choose j}|A+B|^{\ell-j} \mathrm{d} \tau  \cdot \max_{\tau \in [I^-, I^+]} |B|^j\\
\leq & \int_{I^-}^{I^+} (A+B)^\ell \mathrm{d} \tau  + \ell 2^{\ell} \cdot \max_{\tau \in [I^-, I^+]} |B| &\text{~by~}|H(t)|\leq 1 \text{~and~} |B|^j \leq |B|\\
\leq &   \int_{I^-}^{I^+} (A+B)^\ell \mathrm{d} \tau + \epsilon \cdot (s_1)^{-\Theta(\ell)} &\text{~by~Claim~\ref{cla:upper_bound_for_B_tail}}\\
\leq &  (1+\epsilon) \int_{I^-}^{I^+} (A+B)^\ell \mathrm{d} \tau &\text{~by~Claim~\ref{cla:min_cost_of_Ht}}
\end{align*}
where all the steps by setting $d\gtrsim s_1 +\ell \log(s_1) +\log(1/\epsilon)$. Similarly, we can prove a lower bound.
\end{proof}

\begin{claim}\label{cla:upper_bound_for_B_tail}
Let $B(\tau)= \sum_{i=d+1}^{+\infty} (-1) \frac{\tau^{2i}}{(2i+1)!}$, if $d\gtrsim \tau + \ell \log(s_1) + \log(1/\epsilon)$ then $|B(\tau)|\leq \epsilon (1/s_1)^{O(\ell)}$.
\end{claim}
\begin{proof}
We first show, for any $i \geq d+1$,
\begin{align*}
  ~& \frac{\tau^{2i}}{(2i+1)!} \\
 \leq~ & \frac{\tau^{2i}}{ e( (2i+1)/e)^{2i+1}} &\text{~by~} e(n/e)^n \leq n!\\
 \leq~ & 2^{-2i} & \text{~by~} i \gtrsim \tau \\
 \leq ~ & \eps (1/s_1)^{O(\ell)} & \text{~by~} i \gtrsim \ell \log(s_1) +\log(1/\epsilon)
\end{align*}
Second, we can show that
\begin{equation*}
\sum_{i=d+1}^{+\infty} (-1)\frac{\tau^{2i}}{ (2i+1)! } \lesssim \frac{\tau^{2(d+1)}}{ (2(d+1)+1)! } \leq \eps (1/s_1)^{O(\ell)}
\end{equation*}
Thus, we complete the proof.
\end{proof}



\begin{claim}\label{cla:min_cost_of_Ht}
$\min_{t\in [-1/2,1/2]} |H(t)|\geq  (s_1)^{-\Omega(\ell)}$.
\end{claim}
\begin{proof}

By the property of $H(t)$, 
\begin{equation*}
\min_{ \frac{1}{2}s_3 < |t| \leq \frac{1}{2} } H(t) = \min_{|t|\leq \frac{1}{2} } H(t)
\end{equation*}
Thus, it suffices to prove a lower bound on $H(t)$ for any $t$ such that $ \frac{1}{2}s_3 < |t| \leq \frac{1}{2}$. 
Because of symmetric property, we only need to prove a lower bound for one side.
Let's consider $t\in [ \frac{1}{2} s_3, 1/2]$,
\begin{align*}
H(t) & ~ \geq \frac{s_0}{\pi s_1} \int_{ (t-\frac{s_2}{2})\pi s_1 }^{ ( t + \frac{s_2}{2} ) \pi s_1 } ( \frac{\sin(\tau)}{\tau} )^{\ell} \mathrm{d} \tau \\
& ~ \geq \frac{s_0}{\pi s_1} \int_{ (t+\frac{s_2}{4})\pi s_1 }^{ ( t + \frac{s_2}{2} ) \pi s_1 } ( \frac{\sin(\tau)}{\tau} )^{\ell} \mathrm{d} \tau \\
& ~ \geq \frac{s_0}{\pi s_1} \cdot \Theta( ( t+\frac{s_2}{2}) s_1 ) \cdot \frac{1}{2} \cdot \pi \cdot \Theta( (t+ \frac{s_2}{2}) \pi s_1 )^{-\ell} \\
& ~\geq (s_1)^{-\Omega(\ell)}
\end{align*}
\end{proof}

\subsection{Analysis of filter function \texorpdfstring{$(G(t), \widehat{G}(f) )$}\quad}\label{sec:properties_of_G}

We construct $(G(t), \widehat{G}(f))$ in a similar way of $(H_1(t), \widehat{H_1}(f))$ by switching the time domain and the frequency domain of $(H_1(t), \widehat{H_1}(f))$ and modify the parameters for the permutation hashing $P_{\sigma,a,b}$. 

\begin{definition}
Given $B >1$, $\delta >0$, $\alpha>0$, we construct $G(t),\widehat{G}(f)$ by doing the following operations,
\begin{itemize}
\item $s_2 = \frac{\pi}{2B}$,
\item $s_1 = \frac{B}{\alpha\pi}$,
\item $\ell = l = \Theta( \log(k/\delta) )$.
\end{itemize}
Then $G(t), \widehat{G}(f)$ becomes
\begin{eqnarray*}
G(t) & = & b_0  \cdot (\rect_{s_1}(t))^{*l} \cdot \sinc(t s_2)\\
& = & b_0 \cdot (\rect_{ \frac{B}{(\alpha \pi)}} (t) )^{* l} \cdot  \sinc(t \frac{\pi}{2B}), \\
\widehat{G}(f) & = & b_0 \cdot ( \sinc(s_1 f) )^{\cdot l} * \rect_{s_2}(f) \\
& = & b_0 \cdot ( \sinc(\frac{B}{\alpha \pi} f) )^{\cdot l} * \rect_{\frac{\pi}{2B}}(f).
\end{eqnarray*}
where the scalar $b_0 = \Theta(s_1\sqrt{l}) = \Theta(B \sqrt{l}/\alpha)$ satisfying $\wh{G}(0)=1$.
\end{definition}

\restate{lem:property_of_filter_G}

\begin{proof}
The first five Properties follows from Lemma \ref{lem:property_of_filter_H} directly.

\end{proof}

\subsection{Parameters setting for filters}\label{sec:parameters_setting_for_filters}
\paragraph{One-cluster Recovery.}

\begin{figure}[t]
  \centering
    \includegraphics[width=0.8\textwidth]{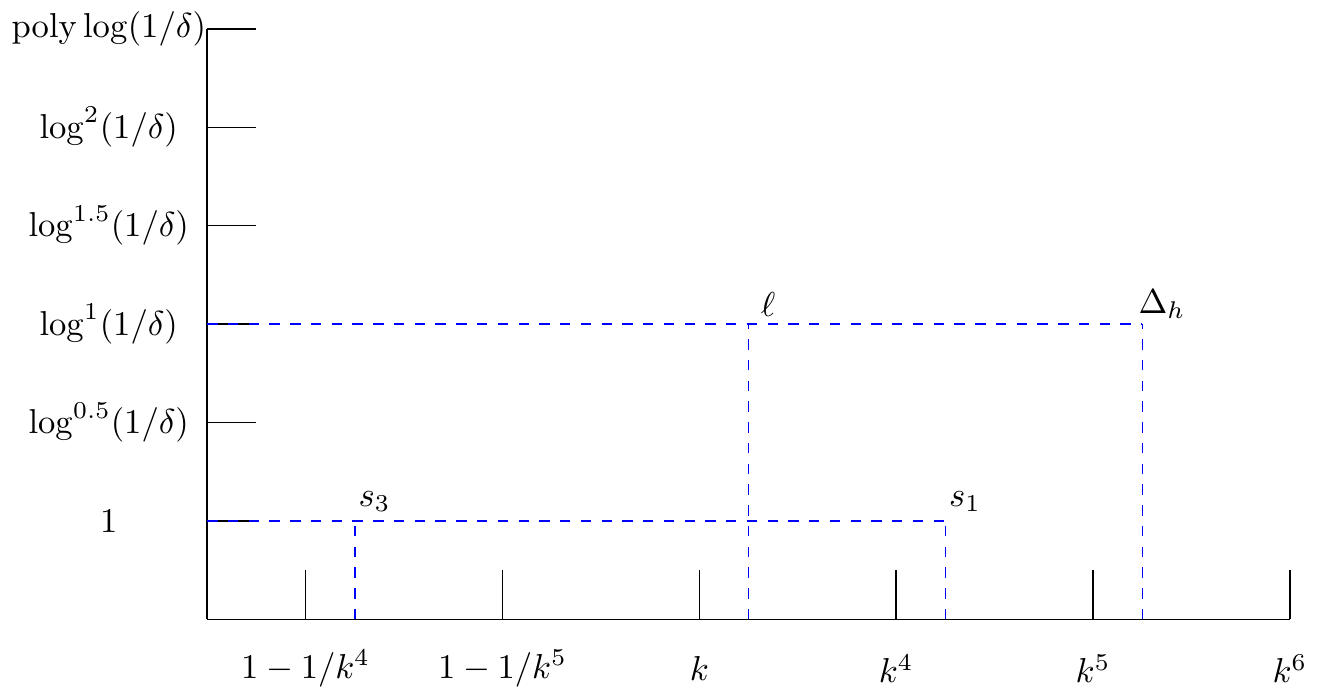}
    \caption{Parameters for $s_1, s_3$ and $\ell$.}\label{fig:parameters_setting1}
\end{figure}

In one-cluster, we donot need filter function $(G(t), \wh{G}(f))$.

In section \ref{sec:properties_of_H}, by Equation (\ref{eq:eq2_proof_of_property_6}) in the proof of Property \RN{6} of filter function  $(H(t), \widehat{H}(f))$, we need $\min (\frac{1}{1-s_3}, s_1) \geq \wt{O}(k^4)$.

In section \ref{sec:properties_of_H}, by Equation (\ref{eq:why_ell_at_least_k_log_k_over_delta}) in the proof of Property \RN{5} of filter function $(H(t),\wh{H}(f))$, we set $\ell \gtrsim k \log(k/\delta)$.

$\Delta_h$ is determined by the parameters of filter $(H(t), \widehat{H}(f))$ in Equation (\ref{eq:def_of_delta_h}): $\Delta_h \eqsim \frac{s_1 \ell}{s_3 T}$ in section \ref{sec:properties_of_H}.
Combining the setting of $s_1$, $s_3$ $\ell$, we should set $\Delta_h \geq \wt{O}(k^5 \log(1/\delta))/T$.

\paragraph{$k$-cluster Recovery.}
Note that in the $k$-cluster recovery, we need to use filter function $( G(t),\wh{G}(f) )$. We choose $l = \log(k/\delta)$, $\alpha \eqsim 1$, $B \eqsim k$ , and $D =l/\alpha$.

By proof of Property \RN{2} of $z$ in Lemma \ref{lem:full_proof_of_3_properties_true_for_z} from section \ref{sec:k_mountain_frequency_recovery}, we need $T(1-s_3) >\sigma B l$. By the same reason in one-cluster recovery, $1-s_3 \leq \frac{1}{\wt{O}(k^4)}$. Combining $T(1-s_3)> \sigma Bl$ and  $1-s_3 \leq \frac{1}{\wt{O}(k^4)}$, we obtain
\begin{equation}\label{eq:T_over_polyk_at_least_sigma_B_l}
\frac{T}{\wt{O}(k^4)} > \sigma B l 
\end{equation}
Because in our algorithm, we will sample $\sigma$ from $[\frac{1}{B \Delta_h}, \frac{2}{B\Delta_h}]$. Thus, plugging $\sigma = \Theta(\frac{1}{B \Delta_h})$ in Equation (\ref{eq:T_over_polyk_at_least_sigma_B_l}) we have
\begin{equation*}
\frac{T}{\wt{O}(k^4)} > \frac{l}{\Delta_h}
\end{equation*}
which implies another lower bound for $\Delta_h$,
\begin{equation*}
\Delta_h \geq \wt{O}(k^4) l /T
\end{equation*}
Combining the above bound with previous lower bound in one-cluster recovery, we get
\begin{equation*}
\Delta_h \geq \wt{O}(k^4 \log(1/\delta))/T + \wt{O}(k^5 \log(1/\delta ))/T =  \wt{O}(k^5 \log(1/\delta ))/T
\end{equation*}
For $s_1$ and $\ell$, we still choose the same setting as before, $s_1 \eqsim \wt{O}(k^4)$ and $\ell \eqsim O(k\log(k/\delta))$.

\subsection{Analysis of {\normalfont \texorpdfstring{\textsc{HashToBins}} \quad} }\label{sec:proof_hashtobins}
In this section, we explain the correctness of Procedure $\textsc{HashToBins}$ in Algorithm \ref{alg:locateksignal_locatekinner_hashtobins}. Before giving the proof of that algorithm, we show how to connect $\mathrm{CFT}$, $\mathrm{DTFT}$ and $\mathrm{DFT}$.

\begin{lemma}\label{lem:property_with_DTFT_and_DFT}
For any signal $W:\mathbb{R} \rightarrow \mathbb{C}$, let $A:\mathbb{Z} \rightarrow \mathbb{C}$ and $B: [n] \rightarrow \mathbb{C}$ be defined as follows:
\begin{eqnarray*}
A[i] = W(i), \forall i \in \mathbb{Z} \text{~and~} B[i] = \sum_{j\in \mathbb{Z} } A[i+jn], \forall i \in [n].
\end{eqnarray*} 
Then we consider the Fourier transform on $W, A, $ and $B$: 
\begin{align*}
\mathrm{CFT}\quad & \quad\widehat{W} : \mathbb{R} \rightarrow \mathbb{C}, \\
\mathrm{DTFT} \quad&\quad \widehat{A} : [0,1] \rightarrow \mathbb{C},\\
\mathrm{DFT}\quad & \quad\widehat{B} : [n] \rightarrow \mathbb{C}.
\end{align*}
We have:
\[ \forall f \in [0,1), \wh{A}(f) = \sum_{j \in \mathbb{Z}} \wh{W}(f+j); \quad \forall i \in [n], \wh{B}[i]=\sum_{j \in \mathbb{Z}}\wh{W}(i/n+j).\]
\end{lemma}
\begin{proof}
Recall that $\Comb_s(t)=\sum_{j \in \mathbb{Z}} \delta_{js}(t)$. First, we show $\widehat{A}(f) = \sum_{j \in \Z } e^{2\pi \i j f} A[j]$ equals to $\underset{j\in \Z}{\sum} \widehat{W}(f+j)$:
\begin{align}\label{eq:hat_A_is_sum_of_hat_W}
\widehat{A}(f) \notag 
 = \quad & \sum_{j \in \Z } e^{2\pi \i j f} W[j] &\text{~by~$A[j] = W(j)$} \notag \\
 = \quad & \int_{-\infty}^{+\infty} e^{2\pi\i fj } W(j) \cdot \Comb_1(j) \mathrm{d} j \notag \\
 = \quad & \wh{W \cdot \Comb_1}(f) \notag\\
 = \quad & ( \widehat{W} *  \widehat{\Comb}_1 )(f) \notag \\
 =\quad &\sum_{j\in \Z} \widehat{W}(f+j).
\end{align}

Next, we prove that $\forall i\in [n],\widehat{B}[i] = \widehat{A}(i/n)$,
\begin{align}\label{eq:hat_B_is_hat_A}
 \widehat{B}[i] \notag = \quad & \sum_{j=1}^n B[j] e^{ \frac{2\pi\i}{n} i j} & \text{~by~$\mathrm{DFT}$} \notag \\
 = \quad & \sum_{j=1}^n (\sum_{k \in \Z} A[j+kn] )e^{\frac{2\pi\i}{n} ij} & \text{~by~ $B[j] = \sum_{k\in \Z} A[j+kn]$} \notag \\
 = \quad & \sum_{j=1}^n \sum_{k\in \Z} A[j+kn] e^{\frac{2\pi\i}{n} i (j + kn) } & \text{~by~$e^{\frac{2\pi\i}{n} \cdot ikn} =1$} \notag \\
 = \quad & \sum_{j \in \Z} A[j]e^{ 2\pi\i j \frac{i}{n}} = \widehat{A}(i/n) & \text{~by~DTFT}.
\end{align}
Combining Equation (\ref{eq:hat_B_is_hat_A}) and Equation (\ref{eq:hat_A_is_sum_of_hat_W}), we obtain that $\widehat{B}[j] = \widehat{A}(j/n) = \sum_{i\in \Z}\widehat{W}(j/n + i)$ for all $j \in [n]$.
\end{proof}
\begin{claim}\label{cla:DFT_B_to_BD}
Let $u\in \C^B$ and $V\in \C^{BD}$ such that for any $j\in B$, $u[j] =\underset{i\in [D] }{\sum} V[j+(i-1)B]$. Then
\[ \widehat{u}[j ] = \widehat{V}[jD] , \forall j \in [B].\]
\end{claim}
\begin{proof}
We prove it through the definition of the Fourier transform:
\begin{align*}
\widehat{V}[jD]  = \quad & \sum_{i=1}^{BD} V[i] \cdot e^{ \frac{2\pi \i}{BD }\cdot i \cdot (jD)} &\text{~by~definition~of~DFT} \\
  = \quad & \sum_{i=1}^B \sum_{k=1}^D V[i+kB] e^{\frac{2\pi\i }{B}  \cdot (i+kB) \cdot j} &\text{~by~replacing~$i$~by~$i+kB$}\\
  = \quad & \sum_{i=1}^B e^{\frac{2\pi\i}{B} \cdot j \cdot i} \sum_{k=1}^D V[i+(k-1)B] &\text{~by~$e^{2\pi\i jk}=1$}  \\
    = \quad & \sum_{i=1}^B e^{\frac{2\pi\i}{B} \cdot j \cdot i} u[i] = \widehat{u}[j]& \text{~by~definition~of~DFT on $u$}
\end{align*}
\end{proof}

We use Definition \ref{def:permutation} and Lemma \ref{lem:permutation} to generalize Lemma \ref{lem:property_with_DTFT_and_DFT},
\begin{corollary}\label{cor:DTFT}
If for all $j\in [n]$, $B[j] = \underset{i\in\Z}{\sum} W\big( (j+in)\sigma -\sigma a \big)$, then $\forall j\in [n]$,
\begin{equation*}
\widehat{B}[j] = \sum_{i\in \Z} \widehat{W} \left( (\frac{j}{n} + i)/\sigma \right)\cdot \frac{1}{\sigma} e^{-2\pi \i (\frac{j}{n} + i) a}.
\end{equation*}
If for all $j\in [n]$, $B[j] = \underset{i\in\Z}{\sum} W\big( (j+in)\sigma -\sigma a \big)e^{-2\pi \i \sigma b(j+in)}$, then $\forall j\in [n]$,
\begin{equation*}
\widehat{B}[j] = \sum_{i\in \Z} \widehat{W} \left( (\frac{j}{n} + i)/\sigma +b\right)\cdot \frac{1}{\sigma} e^{-2\pi \i (\frac{j}{n} + i) a - 2\pi\i \sigma ab}.
\end{equation*}
\end{corollary}

\begin{remark}[Samples of $\textsc{HashToBins}$]
Procedure $\textsc{HashToBins}$ in Algorithm \ref{alg:locateksignal_locatekinner_hashtobins} takes $BD$ samples in $x(t)$:
\begin{equation*}
x(\sigma (1-a)), x(\sigma(2-a)), \cdots, x(\sigma(BD-a)).\end{equation*}
\end{remark}

\begin{figure}[t]
  \centering
    \includegraphics[width=1.0\textwidth]{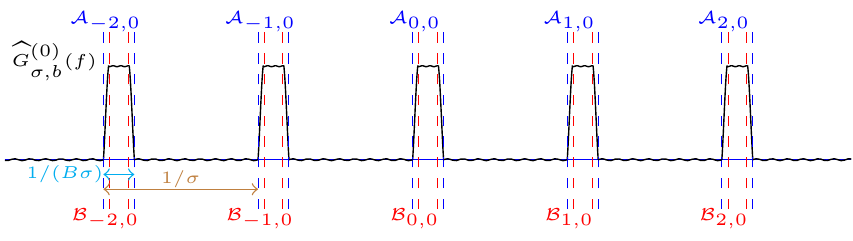}
    \includegraphics[width=1.0\textwidth]{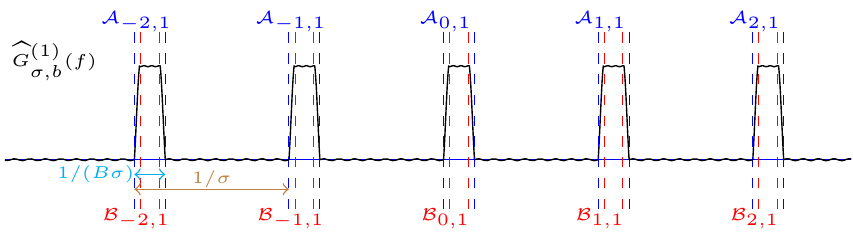}
    \caption{$\widehat{G}_{\sigma,b}^{(j)}(f)$ where the top one is $j=0$ and the bottom one is $j=1$, ${\cal A}_{i,j} = [ \frac{1}{\sigma} (2\pi( i + \frac{j}{B}) - \frac{2\pi}{2B} ), \frac{1}{\sigma} (2\pi( i + \frac{j}{B}) + \frac{2\pi}{2B} ) ]$, ${\cal B}_{i,j} = [ \frac{1}{\sigma} (2\pi( i + \frac{j}{B}) - \frac{2\pi(1-\alpha)}{2B} ), \frac{1}{\sigma} (2\pi( i + \frac{j}{B}) + \frac{2\pi(1-\alpha)}{2B} ) ]$}\label{fig:hatG_j_sigma_b}
\end{figure}

To analyze our algorithm, we use filter function $(G(t), \widehat{G}(f))$ and $\Comb_s(t) = \underset{j\in \Z}{\sum} \delta_{sj}(t)$ to define the discretization of $G$.
\begin{definition}\label{def:G_discretization}
Define the discretization of $G(t)$ and $\widehat{G}(f)$,
\begin{eqnarray*}
G^{\dis}(t)  & = & G(t) \cdot \Comb_s(t) \\
\widehat{G}^{\dis}(f) &=& \frac{1}{s} ( \widehat{G} * \Comb_{1/s})(f) \\
&=&   ( \widehat{G} * \Comb_{1})(f) \\
& = &  \bigl( {\vcenter{\hbox{\includegraphics[width=0.25\textwidth]{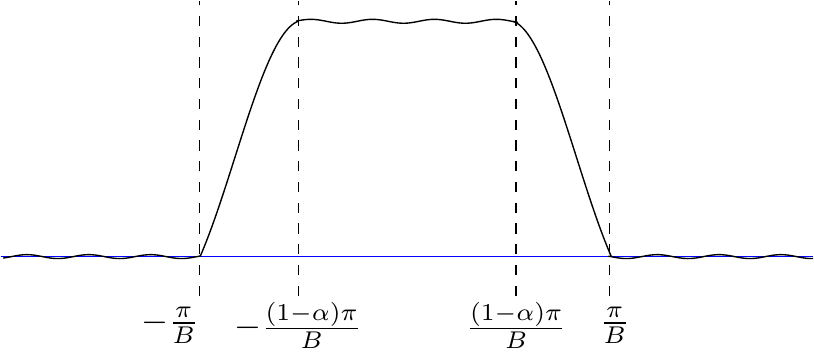}}}} \bigr) * {\vcenter{\hbox{\includegraphics[width=0.3\textwidth]{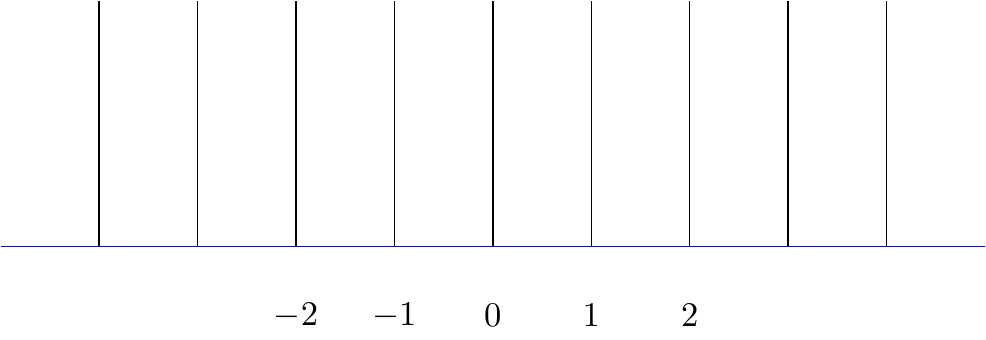}}}} 
\end{eqnarray*}
where $| \supp(G(t)) | = \frac{lB}{\pi \alpha} $, $D = \frac{l}{\pi\alpha}$, $s = | \supp(G(t)) | / (BD) = l /(\pi\alpha D) =1$.
\end{definition}

\restate{def:G_j_sigma_b}

\restate{lem:hashtobins}

\begin{proof}
Recall $B$ is the number of hash bins. $B\cdot D$ is the number of samples in time signal. 
Let $W(t) = x\cdot H(t)$, define vector $y \in \C^{BD}$, then $\forall j \in [BD]$, define
\begin{equation*}
y[j] = W(\sigma(j-a))e^{2\pi\i \sigma bj}
\end{equation*}

Recall $G(t)$ denote the $\rect^{*l}_{B/\alpha}(t) \cdot \sinc(t/B)$, then $|\supp(G(t))| = \frac{l B}{\alpha}$. Let vector $G'  \in \mathbb{C}^{BD}$  is the discretization of $G(t)$, where $G'[i] = G(i)$.  Then, $\forall j \in [B]$,
\begin{equation*}
u[j] = \sum_{i\in [D]} V[j + iB]
\end{equation*}
where $V[j] = y[j] \cdot G'[j]$ and $G'[j]$ is the value at the $j$th nonzero point of $G^{\dis}(t)$. Applying Claim \ref{cla:DFT_B_to_BD} with the definition of $u[j]$ and $V[j+iB]$, gives $\wh{u}[j] = \wh{V}[jD], \forall j \in [B]$.


Because of $u $ is the result of $ \textsc{HashToBins}(x\cdot H, P_{\sigma,a,b}, G)$ and $| \supp(G(t) ) | = BD$(choosing $D =l/\alpha$), then
\begin{equation*}
u[j] = \sum_{i\in \Z}W(\sigma(j+iB -a)) e^{-2\pi\i \sigma b(j+iB)} G(j+iB)
\end{equation*}

Then we define $G''(t) = G(t/\sigma+a) e^{-2\pi\i b\sigma(t/\sigma+a)}$ and $Y(t) =W(t) \cdot G''(t)$, then immediately, we have
\begin{equation*}
\widehat{G}''(f) = \sigma\widehat{G}(\sigma (f-b) ) e^{2\pi\i a \sigma f} \quad \text{~and~}\quad \widehat{Y}(f) = \widehat{W}(f) * \widehat{G}''(f)
\end{equation*}
Thus, we can rewrite $u[j]$ in the following sense,
\begin{align*}
\quad  &u[j] \\
= \quad &\sum_{i\in \Z}W(\sigma(j+iB -a)) e^{-2\pi\i \sigma b(j+iB)} G(j+iB) \\
=\quad &\sum_{i\in \Z}W(\sigma(j+iB -a)) G''(\sigma(j+iB-a)) & \text{~by~$G''(t) = G(t/\sigma+a) e^{-2\pi\i b\sigma(t/\sigma+a)}$} \\
=\quad &\sum_{i\in \Z}Y(\sigma(j+iB -a)) &\text{~by~$Y(t) =W(t) \cdot G''(t)$}
\end{align*}

Then
\begin{align*}
& \widehat{u}[j] \\
= \quad &  \sum_{i\in \Z} \widehat{Y}( (\frac{j}{B} + i)/\sigma ) \cdot \frac{1}{\sigma} \cdot e^{-2\pi \i (\frac{j}{B}+i) a} &\text{~by~Corollary~\ref{cor:DTFT}} \\
= \quad & \sum_{i\in \Z}\int_{-\infty}^{+\infty} \widehat{W}(s) \cdot \widehat{G}''( \frac{j/B+i}{\sigma} -s ) \cdot \frac{1}{\sigma} \cdot e^{-2\pi \i \cdot (j/B+i ) a} \mathrm{d} s&\text{~by~$\widehat{Y}(f) = \widehat{W}(f)* \widehat{G}''(f)$}\\
= \quad & \sum_{i\in \Z}\int_{-\infty}^{+\infty} \widehat{W}(s) \cdot \widehat{G}(  j/B+i -\sigma s -\sigma b) \cdot  e^{-2\pi \i \cdot ( -\sigma s) a}  \mathrm{d} s&\text{~by~$\widehat{G}''(f) = \sigma\widehat{G}(\sigma (f-b) ) e^{2\pi\i a \sigma f}$}\\
= \quad &\int_{-\infty}^{+\infty} \widehat{W}(s) \cdot \sum_{i\in \Z} \widehat{G} ( (j/B+i) - \sigma s -\sigma b ) e^{2\pi\i a \sigma s} \mathrm{d} s \\
= \quad & \int_{-\infty}^{+\infty} \widehat{W}(s) \cdot \widehat{G}^{\dis}( \frac{j}{B} -\sigma s -\sigma b ) e^{-2\pi\i a \sigma s} \mathrm{d} s &\text{~by~$\widehat{G}^{\dis}(f) = \sum_{i\in \Z} \widehat{G}(f+i) $}
\end{align*}
By definition \ref{def:G_discretization},
\begin{equation*}
\widehat{G}^{(j)}_{\sigma,b} = \widehat{G}^{\dis}(\frac{j}{B}-\sigma s -\sigma b) = \sum_{i\in \Z} \widehat{G}(i + \frac{j}{B} - \sigma s -\sigma b)
\end{equation*}

By definition of $\widehat{z}$, we have
\begin{equation*}
\widehat{z}(s) = \widehat{x\cdot H}(s) \cdot \widehat{G}^{(j)}(s) =  \widehat{W}(s) \cdot \widehat{G}^{(j)}(s)
\end{equation*}
Then $\widehat{u}[j]$ is the $(a\sigma)^{th}$ inverse Fourier coefficients of $\widehat{z}$, basically,
\begin{equation*}
\widehat{u}[j] = z_{a\sigma} = z(a\sigma)
\end{equation*}
Thus, we can conclude first computing vector $u\in \mathbb{C}^B$. Getting vector $\widehat{u} \in \mathbb{C}^B $ by using the Discrete Fourier transform $\widehat{u} = \mathrm{DFT}(u)$. This procedure allows us to sample from time domain to implicitly access the time signal's Fourier transform $\widehat{z}$. If $z$ is one-cluster in frequency domain, then apply one-cluster recovery algorithm.

\end{proof}



%



\section{Acknowledgments}
The authors would to like thank Aaron Sidford and David Woodruff for useful discussions.


\section{Algorithm}
This section lists the pseudocode of our algorithms.



\begin{algorithm}[ht]
\caption{}\label{alg:getempirical1enery_getlegal1sample}
\begin{algorithmic}[1]
\Procedure{$\textsc{GetEmpirical1Energy}$}{$z,T,\Delta$} --- Claim \ref{cla:get_empirical_1_energy}
	\State $R_{\est} \leftarrow (T \Delta )^2$
	\For {$i=1 \to R_{\est}$}
		\State Choose $\alpha_i \in [0,T]$ uniformly at random
		\State $z_{\emp} \leftarrow z_{\emp} + |z(\alpha_i)|^2$
	\EndFor
	\State $z_{\emp}\leftarrow \sqrt{ z_{\emp} / R_{\est}}$
	\State \Return $z_{\emp}$
\EndProcedure
\Procedure{$\textsc{GetLegal1Sample}$}{$z,\Delta,T,\beta,z_{\emp}$} --- Lemma \ref{lem:get_legal_1_sample}
	\State $R_{\repeats} \leftarrow (T \Delta )^3, S_{\heavy} \leftarrow \emptyset$
	\For {$i=1 \to R_{\repeats}$}
		\State Choose $\alpha_i \in [0,T]$ uniformly at random
		\If {$|z(\alpha_i)| \geq 0.5 \cdot z_{\emp} $}
			\State $S_{\heavy} \leftarrow S_{\heavy} \cup i$
		\EndIf		
	\EndFor
	\For { $i \in S_{\heavy}$}
		\State $w(i)\leftarrow |z(\alpha_i)|^2+|z(\alpha_i+\beta)|^2 $
	\EndFor
	\State $\alpha \leftarrow \alpha_i$ with probability $w(i)/\sum_{j\in S_{\heavy}} w(j)$ for $i \in S_{\heavy}$
	\State \Return $\alpha$
\EndProcedure
\end{algorithmic}
\end{algorithm}

\begin{algorithm}[!t]
\caption{ }\label{alg:locate1signal_locate1inner_frequencyrecovery1cluster}
\begin{algorithmic}[1]

\Procedure{$\textsc{Locate1Signal}$}{$z,T,F,\Delta,z_{\emp}$} --- Lemma \ref{lem:locate_1_signal}
	\State Set $t \eqsim \log(FT)$, $t' = t/4$, $D_{\max} \eqsim \log_{t'}(FT)$, $R_{\loc} \eqsim \log_{1/c}(tc)$, $L^{(1)}=2F$
	\For{$i\in [D_{\max}]$}
		\State $ l \eqsim 2F/(t')^{i-1}\Delta$, $s\eqsim c$, $\wh{\beta} = \frac{ts}{2 \Delta l}$
		\If {$ \wh{\beta} \gtrsim T/(T\Delta)^{3/2}$}
			\State {\bf break}
		\Else
			\State $L^{(i)} \leftarrow \textsc{Locate1Inner}$($z,\Delta,T,\wh{\beta},z_{\emp}, L^{(i-1)}$)
		\EndIf
	\EndFor
	\State \Return $L^{(i)}$
\EndProcedure
\Procedure{$\textsc{Locate1Inner}$}{$z,\Delta,T,\wh{\beta},z_{\emp},\wt{L}$}
	\State Let $v_{q}\leftarrow 0$ for $q\in [t]$ 
	\While{ $ r=1\to R_{\loc}$}
		\State Choose $\beta \in [\frac{1}{2}\wh{\beta}, \wh{\beta} ]$ uniformly at random
		\State $\gamma\leftarrow \textsc{GetLegal1Sample}(z,\Delta,T,\beta,z_{\emp}) $
			\For {$i \in [m]$}
				\State $s_i\in [ \beta( \wt{L} -\Delta l/2) ,  \beta( \wt{L} + \Delta l/2)  ] \cap \mathbb{Z}_+,\theta_{i} = \frac{1}{2\pi\sigma\beta} ( \phi(x(\gamma)/x(\gamma+\beta)) +2\pi s_i)$
				\State Let $\theta_{i}$ belong to $\text{region}(q)$
				\State Then add a vote to region($q$) and its two neighbors, i.e., region($q-1$) and region($q+1$)
			\EndFor
	\EndWhile	
	\State $q_j^* \leftarrow \{q | v_{q}>\frac{R_{\loc}}{2} \}$
	\State \Return $L \leftarrow \mathrm{center~of~region}( q_j^*)$
\EndProcedure
\Procedure{$\textsc{FrequencyRecovery1Cluster}$}{$z,T,F,\Delta$} --- Theorem \ref{thm:frequency_recovery_1_cluster}
	\State $z_{\emp} \leftarrow$ \textsc{GetEmpirical1Energy}($z,T,\Delta$)
	\For{$r=1\to O(  k)$} 
		\State $L_r \leftarrow \textsc{Locate1Signal}$($z,T,F,\Delta,z_{\emp}$)
	\EndFor
	\State \Return $L^*\leftarrow \underset{r\in [O( k)]}{\median} ~L_r$
\EndProcedure
\end{algorithmic}
\end{algorithm}

\begin{algorithm}[t]
\caption{Main algorithm for one-cluster recovery}\label{alg:main_1}
\begin{algorithmic}[1]
\Procedure{$\textsc{CFT1Culster}$}{$x,H,T,F$} --- Theorem \ref{thm:cft1cluster}
	\State $\wt{f_0}  \leftarrow\textsc{FrequencyRecovery1Cluster}$($x,H,T,F$)
	\State $\wt{x}\leftarrow \textsc{SignalRecovery1Cluster}$($\wt{f}_0, \poly(k) \Delta_h $)
	\State \Return $\wt{x}$
\EndProcedure

\Procedure{$\textsc{GenerateIntervals}$}{$d$}
\State $n \leftarrow y_0 \leftarrow i \leftarrow 0$, $m \leftarrow \Theta(d)$
	\While{$y_i \le 1 - \frac{9}{m^2}$}
		\State $y_{i+1}\leftarrow y_i + \frac{\sqrt{1-y_i^2}}{m}$,  $I_{n+1}\leftarrow[y_i,y_{i+1}], I_{n+2}\leftarrow[-y_{i+1},-y_i]$
		\State $i\leftarrow i+1, n\leftarrow n+2$
	\EndWhile
	\State  $I_{n+1}\leftarrow [y_i,1], I_{n+2}\leftarrow [-y_{i},-1]$, $n\leftarrow n+2$
	\State \Return $n,I$
\EndProcedure
\Procedure{\textsc{RobustPolynomialLearning}}{$x,d,T$}  --- Theorem \ref{thm:faster_poly_learning}
	\State $(n, I) \leftarrow \textsc{GenerateIntervals}$($d$)
	\For {$j=1\to n$}
		\State $w_j \leftarrow |I_j|/2$
		\State Choose $t_j$ from $I_j$ uniformly at random
		\State $z_j \leftarrow x(T \cdot \frac{t_j+1}{2}) $
	\EndFor 
	\State $\wt{A}_{j,i}\leftarrow t_j^i$, for each $(j,i)\in [n]\times \{0,1,\cdots,d\}$
	\State $\alpha \leftarrow$ \textsc{LinearRegressionW}($\wt{A},\wt{b}=z,w$)
	\State $Q(t) \leftarrow \sum_{i=0}^d \alpha_i t^i$
	\State \Return $\wt{Q}(t)=Q(T \cdot \frac{t+1}{2})$
\EndProcedure
\Procedure{\textsc{RobustPolynomialLearning$^{+}$}}{$x,d,T$} 
--- Theorem \ref{thm:accurate_poly_learning} --- a.k.a. \textsc{SignalRecovery1Cluster} 
\State $R \leftarrow \Theta(d)$
\State $(n, I) \leftarrow \textsc{GenerateIntervals}$($d$)
\State $w_j\leftarrow |I_j|/2$, for each $j\in [n]$
\For {$i=1\to R$}
	\State $Q_i \leftarrow \textsc{RobustPolynomialLearning}$($x,d,T$)
\EndFor
\State Choose $t_j$ from $I_j$ uniformly at random, for each $j\in [n]$
\For {$i=1\to R$}
	\State $Q_i(t_1), Q_i(t_2), \cdots, Q_i(t_n)\leftarrow$ \textsc{MultipointEvaluation}($Q_i,\{t_1,t_2,\cdots,t_n\}$)
\EndFor

\State $\wt{Q}_j \leftarrow \underset{i\in [R] }{\median} ~Q_i(t_j)$, for each $j \in[n]$
\State $\wt{A}_{j,i}\leftarrow t_j^i$, for each $(j,i)\in [n]\times \{0,1,\cdots,d\}$
\State $\alpha \leftarrow$ \textsc{LinearRegressionW}($\wt{A},\wt{b}=\wt{Q},w$)
\State \Return $Q(t) \leftarrow \sum_{i=0}^d \alpha_i t^i$

\EndProcedure
\end{algorithmic}
\end{algorithm}

\begin{algorithm}[t]
\caption{}\label{alg:locateksignal_locatekinner_hashtobins}
\begin{algorithmic}[1]
\Procedure{$\textsc{LocateKSignal}$}{$x,H,G,T,\Delta,\sigma,b,z_{\emp}$} --- Clain~\ref{cla:locate_k_signal}
	\State Set $t \eqsim \log(FT)$, $t' = t/4$, $D_{\max} \eqsim \log_{t'}(FT)$, $R_{\loc} \eqsim \log_{1/c}(tc)$, $L^{(1)}=2F$
	\For{$i\in [D_{\max}]$}
		\State $\Delta l \eqsim 2F/(t')^{i-1}$, $s\eqsim c$, $\wh{\beta} = \frac{ts}{2\sigma \Delta l}$
		\If {$\sigma \wh{\beta} \gtrsim T/(T\Delta)^{3/2}$}
			\State {\bf break}
		\Else
			\State $L^{(i)} \leftarrow \textsc{LocateKInner}$($x,H,G,T,\Delta,\sigma,b,z_{\emp}\wh{\beta},U, L^{(i-1)}$)
		\EndIf
	\EndFor
	\State \Return $L^{(i)}$
\EndProcedure
\Procedure{$\textsc{LocateKInner}$}{$x,H,G,T,\Delta,\sigma,b,z_{\emp}\wh{\beta},U,\wt{L}$}
	\State Let $v_{j,q}\leftarrow 0$ for $(j,q)\in [B]\times [t]$
	\For {$r=1 \to R_{\loc}$}
		\State Choose $\beta \in [\frac{1}{2}\wh{\beta}, \wh{\beta} ]$ uniformly at random
		\State $\wh{u}, \wh{u}' \leftarrow$ \textsc{GetLegalKSample}($x,H,G, T,\Delta,\sigma,\beta, z_{\emp}$)
		\For {$j\in [B]$}
			\For {$i \in [m]$}
				\State $\theta_{j,i} = \frac{1}{2\pi\sigma\beta} ( \phi(\widehat{u}[j]/\widehat{u'}[j]) +2\pi s_i), s_i\in [ \sigma\beta( \wt{L}_j -\Delta l/2) ,  \sigma\beta( \wt{L}_j + \Delta l/2)  ] \cap \mathbb{Z}_+$

				\State $f_{j,i} = \theta_{j,i} + b \pmod F$
				\State suppose $f_{j,i}$ belongs to $\text{region}(j,q)$, 
				\State add a vote to both region($j,q$) and two neighbors nearby that region, e.g. region($j,q-1$) and region($j,q+1$)
			\EndFor
		\EndFor
	\EndFor
	\For {$j \in [B]$}
		\State $q_j^* \leftarrow \{q | v_{j,q}>\frac{R_{\loc}}{2} \}$
		\State $L_j \leftarrow \mathrm{center~of~region}(j, q_j^*)$
	\EndFor
	\State \Return $L$
\EndProcedure
\Procedure{$\textsc{HashToBins}$}{$x,H,G,P_{\sigma,a,b}$} --- Lemma \ref{lem:hashtobins}
	\State Compute $u[j] = \sum_{i\in D} v[j+iB]$
	\State $\wh{u} \leftarrow \mathrm{FFT}(u)$
	\State \Return $\wh{u}$ 
\EndProcedure
\end{algorithmic}
\end{algorithm}

\begin{algorithm}[ht]
\caption{}\label{alg:getempiricalkenergy_getlegalksample_onestage}
\begin{algorithmic}[1]
\Procedure{$\textsc{GetEmpiricalKEnergy}$}{$x,H,G,T,\Delta,\sigma,b$} --- Claim \ref{cla:get_empirical_k_energy}
	\State $R_{\est} \leftarrow (T \Delta )^2$
	\For {$i=1 \to R_{\est}$}
		\State Choose $\alpha \in [0,T]$ uniformly at random
		\State $\wh{u} \leftarrow \textsc{HashToBins}$($x,H,G,P_{\sigma,\alpha,b}$)
		\For{$j=1\to B$}
			\State $z_{\emp}^j \leftarrow z_{\emp}^j + |\wh{u}_j|^2$
		\EndFor
		
	\EndFor
	\For {$j=1 \to B$}
		\State $z_{\emp}^j\leftarrow \sqrt{ z_{\emp}^j / R_{\est}}$
	\EndFor
	\State \Return $z_{\emp}$.
\EndProcedure
\Procedure{$\textsc{GetLegalKSample}$}{$x,H, G,T,\Delta,\beta,z_{\emp}$} --- Lemma \ref{lem:get_legal_k_sample}
	\State $R_{\repeats} \leftarrow (T \Delta )^3$.
	\State $S_{\heavy}^j \leftarrow \emptyset, \forall j \in [B]$
	\For {$i=1 \to R_{\repeats}$}
		\State Choose $\alpha \in [0,T]$ uniformly at random
		\State $\wh{u}^{i} \leftarrow \textsc{HashToBins}$($x,H,G,P_{\sigma,\alpha,b}$)
		\State $\wh{u}^{'i} \leftarrow \textsc{HashToBins}$($x,H,G,P_{\sigma,\alpha+\beta,b}$)
		\For {$j=1\to B$}
			\If {$| \wh{u}^{i}_j | \geq 0.5 \cdot z_{\emp}^j $}
				\State $S_{\heavy,j} \leftarrow S_{\heavy}^j \cup i$
			\EndIf		
		\EndFor
	\EndFor
	\For{$j=1 \to B$}
		\For { $i \in S_{\heavy}^j$}
			\State $w(i)\leftarrow |\wh{u}_j^i|^2+|{\wh{u}}_j^{'i}|^2 $
		\EndFor
		\State $(\wh{v}_j, \wh{v}'_j) \leftarrow (\wh{u}_j^i, \wh{u}_j^{'i})$ with probability $w(i)/\sum_{i'\in S_{\heavy}^j} w(i')$ for $i \in S_{\heavy}^j$
	\EndFor
	\State \Return $\wh{v}, \wh{v}'\in \mathbb{C}^{B}$
\EndProcedure
\Procedure{$\textsc{OneStage}$}{$x,H,G,\sigma,b$} --- Lemma \ref{lem:one-round-full}
	\State $z_{\emp} \leftarrow $ \textsc{GetEmpiricalKEnergy}($x,H,G,T,\Delta,\sigma,b$)
	\State $L \leftarrow \textsc{LocateKSignal}$($x,H,G,T,\Delta,\sigma,b,z_{\emp}$)
\EndProcedure
\end{algorithmic}
\end{algorithm}

\begin{algorithm}[t]
\caption{Main algorithm for $k$-cluster recovery}\label{alg:main_k}
\begin{algorithmic}[1]
\Procedure{$\textsc{CFTKCluster}$}{$x,H,G,T,F$}
	\State $ \{\wt{f}_1,\cdots,\widetilde{f}_l\} \leftarrow$ \textsc{FrequencyRecoveryKCluster($x,H,G,T,F$)} 
	\State $\wt{x} \leftarrow $ \textsc{SignalRecoveryKCluster$^+$}$(\wt{f}_1,\cdots,\widetilde{f}_l,\Delta=\polydelta,T)$
	\State \Return $\wt{x}$ as our hypothesis
\EndProcedure
\Procedure{$\textsc{FrequencyRecoveryKCluster}$}{$x,H,G$} --- Theorem \ref{thm:frequency_recovery_k_cluster}
	\For {$r\in [R]$}
		\State {Choose $\sigma \in [\frac{1}{B\Delta_h}, \frac{2}{B \Delta_h}]$ uniformly at random}
		\State {Choose $b \in [0, \frac{ 2\pi\lfloor F/\Delta_h \rfloor] }{ (\sigma B) } ]$ uniformly at random}
		\State $L_r \leftarrow\textsc{OneStage}$($x,H,G,\sigma,b$)
	\EndFor
	\State $L^* \leftarrow \textsc{MergedStages}(L_1,L_2, \cdots, L_R)$
\EndProcedure
\Procedure{$\textsc{SignalRecoveryKCluster}$}{$\wt{f}_1,\cdots,\widetilde{f}_l,\Delta,T$}
	\State  $d \leftarrow 5 \pi ( (\Delta T)^{1.5} + k^3 \log k + k \log 1/\delta)$
	\State  $m \leftarrow O( (kd)^{C_3} \cdot \log^{C_3} d )$ for a constant $C_3=5$
	\For {$j= 1\to m$}
		\State {Sample $t_j$ from $[0,T]$ uniformly at random} 
		\State $\wt{A}_{j, i_1 \cdot l + i_2} \leftarrow t_j^{i_1} \cdot e^{2 \pi \i \wt{f}_{i_2} t_j}$ for each $(i_1,i_2) \in \{0,\cdots,d\} \times [l]$
		\State  $\wt{b}_j \leftarrow x(t_j)$
	\EndFor
	\State $\alpha \leftarrow $ \textsc{LinearRegression}($\wt{A},\wt{b}$)
	\State \Return $\wt{x}(t) \leftarrow \overset{d}{ \underset{i_1=0}{\sum}}  \overset{l}{\underset{i_2=1}{\sum}} \alpha_{i_1\cdot l + i_2} t^{i_1} \cdot e^{2 \pi \i \wt{f}_{i_2} t}$
\EndProcedure
\Procedure{\textsc{SignalRecoveryKCluster$^+$}}{$\wt{f}_1,\cdots,\widetilde{f}_l,\Delta,T$} --- Theorem \ref{lem:signal_recovery_k_cluster}
	\State $R\leftarrow \Theta(k)$
	\State  $d \leftarrow 5 \pi ( (\Delta T)^{1.5} + k^3 \log k + k \log 1/\delta)$
	\State  $m \leftarrow O( (kd)^{C_3} \cdot \log^{C_3} d )$ for a constant $C_3=5$ 
	\For{$i=1\to R$}
		\State $\wt{x}_i(t) \leftarrow \textsc{SignalRecoveryKCluster}$($\wt{f}_1,\cdots,\widetilde{f}_l,\Delta,T$)
	\EndFor
	\For {$j= 1\to m$}
		\State {Sample $t_j$ from $[0,T]$ uniformly at random} 
		\State $\wt{A}_{j, i_1 \cdot l + i_2} \leftarrow t_j^{i_1} \cdot e^{2 \pi \i \wt{f}_{i_2} t_j}$ for each $(i_1,i_2) \in \{0,\cdots,d\} \times [l]$
		\State  $\wt{b}_j \leftarrow \underset{i\in [R]}{\median} ~ \wt{x}_i(t_j)$
	\EndFor
	\State $\alpha \leftarrow$ \textsc{LinearRegression}($\wt{A},\wt{b}$)
	\State \Return $\wt{x}(t) \leftarrow \overset{d}{ \underset{i_1=0}{\sum}}  \overset{l}{\underset{i_2=1}{\sum}} \alpha_{i_1\cdot l + i_2} t^{i_1} \cdot e^{2 \pi \i \wt{f}_{i_2} t}$
\EndProcedure
\end{algorithmic}
\end{algorithm}

\end{document}